\documentclass[twocolumns,11pt,final]{IEEEtran}
\usepackage{graphics}
\usepackage{amsmath}
\usepackage{times}
\usepackage{amsmath,dsfont}
\usepackage{amssymb,amsthm,bm}
\usepackage{epsfig,verbatim}
\usepackage[final]{hyperref}
\usepackage{subfig}

\newtheorem{definition}{Definition}
\newtheorem{theorem}{Theorem}
\newtheorem*{theoremB}{Theorem}
\newtheorem{corollary}{Corollary}
\newtheorem{proposition}{Proposition}
\newtheorem{lemma}{Lemma}
\newtheorem*{lemmaB}{Lemma}
\newtheorem{remark}{Comment}

\newcommand{\Tgood}{\mbox{Tx}_1}
\newcommand{\Tbad}{\mbox{Tx}_2}
\newcommand{\Rgood}{\mbox{Rx}_1}
\newcommand{\Rbad}{\mbox{Rx}_2}
\newcommand{\E}{\mathds{E}}
\newcommand{\cov}{\mbox{cov}}
\newcommand{\C}{\mathds{C}}
\newcommand{\tmR}{\tilde{\mathcal{R}}}
\newcommand{\tmC}{\tilde{\mathcal{C}}}
\newcommand{\Amat}{\mathds{A}}
\newcommand{\Bmat}{\mathds{B}}
\newcommand{\Hmat}{\mathds{H}}

\newcommand{\Imat}{\mathds{I}}
\newcommand{\Tmat}{\mathds{T}}
\newcommand{\dsE}{\mathds{E}}
\newcommand{\tvec}{\mathbf{t}}
\newcommand{\Yvec}{\mathbf{Y}}
\newcommand{\yvec}{\mathbf{y}}

\newcommand{\Xvec}{\mathbf{X}}
\newcommand{\xvec}{\mathbf{x}}
\newcommand{\Zvec}{\mathbf{Z}}
\newcommand{\zvec}{\mathbf{z}}

\newcommand{\Svec}{\mathbf{S}}
\newcommand{\uavec}{\underline{\mathbf{a}}}
\newcommand{\avec}{\mathbf{a}}

\newcommand{\setC}{\mathfrak{C}}

\newcommand{\setR}{\mathfrak{R}}
\newcommand{\Rset}{\mathfrak{R}}
\newcommand{\hvec}{\tilde{h}}
\newcommand{\tX}{\tilde{{X}}}
\newcommand{\tY}{\tilde{Y}}
\newcommand{\tx}{\tilde{{x}}}
\renewcommand{\th}{\tilde{h}}
\newcommand{\utH}{\underline{\tilde{H}}}
\newcommand{\uth}{\underline{\tilde{h}}}
\newcommand{\hv}{\mathbf{h}}
\newcommand{\Hv}{\mathbf{H}}
\newcommand{\CN}{\mathcal{CN}}
\newcommand{\N}{\mathcal{N}}
\newcommand{\R}{\mathcal{R}}
\newcommand{\M}{\mathcal{M}}
\newcommand{\mR}{\mathcal{R}}
\newcommand{\mC}{\mathcal{C}}

\newcommand{\mT}{\mathcal{T}}
\newcommand{\mS}{\mathcal{S}}
\newcommand{\mD}{\mathcal{D}}
\newcommand{\hY}{\hat{Y}}
\newcommand{\dsQ}{\mathds{Q}}
\newcommand{\dsI}{\mathds{I}}
\newcommand{\tH}{\tilde{H}}

\newcommand{\tU}{\tilde{U}}
\newcommand{\utU}{\underline{\tilde{U}}}

\renewcommand{\P}{\makebox{P}}
\newcommand{\Real}{\mathfrak{Re}}

\long\def\symbolfootnote[#1]#2{\begingroup\def\thefootnote{\fnsymbol{footnote}}\footnote[#1]{#2}\endgroup}
\newcommand{\tend}{\hfill$\blacksquare$}

\IEEEoverridecommandlockouts

\title{Capacity Theorems for the Fading Interference Channel with a Relay and Feedback Links}
\author {Daniel Zahavi and Ron Dabora\\Department of  Electrical and Computer Engineering \\Ben-Gurion University, Israel
\thanks{Email: {\tt zahavida@post.bgu.ac.il, ron@ee.bgu.ac.il}. This work was partially supported by the European Commission's Marie Curie IRG Fellowship PIRG05-GA-2009-246657  under the Seventh Framework Programme.
Parts of this work were presented at the International Symposium on Information Theory (ISIT), July 2011, St. Petersburg, Russia.}
}

\begin{document}
\maketitle
\begin{picture}(0,0)
\put(-10,180){\tt\small Accepted to the IEEE Transactions on
Information Theory,  March 2012.}
\end{picture}

\begin{abstract}
Handling interference is one of the main challenges in the design of wireless networks. One of the key approaches to interference management is node
cooperation, which can be classified into two main types: relaying and feedback. In this work we consider simultaneous application of both cooperation
types in the presence of interference. We obtain exact characterization of the capacity regions for Rayleigh fading and phase fading interference channels with
a relay and with feedback links, in the strong and very strong interference regimes. Four feedback configurations are considered: (1) feedback from {\em both receivers}
to the relay, (2) feedback from {\em each receiver} to the relay and to one of the transmitters (either corresponding or opposite), (3) feedback from {\em one of the receivers}
to the relay, (4) feedback from {\em one of the receivers} to the relay and to one of the transmitters. Our results show that there is a strong motivation for incorporating relaying and feedback into wireless networks.
\end{abstract}

\section{Introduction}
Communication in the presence of interference is one of the main areas of research in information theory. The most basic network in which there is interference is the interference
channel (IC), introduced by Shannon in \cite{Shannon:61}. The IC consists of two transmitter-receiver pairs, Tx$_k$-Rx$_k$, $k\in\{1,2\}$, sharing the same physical channel.
The very strong interference (VSI) regime was first characterized for ICs by Carleial in \cite{Carleial:75}. When VSI occurs in ICs, each receiver can decode the interference
by treating its own signal as noise, without limiting the rate of the other pair. Thus, each pair can communicate at a rate equal to its point-to-point (PtP) interference-free capacity.
A weaker notion called strong interference (SI) was introduced by Sato in \cite{Sato:81}. When SI occurs in ICs, each receiver can decode both messages without reducing the
capacity region of the IC. In \cite{Sato:81} Sato showed that in such a scenario the capacity region of the IC is given by the intersection of the capacity regions of two
multiple-access channels (MACs) -- derived from the IC. The capacity region of the scenario where both messages are required by both receivers was first derived by Ahlswede in \cite{Ahlswede:74}.

One of the key approaches to interference management in wireless networks is relaying. The relay channel was first introduced by van der Meulen in \cite{venderMeulen:68} and it consists of
three nodes -- a transmitter, a receiver, and a relay, which assists the communication between the transmitter and the receiver. In \cite{CoverElGamal:IT79} Cover and El Gamal derived an
achievable rate for the relay channel by using a superposition block-Markov codebook and by decoding the source message at the relay. The relay then sends a message that assists the decoder
resolve the uncertainty about the source message. This scheme is called \textit{decode-and-forward} (DF). Another fundamental scheme introduced in \cite{CoverElGamal:IT79} is based on compression
at the relay. This scheme is commonly referred to as \textit{compress-and-forward} (CF). In addition, Cover and El Gamal provided an outer bound on the capacity of a general relay channel, but the exact capacity remains unknown.
An important contribution to the study of relay networks is the work of Kramer et al. in \cite{Kramer:05}. Kramer et al. obtained capacity theorems as well as achievable rate regions for different
relay networks by using the DF and CF strategies. In \cite{Kramer:05}, capacity results were presented for several relay networks for phase fading and Rayleigh fading channel models.

The classic relay channel of \cite{CoverElGamal:IT79} can be extended by adding a second source node, such that (s.t.) the relay assists the
communications from both sources to the (single) destination. This model is called the multiple-access relay channel (MARC). Some capacity results as well as inner and
outer bounds for the white Gaussian MARC were derived by Kramer et al. in \cite{Kramer:00}. The capacity region of the phase fading MARC was characterized in \cite{Kramer:05}.
Sankaranarayanan et al. presented outer bounds on the capacity region as well as achievable rate regions for the MARC in \cite{Sankaranarayanan:04}. The sum-capacity of the degraded
Gaussian MARC\footnote{A K-user Gaussian MARC is said to be degraded if, given the transmitted signal at the relay, the multiaccess signal received at the destination is a noisier version of
the multiaccess signal received at the relay.} was studied by Sankar in \cite{Sankar:09}. In \cite{Sankar:09} it was shown that while in the relay channel the degradedness assumption
simplified the cut-set bound to coincide with the DF achievable rate region, in the MARC this is not the case.
The MARC model can be generalized by considering multiple relays. The relay nodes are said to be parallel if there is no direct link between them, while all source-relay, relay-destination
and source-destination links exist. The parallel Gaussian MARC, with the relay nodes using the \textit{amplify-and-forward}\footnote{In \textit{amplify-and-forward} the relay simply transmits a scaled
version of its receives signal.} (AF) strategy, was studied by del Coso et al. in \cite{DelCoso:071}.

The MARC can be further extended by adding a second destination node s.t. each transmitter communicates only with a single destination. This gives rise
to the interference channel with a relay (ICR) which consists of five nodes. This channel was first studied by Sahin and Erkip \cite{Sahin:07} and has gained considerable
interest in the past few years. Inner bounds as well as outer bounds on the capacity region were derived for the ICR, see \cite{Sahin:09}, \cite{Maric:08}, \cite{Maric:09} and
\cite{Lasaulce:09_no1} and the references therein. One of the critical aspects in the study of ICRs is to determine what is the best strategy for the relay, since when assisting one
receiver the relay may degrade the performance of the other receiver. Moreover, in some situations the optimal relay strategy would be to forward interference rather than desired information
\cite{Maric:08}.
Thus, there might not be one scheme which increases the achievable rates for both
pairs simultaneously. In \cite{Stridharan:08} it was shown that when the relay is {\em cognitive} then it is able to
assist both pairs simultaneously by simultaneously zero-forcing the interference at each receiver. This assistance was shown to be optimal from the degrees-of-freedom (DoF) perspective for a large range
of channel coefficients. The capacity region of fading ICRs for a non-degraded scenario with a {\em causal} relay and finite
signal-to-noise ratios on all links, was first characterized in \cite{Dabora:101} and \cite{Dabora:10}.  In these works it was shown that in some situations the best strategy for the relay is DF and
 that the relay can optimally assist both receivers simultaneously, from the capacity perspective.
 Lastly, {\em global, instantaneous} CSI was considered in \cite{TianYener:2010}. In the work \cite{TianYener:2010}, fading ICRs with an ``on-and-off" relay were studied.
 Under the assumption of using ``asynchronous relaying" (i.e., the codebooks of the sources and of the relay are mutually independent) and with the assumption that  the fading coefficient
 equals zero with  a positive probability, \cite{TianYener:2010} obtained an achievable rate region.

Another tool for handling interference in wireless networks is feedback from receiving nodes to transmitting nodes. Feedback allows the nodes to coordinate their transmissions and thereby sometimes
helps in achieving higher rates compared to those achieved without coordination. In \cite{Shannon:56} Shannon showed that
feedback does not increase
the capacity of memoryless PtP channels. However, in \cite{Wolf:75} Gaarder and Wolf showed that in a memoryless MAC, if both transmitters have feedback from the receiver, they can cooperate to
increase the capacity region. This was the first time it was shown that feedback increases the capacity region of a memoryless channel. In \cite{CoverElGamal:IT79} Cover and El Gamal showed that
the cut-set bound for the relay channel is achieved with DF when feedback is available at the relay. In such a scenario feedback to the transmitter does not provide further improvement onto
feedback to the relay. Additional results on the achievable rates in the relay channel with receiver-transmitter feedback were obtained in \cite{Bross:09}. For the MARC with feedback from the relay
to the sources, Hou et al. derived an outer bound on the capacity region as well as achievable rate regions in \cite{Kramer:09}. In \cite{Kramer:09} feedback was used to allow each source to decode the
message of the other source, thereby the transmitters could cooperate and resolve the uncertainty at the receiver. The MARC with generalized feedback (MARC-GF) was studied by Ho et al. in \cite{Ho:08}.
The MARC-GF models cellular networks in which all the mobile stations can listen to the ongoing transmissions through the channel.

Feedback was also studied for ICs.
In \cite{Kramer:02} it was shown that for interference channels at SI, the capacity region is enlarged if each transmitter receives feedback from the receiver to which it is sending messages.
The sum-capacity of symmetric deterministic ICs with infinite-capacity feedback links from the receivers to the transmitters, was studied by Sahai et al. in \cite{Sahai:09}. In \cite{Sahai:09} it was shown that having
a single feedback link from one of the receivers to its own transmitter results in the same sum-capacity as having a total of four feedback links - from both receivers to both transmitters.
\cite{Sahai:09} also considered a practical feedback configuration for a TDD based system, where the forward and the feedback channels are symmetric and time-shared and it was shown that in such a
scenario, feedback does not increase the sum capacity of the IC in the SI regime. 
    In \cite{Cadambe:08}  Cadambe and Jafar provided a tight characterization of the generalized degrees-of-freedom (GDoF) for ICs with feedback for values of
    $\alpha \triangleq \frac{\log(INR)}{\log{SNR}} \ge \frac{2}{3}$.  It was observed in \cite{Cadambe:08} that feedback leads to an unbounded capacity gain in the very strong
    interference regime ($\alpha \ge 2$).  In \cite{Tse:09} the capacity region of the Gaussian IC with feedback was characterized to within 2 bits/symbol/Hz, and the exact
    GDoF was characterized for all values of $\alpha$.  In particular, it was shown in \cite{Tse:09} that feedback provides a capacity gain that
    increases with the SNR to infinity also in the  weak interference regime ($0\le \alpha \le \frac{2}{3}$), in addition to the case $\alpha \ge 2$.
    In \cite{Tuni:07} an achievable rate region for ICs with generalized feedback was derived. In this scenario, each transmitter observes outputs from the channel,
    thereby allowing the transmitters to cooperate and achieve higher rates compared to the no-feedback scenario.
    The effect of finite-capacity feedback links on the  capacity region of ICs was also studied in recent works. The work of
    \cite{Aveshimer:11} considered the effect of rate limited feedback on the ICs. In \cite{Aveshimer:11}, communication schemes, based on sending to the
    transmitter partial information on the interfering
    signal, were developed.  The paper \cite{Aveshimer:11} presented a constant-gap result for Gaussian ICs with rate-limited feedback and a tight characterization
    for linear deterministic ICs. In \cite{Gastpar:06} the effect of noisy feedback on the capacity region of Gaussian ICs  was considered.
    For the situation in which both transmitters observe noisy feedback from both receivers, it was shown that feedback looses its value when
    the noise in the feedback signal is of the same variance as the noise in the direct link.
    Finally, note that generalized feedback (or, equivalently source cooperation), studied in  \cite{Tuni:07}, \cite{Yang:09},
    and \cite{Tandon:11} can also considered rate-limited feedback when the SNR is finite. In \cite{Yang:09} and \cite{Tandon:11} outer bounds were derived for ICs with generalized feedback.

    The impact of both relaying and feedback on the DoF of interference channels  was studied in \cite{Cadambe:09}. The work \cite{Cadambe:09} considered a network with
    multiple sources, multiple relays and multiple destinations, in which
    the channel coefficients are random time-varying/frequency-selective and all channel coefficients are known a-priori at all nodes.
    For such a scenario, \cite{Cadambe:09} showed that relays and feedback (and even noisy cooperation between the destinations and the sources) do not
    provide higher total DoF than that obtained without such techniques. However,
    the impact of the combination of relaying and feedback on the capacity of ICs at finite SNRs has not yet been characterized.
In this work we study the capacity of full-duplex fading interference channel with a relay and with
different feedback configurations. We consider the channel when it is subject to phase fading and Rayleigh fading. The phase fading model is mostly applicable to high-speed microwave communications,
in which phase noise is generated by the oscillators or due to the lack of synchronization. The phase fading model also applies to orthogonal frequency division multiplexing (OFDM) \cite{BanNess:04},
as well as to some applications of naval communications. Rayleigh fading models are commonly used in wireless communications and apply to scenarios in which the multipath effect is not negligible, e.g.,
dense urban environments \cite{Sklar:97}.

\subsection*{Main Contributions}
In this paper we present the first investigation of the application of both relaying and feedback to interference channels.
   We provide capacity characterization for the fading interference channel with a relay and feedback links (ICRF), in the SI and VSI regimes. We assume only receiver channel state information (Rx-CSI).
   All capacity regions obtained in this work are derived under the assumptions that the fading channel coefficients are mutually independent and
    i.i.d. in time,  and that the phase of each fading coefficient is uniformly distributed
    over $[0, 2\pi)$, and is independent of its magnitude. Explicit capacity regions are given for two fading models:
    phase fading and Rayleigh fading, which are special cases of this general model.
\begin{itemize}
    \item We first characterize the capacity regions of ICRFs in which both receivers send (noiseless) causal feedback only to the relay,
        for VSI and SI regimes.

    \item Next, we consider the case where feedback is also available at the transmitters to determine whether the transmitters can exploit this
        additional information to cooperate and enlarge the capacity region compared to the first configuration. The answer to this question is not
        immediate since the availability of feedback at the transmitters can enlarge the capacity region of MACs and ICs, but for the relay channel
        it does not provide any improvement once feedback is available at the relay.

    \item We then study the performance when feedback is available only from one of the receivers and examine whether the performance degradation
        is the same for both pairs. Capacity results are provided for this scenario as well.

\end{itemize}
        Identifying optimal strategies for ICRFs has a direct
        impact on the design of future wireless networks in which interference is a critical issue.
        These implications will be highlighted throughout. Some important consequences of our results include a
        proof that a single relay can be optimal simultaneously for two separate Tx-Rx pairs as well as the maximum performance
        gains that can be obtained in different feedback configurations. {\em To the best of our knowledge these are the first capacity results
        for ICs with relaying and feedback}.

The rest of this paper is organized as follows: in section \ref{sec:Model} we define the system model.
In section \ref{sec:Pre} several frequently used lemmas and theorems are provided. In sections \ref{sec:FULL FB VSI} and \ref{sec:FULL FB SI} we provide
an exact characterization of the capacity regions of ICRFs with feedback from both receivers to the relay, in the VSI and SI regimes.
We also provide explicit expressions for the phase fading and Rayleigh fading models\footnote{
For the Rayleigh fading, the expressions include integrations which can be evaluated numerically in a simple manner.}.
In section \ref{sec:FB to tx} we analyze the scenario in which feedback is available both at the relay and at the transmitters.
In section \ref{sec:Partial FB} we consider the case in which partial feedback (only from one of the receivers) is available at the relay.
For this scenario, we characterize the capacity regions in the VSI and SI regimes and provide explicit expressions for the phase fading and Rayleigh
fading models.
Finally, in section \ref{sec:conclusions} we present concluding remarks.

\section {Notations and Channel Model}
\label{sec:Model}
We denote random variables (RVs) with capital letters, e.g., $X,Y$ and their realizations with lower case letters, e.g., $x,y$.
We denote the probability density function (p.d.f.) of a continuous RV $X$ with $f_{X}(x)$.
Capital double-stroke letters are used for matrices, e.g., $\Amat$, with the exception that $\E\{X\}$ denotes the stochastic expectation of $X$.
Vectors are denoted with bold-face letters, e.g., $\mathbf{x}$ and  the $i$'th element of a vector $\mathbf{x}$ is denoted with $x_i$. We use $x_i^j$ where
$i \le j$ to denote the vector $(x_i, x_{i+1},...,x_{j-1},x_j)$. $X^*$ denotes the conjugate of $X$ and $\Amat^H$ denotes the Hermitian transpose of $\Amat$.
Given two $n\times n$ Hermitian matrices, $\Amat, \Bmat$, we write $\Bmat\preceq\Amat$ if $\Amat-\Bmat$ is positive semidefinite (p.s.d.)
and $\Bmat\prec\Amat$ if $\Amat-\Bmat$ is positive definite (p.d.). $A^{(n)}_\epsilon(X,Y)$ denotes the set of weakly jointly typical sequences with
respect to $f_{X,Y}(x,y)$, as defined in \cite[Sec. 8.2]{cover-thomas:it-book}. We denote with $\varnothing$ the empty set. Finally, we denote the Normal distribution with mean $\mu$ and
variance $\sigma^2$ with $\N(\mu,\sigma^2)$, and the circularly symmetric, complex Normal distribution with mean $\mu$ and variance $\sigma^2$ with
$\CN(\mu,\sigma^2)$.

In the interference channel with a relay there are two transmitters and two receivers. Tx$_1$ wants to send a message to Rx$_1$ and Tx$_2$
wants to send a message to Rx$_2$. The received signals at Rx$_1$, Rx$_2$ and the relay at time $i$ are
denoted by $Y_{1,i}$, $Y_{2,i}$, $Y_{3,i}$ respectively. The channel inputs from $\Tgood$, $\Tbad$ and the relay at time $i$ are denoted by
$X_{1,i}$, $X_{2,i}$ and $X_{3,i}$, respectively.
The relationship  between the channel inputs and its outputs is given by:
\begin{subequations}
\label{eqn:ICR_model}
\begin{eqnarray}
    \label{eqn:Rx_1_sig}
\!\!\!\!\!    Y_{1,i} & = & H_{11,i} X_{1,i} + H_{21,i} X_{2,i} + H_{31,i} X_{3,i} + Z_{1,i}\\
    \label{eqn:Rx_2_sig}
\!\!\!\!\!    Y_{2,i} & = & H_{12,i} X_{1,i} + H_{22,i} X_{2,i} + H_{32,i} X_{3,i} + Z_{2,i}\\
    \label{eqn:Relay_sig}
\!\!\!\!\!    Y_{3,i} & = & H_{13,i} X_{1,i} + H_{23,i} X_{2,i} + Z_{3,i},
\end{eqnarray}
\end{subequations}
$i = 1, 2, ..., n$, where $Z_1$, $Z_2$ and $Z_3$ are mutually independent, zero-mean, circularly symmetric complex Normal RVs, $\CN(0,1)$, independent in time and independent of the channel inputs and the channel coefficients.
The channel input signals are subject to per-symbol average power constraints: $\E\big\{|X_k|^2\big\}\le P_k$, $k\in\{1,2,3\}$.
The channel coefficients $H_{lk,i}$ are mutually independent and i.i.d. in time. The magnitude and phase of $H_{lk,i}$ are independent RVs, and the phase is uniformly distributed over $[0,2\pi)$.
\begin{figure}
    \centering
    \includegraphics[scale=0.30]{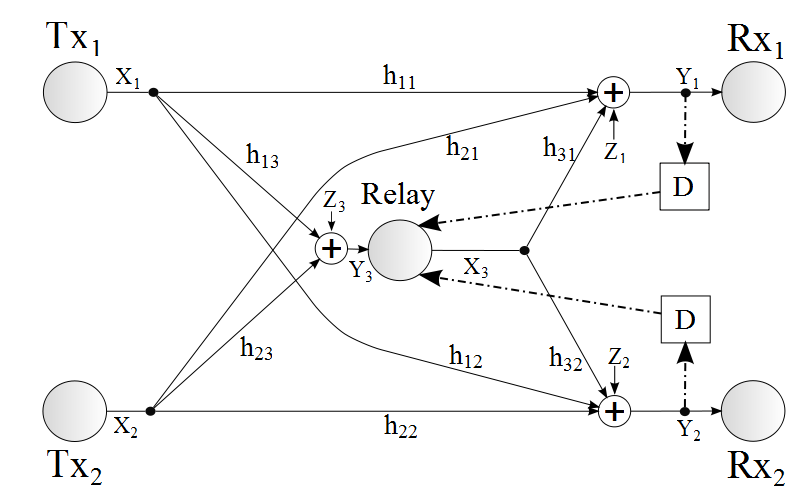}
\caption{The interference channel with a relay and with feedback from both receivers to the relay. The `D' block represents a single-symbol delay.}
\label{fig:ICRF model}
\vspace{-0.7cm}
\end{figure}

Throughout this paper channel state information (CSI) at the receivers is assumed.
        We represent the CSI at receiver $k$  with $\tH_k \triangleq \big( H_{1k}, H_{2k}, H_{3k}\big)$, $k\in\{1,2\}$. As each element in $\tH_k$ is a complex scalar random variable, then
        $\tH_k\in\setC^3$. For consistency of notations
        we use $\tilde{\mathfrak{H}}_k$ to denote the space of the random vector $\tH_k$, thus $\tilde{\mathfrak{H}}_k \equiv \setC^3$.
        In sections \ref{sec:FULL FB VSI} and \ref{sec:FULL FB SI} we assume noiseless feedback links from both receivers to the relay, s.t. the channel outputs $y_{1,1}^{i-1},y_{2,1}^{i-1}$,
        and the corresponding Rx-CSIs, $\th_{1,1}^{i-1}$ and $\th_{2,1}^{i-1}$, are available at the relay at time $i$ prior to transmission. This model is described in Fig.  \ref{fig:ICRF model}.
        Hence, the CSI at the relay is represented by $\utH\triangleq\big(H_{13}, H_{23},\tH_1,\tH_2\big)\in\setC^8$. We denote the space of  $\utH$ with $\underline{\tilde{\mathfrak{H}}}\equiv \setC^8$.

        \vspace{-0.2cm}
\begin{remark}
\em{}
    Note that as feedback contains both channel output and Rx-CSI, feedback from both receivers to the relay, leads to the relay having delayed Tx-CSI on its outgoing links.
    In this work we will show that as the channel is memoryless and the coefficients are i.i.d. with uniformly distributed phases, independent of their magnitudes,
    such feedback does not result in correlated channel inputs.
    Note that  destinations-relay feedback which includes Rx-CSI leads to the conclusion that reliable
    decoding at the destinations guarantees reliable decoding at the relay. This, in turn, leads to the optimality of DF in SI and VSI.
    Without including Rx-CSI in the feedback signal, then, in order to achieve such an implication, it is necessary to impose
    restrictions on the channel coefficients. This decreases the set of channel coefficients for which we can achieve
    the capacity region of the ICRF by using DF at the relay. This will be elaborated upon in Comment \ref{rem:Impact_of_no_CSI_in FB}.
\end{remark}

\begin{remark}
\em{}
    We note that an important problem is the case of global instantaneous CSI. In such a case,
    following the approach in \cite{TianYener:2010} and \cite{Liang:07},  the fading channel is decomposed into parallel Gaussian ICRs. However, for such channels it is not possible
    to use the techniques of the current work to show that mutually independent channel inputs maximize the cut-set bound. This is because  the channel coefficients
    and channel inputs at the same time instant can be correlated, and therefore the nodes can use the CSI to achieve correlation between their signals.
    The case of global instantaneous CSI will not be treated in the current manuscript.
\end{remark}

We now define the code, probability of error, achievable rates, and capacity region:
\begin{definition}
\label{def:code}
\em{}An $(R_1,R_2,n)$ code for the ICRF, depicted in Fig. \ref{fig:ICRF model}, consists of two message sets $\M_k \triangleq \big\{1,2,...,2^{n R_k}\big\}$, $k\in\{1,2\}$,
two encoders at the sources, $e_1, e_2$,  and two decoders at the destinations, $g_1,g_2$; $e_k: \M_k \mapsto \setC^n$, $g_k: \tilde{\mathfrak{H}}_k^n \times \setC^n \mapsto \M_k$, $k\in\{1,2\}$.
At the relay there is a causal encoder. Since in sections \ref{sec:FULL FB VSI} and \ref{sec:FULL FB SI} feedback from both receivers is available at the relay, then the
encoded signal at the relay is a causal function of the channel outputs at the receivers, its own received symbols and the corresponding Rx-CSIs, i.e.,
\begin{equation}
\label{eq:relay encoding}
x_{3,i} = t_i\big(y_{1,1}^{i-1},y_{2,1}^{i-1},y_{3,1}^{i-1},h_{13,1}^{i-1},h_{23,1}^{i-1},\tilde{h}_{1,1}^{i-1},\tilde{h}_{2,1}^{i-1}\big) \in \setC,
\end{equation}
$i=1,2,...,n$.
\end{definition}
\begin{definition}
\em{}The average probability of error is defined as $\P_e^{(n)} \triangleq \Pr\big(g_1(\tH_1^n, Y_1^n) \ne M_1 \mbox{ or } g_2(\tH_2^n, Y_2^n) \ne M_2\big)$,
where $M_1$ and $M_2$ are selected independently and uniformly over their message sets.
\end{definition}
\begin{definition}
\em{}A rate pair $(R_1, R_2)$ is called achievable if for any $\epsilon >0$ and $\delta >0$ there exists some block length $n_0(\epsilon,\delta)$ s.t. for every integer
$n > n_0(\epsilon,\delta)$ there exists an  $(R_1 - \delta, R_2 - \delta ,n)$ code with $\P_e^{(n)} < \epsilon$.
\end{definition}
\begin{definition}
\em{}The capacity region is defined as the convex hull of all achievable rate pairs.
\end{definition}
\noindent
In sections \ref{sec:FB to tx} and \ref{sec:Partial FB}, the definitions of Rx-CSI and the code will be specialized according to the feedback configurations of these sections.

In this paper we also present explicit capacity expressions for phase fading and Rayleigh fading models, which are two fading models that satisfy the general fading model
defined above. These models are defined as follows:
\begin{itemize}
\item {\bf Phase fading channels}:
      The channel coefficients are given by
      $H_{lk,i} = a_{lk} e^{j\Theta_{lk,i}}$, $a_{lk} \in \setR_+$ are non-negative constants corresponding to the attenuation of the signal power from node $l$ to node $k$, and  $\Theta_{lk,i}$ are
       uniformly distributed over $[0,2\pi)$, independent in time and independent of each other and of the additive noises $Z_k$, $k\in\{ 1,2,3\}$.

\item {\bf Rayleigh fading channels}:
      The channel coefficients are given by
      $H_{lk,i} = a_{lk} U_{lk,i}$ , $a_{lk} \in \setR_+$ are non-negative constants corresponding to the attenuation of the signal power from node $l$ to node $k$, and  $U_{lk,i}$ are
      circularly symmetric, complex Normal RVs, $U_{lk,i} \sim \CN(0,1)$, independent in time and independent of each other and of the additive noises $Z_k$, $k\in\{1,2,3\}$.
\end{itemize}

\section {Preliminaries}
In this section we present some of the frequently used lemmas.
\label{sec:Pre}
\subsection{Maximum Entropy for Complex Random Vectors}

\begin{lemma}
\label{lem:lem00}
    Consider a complex random vector, $\Xvec\triangleq(\Xvec_1,\Xvec_2)$.  Let $\Xvec'\triangleq(\Xvec'_1,\Xvec'_2)=\Xvec-\E\{\Xvec\}$. Then  $h(\Xvec'_1|\Xvec'_2)=h(\Xvec_1|\Xvec_2)$.
\end{lemma}
\begin{proof}
The proof follows directly from the definition of the differential entropy.
\end{proof}

\begin{lemma}
\label{lem:lem0}
    Let $X_1,X_2,...,X_k$ be an arbitrary set of $k$ zero-mean complex random variables with covariance matrix $\mathds{K}$. Let $\mathcal{S}$ be any subset of $n$ elements from $\{1,2,...,k\}$ and $\mathcal{S}^C$ be its complement. Then:
    \begin{equation*}
        h(\Xvec_{\mathcal{S}}|\Xvec_{\mathcal{S}^C}) \le \log\Big((\pi e)^n \det\big(\textnormal{cov} (\Xvec_{\mathcal{S}}|\Xvec_{\mathcal{S}^C})\big)\Big),
    \end{equation*}
    with equality if and only if $X_1,X_2,...,X_k \sim \CN(0,\mathds{K})$.
\end{lemma}
\begin{proof}
The proof follows along the lines of the proof of \cite[Lemma 1]{Thomas:87} and an application of \cite[Theorem 1]{Massey:93} and \cite[Theorem 2]{Massey:93}.
\end{proof}

\subsection{The Positive Semidefinite Ordering}
\begin{lemmaB}[{\cite[Lemma 3.1]{Wang:05}}]
\label{Wang}
    Let $\mathbf{X_1}$ and $\mathbf{X_2}$ be random vectors with zero mean and covariance matrices $\C_{mk} \triangleq \E\{\Xvec_m \Xvec^H_k\}, m,k\in\{1,2\}$. Define: $\Amat \triangleq \C^{-\frac{1}{2}}_{11}\cdot \C_{12} \cdot \C^{-\frac{1}{2}}_{22}$. Then there exists $\rho \in [0,1]$ s.t.
    \begin{equation*}
        \Imat-\Amat\Amat^H \preceq (1-\rho^2)\Imat.
    \end{equation*}
\end{lemmaB}

\subsection{Joint Typicality}
\begin{lemmaB}[{\cite[Lemma 2]{CoverElGamal:IT79}}]
\label{Cover}
    Let $(\Svec_1,\Svec_2,\Svec_3) \sim \prod_{i=1}^n p(s_{1,i},s_{2,i},s_{3,i})$ and ${(\Svec'_1 ,\Svec'_2, \Svec_3)\sim\prod_{i=1}^n p(s_{3,i})\times}$ $ p(s_{1,i}|s_{3,i}) p(s_{2,i}|s_{3,i})$. Then for $n$ s.t. $\Pr\big\{A^n_\epsilon(\Svec_1,\Svec_2,\Svec_3)\big\}\geq 1-\epsilon$, it holds that:
    \begin{eqnarray*}
        & &\Pr\big\{(\Svec'_1,\Svec'_2,\Svec_3)\in A^{(n)}_\epsilon(\Svec_1,\Svec_2,\Svec_3)\big\} \\
        & & \quad \qquad \qquad \qquad \qquad\qquad \qquad\le 2^{-n\big(I(S_1;S_2|S_3)-7\epsilon\big)}.
    \end{eqnarray*}
\end{lemmaB}

\subsection{The Capacity of Phase Fading and of Rayleigh Fading MIMO Relay Channels}
    \label{thm:upper_bound}
We now state a slight variation of \cite[Theorem. 8]{Kramer:05} which will be used in this paper:
\begin{theoremB}[{\cite[Theorem. 8]{Kramer:05}}]
    For phase fading and for Rayleigh fading relay channels with multiple antennas, and with Rx-CSI available, the channel inputs
    $\tX_1$ and $\tX_3$ that maximize both the cut-set bound,
    \begin{eqnarray*}
       &  &  \max_{p(\tx_1,\tx_3)}\min\big\{I(\tX_1;\tY_1,\tY_3|\tX_3,\tH_{11},\tH_{31}, \tH_{13}),\\
        &  & \qquad \qquad \qquad \qquad  \qquad  I(\tX_1, \tX_3;\tY_1|\tH_{11},\tH_{31})\big\},
    \end{eqnarray*}
    and the DF rate, $\max_{p(\tx_1,\tx_3)}\min\big\{I(\tX_1;\tY_3|\tX_3,\tH_{13}),$ $I(\tX_1, \tX_3;\tY_1|\tH_{11},\tH_{31})\big\}$,
    are independent complex Normal variables. The best covariance matrix for transmitter $\mbox{Tx}_t$ is
    $\dsQ_{\tX_t} = \sqrt{\frac{P_t}{n_t}} \dsI_{n_t}$, $t = 1,2$, where $\dsI_{n_t}$ is the $n_t\times n_t$ identity matrix.
    DF achieves capacity if its rate is $I(\tX_1,\tX_3;\tY_1)$. The capacity is then given by
    \begin{eqnarray*}
        C_{\mbox{\scriptsize fading relay}} & = & \int_{\th_\mT,1}f(\th_{\mT,1})\log_2 \bigg| \dsI_{l_1} + a_{11}^2\frac{P_1}{n_1}\th_{11}\th_{11}^H\\
    &  & \qquad \qquad\qquad \qquad
                + a_{31}^2\frac{P_3}{n_3}\th_{31}\th_{31}^H\bigg| d\th_{\mT,1},
    \end{eqnarray*}
    where $\th_{\mT,1} \triangleq \big\{ \th_{11}, \th_{31}\big\}$.
\end{theoremB}


\section {ICRFs in the Very Strong Interference Regime}
\label{sec:FULL FB VSI}
In this section, we consider the ICRF with two noiseless feedback links from the receivers to the relay (see Fig. \ref{fig:ICRF model}) and we characterize the capacity region of ICRFs in the VSI regime.
This result is stated in the following theorem:

\begin{theorem}
\label{thm:ICRF VS}
                Consider the fading ICRF with Rx-CSI. Assume that the channel coefficients are independent in time and independent of each other s.t. their phases are i.i.d. and distributed
                uniformly over $[0,2\pi)$. Let the additive noises be i.i.d. circularly symmetric complex Normal processes, $\CN(0,1)$, and let the sources have power constraints $\E\big\{|X_k|^2\big\} \le P_k$,
                $k\in\{1,2,3\}$. Assume noiseless feedback links from both receivers to the relay (see Fig. \ref{fig:ICRF model}). If
                \begin{subequations}
                \label{eq:ICRF VSI Con}
                \begin{eqnarray}
                \label{eq:first con in VSI con}
                I(X_1,X_3;Y_1|X_2,\tH_1) &\le& I(X_1;Y_2|\tH_2)\\
                \label{eq:first con in VSI con2}
                I(X_2,X_3;Y_2|X_1,\tH_2) &\le& I(X_2;Y_1|\tH_1),
                \end{eqnarray}
                \end{subequations}
                where the mutual information expressions are evaluated with $X_k\sim \CN(0,P_k), k\in\{1,2,3\}$, mutually independent,
                then the capacity region is given by all the nonnegative rate pairs s.t.
                \begin{subequations}
                \label{eq:ICRF_region}
                \begin{eqnarray}
                R_1 &\le& I(X_1,X_3;Y_1|X_2,\tH_1)\\
                R_2 &\le& I(X_2,X_3;Y_2|X_1,\tH_2),
                \end{eqnarray}
                \end{subequations}
                and it is achieved with $X_k\sim \CN(0,P_k), k\in\{1,2,3\}$, mutually independent and with DF strategy at the relay.
\end{theorem}

\subsection{Proof of Theorem \ref{thm:ICRF VS}}
The proof consists of the following steps:
\begin{itemize}
    \item We obtain an outer bound on the capacity region using the cut-set bound.
    \item We show that the input distribution that maximizes the outer bound is zero-mean, circularly symmetric complex Normal with channel inputs independent of each other and with maximum allowed power.
    \item We derive an achievable rate region based on DF at the relay and by using mutually independent codebooks generated according to the zero-mean, circularly symmetric complex Normal input distribution:
         \begin{itemize}
         \item We derive an achievable rate region for decoding at the relay using steps similar to \cite[Sec. 4.D]{Kramer:05}.
        \item We obtain an achievable rate region for decoding at the destination by decoding the interference first, while treating the relay signal and the desired signal as additive i.i.d. noises, followed by using a backward decoding scheme for decoding the desired message.
        \end{itemize}
    \item We derive the VSI conditions which guarantee that decoding the interference first at each receiver, does not constrain the rate of the other pair.
    \item We conclude that when the VSI conditions hold the achievable region coincides with the cut-set bound.
\end{itemize}
These steps are elaborated in sections \ref{sec:Full FB VSI upper bounds }, \ref{sec:Achievabiliy VSI} and \ref{sec: Full FB CR for VSI}.

\subsubsection{An Outer Bound}
\label{sec:Full FB VSI upper bounds }
The cut-set theorem \cite[Theorem 15.10.1]{cover-thomas:it-book} applied to the ICRF results in the following upper bounds:
\begin{subequations}
\label{eq:cut-set bound}
\begin{eqnarray}
\label{eq:cs bound R11}
\mathcal{S} &\triangleq& \{\mbox{Tx}_1\}, \mathcal{S}^C\triangleq \{\mbox{Tx}_2, \mbox{Relay}, \mbox{Rx}_1, \mbox{Rx}_2\} :\nonumber\\
&  & \quad R_1\le I(X_1;Y_1,Y_2,Y_3|X_2,X_3,\utH)\\
\label{eq:cs bound R12}
\mathcal{S} &\triangleq& \{\mbox{Tx}_1, \mbox{Relay}, \mbox{Rx}_2\},\mathcal{S}^C \triangleq \{\mbox{Tx}_2, \mbox{Rx}_1\} :\nonumber\\
&  & \quad  R_1\le I(X_1,X_3;Y_1|X_2,\tH_1)\\
\label{eq:cs bound R21}
\mathcal{S} &\triangleq& \{\mbox{Tx}_2\}, \mathcal{S}^C \triangleq\{\mbox{Tx}_1, \mbox{Relay}, \mbox{Rx}_1, \mbox{Rx}_2\} :\nonumber\\
&  & \quad  R_2\le I(X_2;Y_1,Y_2,Y_3|X_1,X_3,\utH)\\
\label{eq:cs bound R22}
\mathcal{S} &\triangleq& \{\mbox{Tx}_2, \mbox{Relay}, \mbox{Rx}_1\},\mathcal{S}^C \triangleq \{\mbox{Tx}_1, \mbox{Rx}_2\} :\nonumber\\
&  & \quad  R_2\le I(X_2,X_3;Y_2|X_1,\tH_2)\\
\mathcal{S} &\triangleq& \{\mbox{Tx}_1, \mbox{Tx}_2\}, \mathcal{S}^C \triangleq \{ \mbox{Rx}_1, \mbox{Rx}_2, \mbox{Relay}\} :\nonumber\\
&  & \quad  R_1+R_2\le I(X_1,X_2;Y_1,Y_2,Y_3|X_3,\utH).
\end{eqnarray}
\end{subequations}
Next, we find the channel input distribution that maximizes the cut-set bound. We follow the same approach as in \cite[Proposition 2]{Kramer:05} and
\cite[Theorem 8]{Kramer:05}.
Let $\Xvec$ denote the channel inputs with the maximizing distribution. Note that Lemma \ref{lem:lem00} states that
the zero-mean complex random vector $\Xvec'\triangleq\Xvec-\E\{\Xvec\}$ has the same entropy as $\Xvec$. Hence, the most efficient strategy would be to transmit $\Xvec'$ rather than $\Xvec$, since subtracting
the average reduces the power consumption. {Using the steps detailed in Appendix \ref{app:maximzing dist for cut set bound}}, we conclude that each mutual information expression in \eqref{eq:cut-set bound} is
maximized by (zero-mean) circularly symmetric complex Normal channel inputs, independent of each other, and with the sources transmitting at their maximum available power {\em even though the scenario
consists of a combination of relaying and feedback}.

\begin{remark}
\em{}Note that this conclusion is not immediate from \cite[Theorem 8]{Kramer:05}, since in the ICRF there are two destinations, while the cut-set bound in \cite[Theorem 8]{Kramer:05} considers only
one destination and $T$ transmitting relays. Hence, the conditional entropies in the present case contain more complicated combinations of the correlation coefficients between the channel inputs and thus
each expression needs to be examined individually.
\end{remark}

\begin{remark}
\em{} Note that although the cut-set bound of the ICRF scenario requires maximization over all input distributions of the type $f(x_1)f(x_2)f(x_3|x_1,x_2)$, in  Appendix \ref{app:maximzing dist for cut set bound}
it is shown that $f(x_3|x_1,x_2) = f(x_3)$ is the maximizing distribution at the relay, and that the input distribution is jointly Gaussian (as follows from \cite[Proposition 2]{Kramer:05}).
The intuition behind the mathematical result is that as receivers have Rx-CSI, then the mutual information expressions
involve averaging over all channel coefficients. However, as the phases are all uniformly distributed over $[0,2\pi)$, it follows that for any cross-correlation structure between the channel inputs, the same rate bounds
can be obtained by the negative cross-correlation. Thus the maximum rate is achieved when the cross-correlations are equal to their negatives, and are therefore zero. As the maximizing distribution is
uncorrelated Gaussians, they are also independent. This also reflects the fact that due to the i.i.d. uniform phase of the fading process, it is not possible to correlate the channel codewords of the different
transmitting nodes, leading, due to Gaussianity, to independence.
\end{remark}

\subsubsection{An Achievable Rate Region}
\label{sec:Achievabiliy VSI}
Now we obtain an achievable rate region using the input distribution that maximizes the cut-set bound in \eqref{eq:cut-set bound}.
The achievability is based on DF strategy at the relay. Fix the blocklength $n$ and the input distribution
$f_{X_1,X_2,X_3}(x_1,x_2,x_3)=f_{X_1}(x_1)\cdot f_{X_2}(x_2)\cdot f_{X_3}(x_3)$ where $f_{X_k}(x_k)\sim\CN(0,P_k), k=1,2,3$.
Consider the following coding scheme, in which $B-1$ messages are transmitted using $nB$ channel symbols:

\paragraph{Code Construction}
\label{sec:ICRF Code Book}
For each message $m_k \in \mathcal{M}_k, k\in\{1,2\}$ select a codeword $\xvec_k(m_k)$ according to the p.d.f.
$ f_{\Xvec_k}\big(\xvec_k(m_k)\big) =\prod_{i=1}^n f_{X_k}\big(x_{k,i}(m_k)\big)$.
For each $(m_1,m_2) \in \mathcal{M}_1 \times \mathcal{M}_2 $ select a codeword $\xvec_3(m_1,m_2)$ according to the p.d.f.
$f_{\Xvec_3}\big(\xvec_3(m_1,m_2)\big) =\prod_{i=1}^n f_{X_3}\big(x_{3,i}(m_1,m_2)\big)$.

\paragraph{Encoding at Block $b$}
\label{sec:ICRF Encoding}
At block $b$, Tx$_k$ transmits $m_{k,b}$ using $\xvec_k(m_{k,b}), k\in\{1,2\}$. Let $(\hat{m}_{1,b-1},\hat{m}_{2,b-1})$ denote
the decoded ($m_{1,b-1},m_{2,b-1}$) at block $b-1$ at the relay. At block $b$ the relay transmits $\xvec_3(\hat{m}_{1,b-1},\hat{m}_{2,b-1})$.
At block $b=1$ the relay transmits $\xvec_3(1,1)$, and at block $b=B$, Tx$_1$ and Tx$_2$ transmit $\xvec_1(1)$ and $\xvec_2(1)$, respectively.

\paragraph{Decoding at the Relay at Block $b$}
\label{sec:ICRF relay decoding}
Decoding at the relay is very similar to the MARC case studied in \cite[Sec. 4.D]{Kramer:05}, the difference being that here feedback is available at the relay.
In the present case, the relay uses its knowledge of $\Yvec_1(b),\Yvec_2(b),\Yvec_3(b)$ and $\mathbf{\utH}(b)$ to decode $(m_{1,b},m_{2,b})$ by using a joint-typicality decoder.
The decoder looks for a unique pair, $(m_1,m_2)\in\mathcal{M}_1\times\mathcal{M}_2$ that satisfies:
\begin{eqnarray}
\label{eq:VSI Decoing at Relay Rule}
& &\!\!\!\! \!\!\!\!\!\!    \Big(\xvec_1(m_1),\xvec_2(m_2),\xvec_3(m_{1,b-1},m_{2,b-1}),\yvec_1(b),\yvec_2(b),\nonumber\\
&  &  \yvec_3(b),\mathbf{\uth}(b)\Big) \in A_\epsilon^{(n)}(X_1,X_2,X_3,Y_1,Y_2,Y_3,\utH).
\end{eqnarray}
Following the analysis in \cite[Sec. 4.D]{Kramer:05}, it is concluded that the achievable rate region for decoding at the relay is given by:
\begin{subequations}
\label{eq:VSI Decoing at Relay}
\begin{eqnarray}
 &  &\!\!\!\!\!\!  \!\!\!\!\!\! \!\!\!\!\!\! \R_{\mbox{\scriptsize{Relay Decoding}}}   = \nonumber\\
 &  &\!\!\!\!\!\! \!\!\!\!\!\!    \bigg\{ (R_1, R_2)\in \Rset^2_+: \nonumber\\
 &  & \;R_1  \le  I(X_1;Y_1,Y_2,Y_3|X_2,X_3,\utH)\nonumber\\
&  &   \; R_2  \le  I(X_2;Y_1,Y_2,Y_3|X_1,X_3,\utH)\nonumber\\
    \label{eq:eq36}
&  &   \; R_1 + R_2  \le  I(X_1,X_2;Y_1,Y_2,Y_3|X_3,\utH) \bigg\}.
\end{eqnarray}
\end{subequations}

\paragraph{Decoding at the Destinations at Block $b$}
  The receivers use a backward block decoding method as in \cite[Appendix A]{Kramer:05}. Assume that each receiver has correctly decoded $(m_{1,b+1}, m_{2,b+1})$. Recall that the codebooks are
  generated independently, thus, in order to decode $m_{k,b}, k\in\{1,2\}$ each receiver first decodes the interference, i.e., Rx$_1$ decodes $m_{2,b}$ and Rx$_2$ decodes $m_{1,b}$ by treating
  the signal from the relay and its own desired signal as i.i.d. additive noise, {\em independent of the interfering signal}, which holds by construction of the codebooks and by the i.i.d. channel assumption.
  Note that for this decoding step the channel is treated as a PtP channel,
  the capacity of which is derived in \cite[Ch. 7.1]{cover-thomas:it-book}. Thus, due to Rx-CSI, Rx$_1$ can decode the interference if
        \begin{subequations}
        \label{eq:VSI condition}
        \begin{equation}
        R_2 \le I(X_2;Y_1|\tH_1),
        \end{equation}
        and Rx$_2$ can decode the interference if
        \begin{equation}
        R_1 \le I(X_1;Y_2|\tH_2).
        \end{equation}
        \end{subequations}
        After decoding the interference, each receiver uses its CSI to decode its desired message. We consider the decoding process at Rx$_1$; the decoding process at Rx$_2$ follows the same steps.

\begin{itemize}
    \item Rx$_1$ generates the sets:
    \begin{eqnarray*}
   \!\!\!\!\!\! \!\!\!\!\!\!      \mathcal{E}_{0,b}  &\triangleq& \Big \{m_1\in\mathcal{M}_1: \big(\xvec_1(m_{1,b+1}),\xvec_2(m_{2,b+1}),\\
    & & \;\xvec_3(m_1,\hat{m}_{2,b}) ,\mathbf{y}_1(b+1), \mathbf{\hvec}_1(b+1)\big) \in A_\epsilon^{(n)}\Big\}.\\
  \!\!\!\!\!\! \!\!\!\!\!\!       \mathcal{E}_{1,b}  &\triangleq& \Big\{m_1\in\mathcal{M}_1: \big( \xvec_1(m_1), \xvec_2(\hat{m}_{2,b}),\\
  & & \;\qquad\qquad \qquad \mathbf{y}_1(b),\mathbf{\hvec}_1(b)\big) \in A_\epsilon^{(n)} \Big\}.
    \end{eqnarray*}
    \item Rx$_1$ then decodes $m_{1,b}$ by finding a unique $m_{1}\in \mathcal{E}_{0,b} \cap \mathcal{E}_{1,b}$.
\end{itemize}
    Note that since the codewords are independent of each other, $\mathcal{E}_{0,b}$ is independent of $\mathcal{E}_{1,b}$. Thus, assuming $\hat{m}_{2,b}=m_{2,b}$,
    and using standard joint-typicality arguments \cite[Theorem. 7.6.1]{cover-thomas:it-book}, it follows that
    the probability of decoding error can be made arbitrarily small by taking $n$ large enough as long as
    \begin{subequations}
    \label{eq:VSI R1R2 decoding Bound}
    \begin{eqnarray}
         R_1 & \le & I(X_1;Y_1|X_2,\tH_1)+I(X_3;Y_1|X_1,X_2,\tH_1)\nonumber\\
         & = & I(X_1,X_3;Y_1|X_2,\tH_1),
    \end{eqnarray}
and for decoding at Rx$_2$ we obtain
    \begin{equation}
         \label{eq:VSI R1R2 decoding Bound22}
         R_2 \le I(X_2,X_3;Y_2|X_1,\tH_2).
    \end{equation}
    \end{subequations}
Combining with \eqref{eq:VSI condition} we conclude that subject to reliable decoding at the relay, the achievable rate region for decoding at the destinations is characterized by:
    \begin{subequations}
    \label{eq:VSI Decoing at Destination}
    \begin{eqnarray}
       &  &\!\!\!\!\!\!  \!\!\!\!\!\! \!\!\!\!\!\!   \R'_{\mbox{\scriptsize{Destination Decoding}}}\nonumber\\
       &  &\!\!\!\!\!\!  \!\!\!\!\!\!    =  \bigg\{ (R_1, R_2)\in \Rset^2_+:\nonumber\\
      & &  \;\; R_1 \le  \min\big\{I(X_1,X_3;Y_1|X_2,\tH_1),\nonumber\\
      & & \;\; \qquad\qquad \qquad \qquad I(X_1;Y_2|\tH_2)\big\}\\
      &  &\;\;  R_2 \le \min\big\{I(X_2,X_3;Y_2|X_1,\tH_2),\nonumber\\
      &  & \;\; \qquad\qquad \qquad\qquad I(X_2;Y_1|\tH_1)\big\} \bigg\}.
    \end{eqnarray}
    \end{subequations}
Hence, an achievable rate region is obtained by
\begin{equation}
\label{eq:VSI Full FB achievable region }
    \R_{\mbox{\scriptsize{Achievable}}}=\R_{\mbox{\scriptsize{Relay Decoding}}} \cap \R'_{\mbox{\scriptsize{Destination Decoding}}}.
\end{equation}

\subsubsection{Capacity Region for the Very Strong Interference Regime}
\label{sec: Full FB CR for VSI}
Now we obtain the conditions on the channel coefficients which guarantee that the interference is strong enough s.t. the receivers can decode the interference without reducing the rate region. Combining \eqref{eq:VSI Decoing at Relay} and \eqref{eq:VSI Decoing at Destination} with \eqref{eq:VSI condition}, we obtain the VSI conditions for the ICRF:
    \begin{subequations}
    \label{eq:eq35}
    \begin{eqnarray}
        &  & \!\!\!\!\!\!  \!\!\!\!\!\!\min\big\{I(X_1;Y_1,Y_2,Y_3|X_2,X_3,\utH),\nonumber\\
        &  & \quad I(X_1,X_3;Y_1|X_2,\tH_1)\big\} \le I(X_1;Y_2|\tH_2)\\
        &  & \!\!\!\!\!\!  \!\!\!\!\!\!\min\big\{I(X_2;Y_1,Y_2,Y_3|X_1,X_3,\utH),\nonumber\\
        &  & \quad I(X_2,X_3;Y_2|X_1,\tH_2)\big\} \le I(X_2;Y_1|\tH_1).
    \end{eqnarray}
    \end{subequations}
Thus, when \eqref{eq:eq35} holds the achievable region is given by \eqref{eq:VSI R1R2 decoding Bound} and \eqref{eq:VSI Decoing at Relay}. Note that since the codebooks are independent of each other and of the channel coefficients,
\begin{subequations}
\label{eq:VSI FB loosened conditions on relay}
     \begin{eqnarray}
         I(X_1;Y_2|\tH_2)&=&h(X_1|\tH_2)-h(X_1|Y_2,\tH_2)\nonumber\\
         &=&h(X_1|X_2,X_3,\tH_2)-h(X_1|Y_2,\tH_2)\nonumber\\
         &\le&h(X_1|X_2,X_3,\utH)\nonumber\\
         &  & -h(X_1|Y_1,Y_2,Y_3,X_2,X_3,\utH)\nonumber\\
         \label{eq:eq421}
         &=&I(X_1;Y_1,Y_2,Y_3|X_2,X_3,\utH),
    \end{eqnarray}
we also obtain
\begin{equation}
        \label{eq:eq422}
        I(X_2;Y_1|\tH_1) \le I(X_2;Y_1,Y_2,Y_3|X_1,X_3,\utH).
\end{equation}
\end{subequations}
Hence, the conditions in \eqref{eq:eq35} reduce to
    \begin{subequations}
    \label{eq:VSI Con}
    \begin{eqnarray}
        I(X_1,X_3;Y_1|X_2,\tH_1) &\le& I(X_1;Y_2|\tH_2)\\
        I(X_2,X_3;Y_2|X_1,\tH_2) &\le& I(X_2;Y_1|\tH_1),
    \end{eqnarray}
    \end{subequations}
which give \eqref{eq:ICRF VSI Con}. Next, note that \eqref{eq:VSI FB loosened conditions on relay} and \eqref{eq:VSI Con} imply that the achievable region is characterized by \eqref{eq:VSI R1R2 decoding Bound}
and \eqref{eq:eq36}. However, when \eqref{eq:VSI FB loosened conditions on relay} and \eqref{eq:VSI Con} hold, then
    \begin{eqnarray*}
 & &\!\!\!\! \!\!\!  I(X_1,X_2;Y_1,Y_2,Y_3|X_3,\utH)\\
    & & = I(X_2;Y_1,Y_2,Y_3|X_3,\utH)\\
    &  & \qquad +I(X_1;Y_1,Y_2,Y_3|X_2,X_3,\utH)\\
    & & \ge I(X_2;Y_1|\tH_1)+I(X_1;Y_2|\tH_2)\\
    & & \ge I(X_2,X_3;Y_2|X_1,\tH_2)+I(X_1,X_3;Y_1|X_2,\tH_1).
    \end{eqnarray*}
Therefore, we see that in the VSI regime, the sum-rate condition, \eqref{eq:eq36}, is always satisfied. We conclude that when \eqref{eq:ICRF VSI Con} holds, \eqref{eq:ICRF_region} defines the achievable region.
Finally, note that the rate region characterized by \eqref{eq:ICRF_region} coincides with the cut-set bound in section \ref{sec:Full FB VSI upper bounds } (since \eqref{eq:ICRF_region} is only a subset of the
constraints but it is achievable), hence it is the capacity region of the ICRF in the VSI regime.
\vspace{-1cm}
\begin{flushright}
$\blacksquare$
\end{flushright}
\vspace{-0.5cm}

\subsection{Ergodic Phase Fading}
The capacity region of ICRFs under ergodic phase fading in the VSI regime is characterized explicitly in the following corollary:

\begin{corollary}
\label{Cor:ICRF VSI Phase fading}
            Consider the phase fading ICRF with Rx-CSI and noiseless feedback links from both receivers to the relay, s.t. $y_{1,1}^{i-1},y_{2,1}^{i-1},\hvec_{1,1}^{i-1}$ and $\hvec_{2,1}^{i-1}$ are available at the relay at time $i$. If the channel coefficients satisfy
            \begin{subequations}
            \label{eq:VSPhaseFadingConditions}
            \begin{eqnarray}
            a_{11}^2P_1+a_{31}^2P_3 &\le& \frac{a_{12}^2P_1}{1+a_{22}^2P_2+a_{32}^2P_3}\\
            a_{22}^2P_2+a_{32}^2P_3 &\le& \frac{a_{21}^2P_2}{1+a_{11}^2P_1+a_{31}^2P_3},
            \end{eqnarray}
            \end{subequations}
            then the capacity region is characterized by all the nonnegative rate pairs s.t.
            \begin{subequations}
            \label{eq:VSPhaseFadingICRF_region}
            \begin{eqnarray}
            R_1 &\le& \log_2\big(1+a_{11}^2P_1+a_{31}^2P_3\big)\\
            R_2 &\le& \log_2\big(1+a_{22}^2P_2+a_{32}^2P_3\big),
            \end{eqnarray}
            \end{subequations}
            and it is achieved with $X_k\sim \CN(0,P_k), k\in\{1,2,3\}$, mutually independent and with DF strategy at the relay.
\end{corollary}

\begin{proof}
The proof follows from the expressions of Theorem \ref{thm:ICRF VS}. In order to obtain the conditions on the channel coefficients in \eqref{eq:VSPhaseFadingConditions}, we evaluate $I(X_1,X_3;Y_1|X_2,\tH_1)$ and $I(X_2,X_3;Y_2|X_1,\tH_2)$ using the right-hand side (r.h.s.) of equation \eqref{eq:eq6} in Appendix \ref{app:maximzing dist for cut set bound}. Recall that the channel inputs that maximize these expressions are mutually independent, zero mean, circularly symmetric complex Normal and with maximum power. Thus, we obtain
    \begin{eqnarray*}
    I(X_1,X_3;Y_1|X_2,\tH_1) &=& \log_2(1+a_{11}^2P_1+a_{31}^2P_3)\\
    I(X_2,X_3;Y_2|X_1,\tH_2) &=& \log_2(1+a_{22}^2P_2+a_{32}^2P_3).
    \end{eqnarray*}
    Note that for evaluating the r.h.s. of \eqref{eq:ICRF VSI Con}, $\{H_{31,i}X_{3,i}\}_{i=1}^n$ and $\{H_{11,i}X_{1,i}\}_{i=1}^n$ are treated as additive Gaussian noises\footnote{Recalling $X_1,X_2$ and $X_3$ are i.i.d. circularly symmetric complex Normal RVs with zero mean, their phases are distributed uniformly over $[0,2\pi)$ i.i.d. and independent of each other and of the magnitudes. Under the phase fading model, the channel coefficients have fixed amplitudes and their phases are i.i.d. and distributed uniformly over $[0,2\pi)$. Thus,  $H_{31}X_3,H_{11}X_1,H_{32}X_3$ and $H_{22}X_2$ are mutually independent, zero mean, circularly symmetric complex Normal RVs as well.} at Rx$_1$ and $\{H_{32,i}X_{3,i}\}_{i=1}^n$ and $\{H_{22,i}X_{2,i}\}_{i=1}^n$ are treated as additive Gaussian noises at Rx$_2$. Hence, we obtain
    \begin{eqnarray*}
    I(X_2;Y_1|\tH_1)&=&\log_2\Big(1+\frac{a_{21}^2P_2}{1+a_{11}^2P_1+a_{31}^2P_3}\Big)\\
    I(X_1;Y_2|\tH_2)&=&\log_2\Big(1+\frac{a_{12}^2P_1}{1+a_{22}^2P_2+a_{32}^2P_3}\Big).
    \end{eqnarray*}
Thus \eqref{eq:ICRF VSI Con} results in conditions \eqref{eq:VSPhaseFadingConditions} and \eqref{eq:ICRF_region} results in \eqref{eq:VSPhaseFadingICRF_region}.
\end{proof}

\subsection{Ergodic Rayleigh Fading}
In this section the capacity region of ICRFs under ergodic Rayleigh fading in the VSI regime is characterized. Define $\tU_k \triangleq(U_{kk}, U_{3k}), k\in\{1,2\}$ and define $E_1(x)$ as in \cite[Eqn. 5.1.1]{Abramowitz}:
    \begin{equation*}
        E_1(x)\triangleq\int^{\infty}_x \frac{e^{-t}}{t}dt .
    \end{equation*}

\begin{corollary}
            Consider the Rayleigh fading ICRF with Rx-CSI and noiseless feedback links from both receivers to the relay, s.t. $y_{1,1}^{i-1},y_{2,1}^{i-1},\hvec_{1,1}^{i-1}$ and $\hvec_{2,1}^{i-1}$ are available at the relay at time $i$. If the channel coefficients satisfy
            \begin{subequations}
            \label{eq:VSRayleighFadingConditions}
            \begin{eqnarray}
            \label{eqn:VSI_conds_Rayleigh_Rx1}
            & &\!\!\!\!\!\!\!\! \frac{\frac{ a_{12}^2  P_1 }{1 + a_{22}^2 P_2 + a_{32}^2  P_3}}{e^{ \frac{1 +a_{22}^2  P_2 + a_{32}^2  P_3}{ a_{12}^2  P_1 }}  E_1\left(  \frac{1 + a_{22}^2 P_2 + a_{32}^2  P_3}{ a_{12}^2  P_1 }\right)}\nonumber\\
               & & \qquad \qquad \qquad \ge     (1 + a_{11}^2 P_{1} + a_{31}^2 P_{3}  )\\
            \label{eqn:VSI_conds_Rayleigh_Rx2}
            & &\!\!\!\! \!\!\!\! \frac{\frac{ a_{21}^2  P_2 }{1 +a_{11}^2  P_1 + a_{31}^2  P_3}}{e^{ \frac{1 + a_{11}^2 P_1 + a_{31}^2  P_3}{ a_{21}^2  P_2 }}  E_1\left(  \frac{1 + a_{11}^2 P_1 + a_{31}^2  P_3}{ a_{21}^2  P_2 }\right)}\nonumber\\
               & & \qquad \qquad  \qquad \ge    (1 +  a_{22}^2 P_{2} + a_{32}^2 P_{3}  ),
            \end{eqnarray}
            \end{subequations}
            then the capacity region is characterized by all the nonnegative rate pairs s.t.
            \begin{subequations}
            \label{eq:VSRayleighFadingICRF_region}
            \begin{eqnarray}
            \!\!\!\!\!\!\!\!\!\!\!\!R_1 &\!\le&\! \E_{\tU_1} \!\big\{\log_2(1+a_{11}^2|U_{11}|^2P_1+a_{31}^2|U_{31}|^2P_3)\big\}\\
            \!\!\!\!\!\!\!\!\!\!\!\!R_2 &\!\le&\! \E_{\tU_2} \!\big\{\log_2(1+a_{22}^2|U_{22}|^2P_2+a_{32}^2|U_{32}|^2P_3)\big\},
            \end{eqnarray}
            \end{subequations}
            and it is achieved with $X_k\sim \CN(0,P_k), k\in\{1,2,3\}$, mutually independent and with DF strategy at the relay.
\end{corollary}
\begin{proof}
The proof follows the same approach as in Corollary \ref{Cor:ICRF VSI Phase fading}. The detailed calculation of \eqref{eq:VSRayleighFadingConditions} can be found in \cite{Dabora:10}. Recall that for Rayleigh fading, the channel coefficients are complex Normal RVs. Thus, for decoding the interference at the receivers, $H_{31}X_3,H_{11}X_1,H_{32}X_3$ and $H_{22}X_2$ cannot be treated as additive Gaussian noises. Hence, the mutual information expressions on the r.h.s. of \eqref{eq:ICRF VSI Con} need to be bounded using the $E_1$ function.
\end{proof}

\subsection{Comments}
\begin{remark}
\label{rem:Comparing_FB_with_NFB_VSI}
\em{}

    We now compare the capacity region of the ICRF in VSI to the ICR without feedback in VSI.
        First, consider the phase fading model, and define the set of coefficients $\mD_{PF}$ as follows:
        \begin{subequations}
    \label{eqns:decode_at_relay}
    \begin{eqnarray}
   & &\!\!\!\!\!\!\!\!\!\! \!\!\!\!\!\!\!\!\!\! \mD_{PF}  \triangleq  \Big\{ (a_{31}, a_{13}, a_{32}, a_{23}) \in \setR_+^4:\nonumber\\
             \label{eqn:PF_VSI_cond_1}
    & &    \qquad \qquad    a_{11}^2P_1 + a_{31}^2 P_3  \le  a_{13}^2 P_1\\
             \label{eqn:PF_VSI_cond_2}
    & &    \qquad \qquad    a_{22}^2P_2 + a_{32}^2 P_3  \le  a_{23}^2 P_2\\
             \label{eqn:PF_VSI_cond_3}
& &             \!\!\!\!\!\!\!(1\! + a_{11}^2P_1\! + a_{31}^2 P_3) (1\! + a_{22}^2P_2\! +\! a_{32}^2 P_3) \nonumber\\
             & &\qquad\qquad\qquad  \le  1\! +\! a_{13}^2 P_1\! +\! a_{23}^2 P_2\Big\}.
    \end{eqnarray}
    \end{subequations}
    In \cite[Theorem 1]{Dabora:101} it is shown that when the channel coefficients satisfy \eqref{eq:VSPhaseFadingConditions}
     and  $(a_{31}, a_{13}, a_{32}, a_{23}) \in \mD_{PF}$, then capacity of the ICR is given by \eqref{eq:VSPhaseFadingICRF_region}.
    Note that $(a_{31}, a_{13}, a_{32}, a_{23}) \in \mD_{PF}$ implies that the links from the transmitters to the relay are good in the sense that if a rate pair can be reliably decoded
    at the destinations, then, it can also be reliably decoded at the relay.
    Observe that feedback does not affect the rate constraints \eqref{eq:VSPhaseFadingICRF_region}, thus, when capacity is achieved without any feedback, then additional  feedback links from each receiver to the relay
    do not enlarge the capacity region. Hence, the main benefit of feedback to the relay is that it allows to achieve capacity in VSI {\em for any quality of links from the transmitters to the relay}.
    Therefore, the set of channel coefficients for which capacity is achieved is defined only by \eqref{eq:VSPhaseFadingConditions} without the additional restrictions of $\mD_{PF}$.

    We also note that for the ICRF, when \eqref{eq:VSPhaseFadingConditions} holds, the cut-set bound is given by
    \eqref{eq:VSPhaseFadingICRF_region}. Consider next the ICR in which \eqref{eq:VSPhaseFadingConditions} holds yet $(a_{31}, a_{13}, a_{32}, a_{23}) \notin \mD_{PF}$.
    Taking $a_{13}, a_{23} \rightarrow 0$, we eventually obtain that
    the cut-set bound (see \cite[Eqn. (C.1)]{Dabora:10}) is a subset of \eqref{eq:VSPhaseFadingICRF_region}. For such scenarios, feedback {\em enlarges} the capacity region compared to the no-feedback case.
    Similar conclusions hold also for Rayleigh fading.
\end{remark}
\begin{remark}
\label{rem:fig:Commnet_on_VSI-FULLNO}
\em{} 

        Fig. \ref{fig:VSI-FULLNO} shows the position of the relay in a 2D-plane in which the VSI conditions,  \eqref{eq:ICRF VSI Con}, are satisfied for the phase
        fading scenario with $P_1=P_2=10, P_3=3$. For phase fading  \eqref{eq:ICRF VSI Con} are evaluated to be \eqref{eq:VSPhaseFadingConditions} .
        Each channel coefficient $a_{ij}$ is related to the distance $d_{ij}$ from node $i$ to node $j$ via $a_{ij}=\frac{1}{d_{ij}^2}$, and hence the path-loss exponent is $4$, corresponding to the two-ray propagation model.
        The locations of the transmitters and the receivers are fixed, thus the corresponding channel coefficients are fixed to be
         $a_{11}=a_{22}=0.18$ and $a_{12}=a_{21}=0.25$. Note that indeed the cross-links are stronger than the direct links.

        From the figure we observe that with feedback, the VSI conditions \eqref{eq:VSPhaseFadingConditions} hold (hence, the capacity region is known) in both
        the black and the gray areas, while without feedback, the conditions \cite[Eqns. (8) and (9)]{Dabora:101}
        hold only in the black area, thus capacity is achieved with DF only in that area. This clearly shows the benefits of feedback.
        Note that without feedback, the relay has to be close to the transmitters and far enough from the destinations, to satisfy the conditions \cite[Eqns. (8) and (9)]{Dabora:101}.
        This is because the signal received from the relay should not increase too much the noise level when decoding the interference first, and it also should not increase too much the rate of the desired information.
        This is needed in order to make sure that the unintended receiver can decode its interference based only on the cross-link signal component, while the desired message and the relay signal are treated as noises.
\end{remark}
\begin{figure}[h]
    \centering
    \includegraphics[scale=0.6]{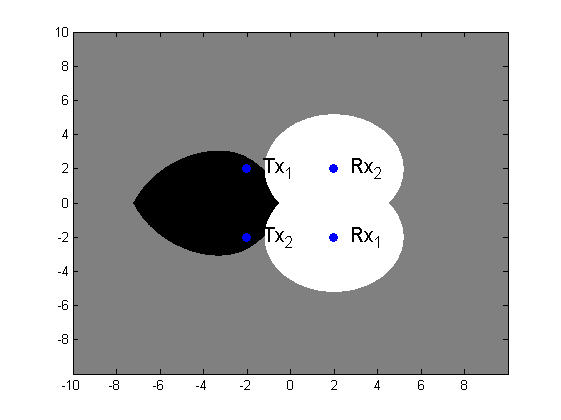}
\caption{\small The geographical position in the 2D-plane in which the VSI conditions hold for the ICR subject to
        phase fading. The black region shows the location of the relay for ICRs without feedback in which DF at the relay achieves capacity at VSI. The union of the black and gray regions shows the location of the
        relay for ICRFs in which DF at the relay achieves capacity at VSI. The scenario parameters are detailed in Comment \ref{rem:fig:Commnet_on_VSI-FULLNO}.}
\label{fig:VSI-FULLNO}
\end{figure}

\begin{remark}
\em{Note from \eqref{eq:ICRF_region} that in the VSI regime, the ICRF behaves like \em{two parallel relay channels}\em{}.}
\end{remark}

\begin{remark}
\em{Although in practice there is only one relay node, it is \em{simultaneously optimal}\em{} for both ``parallel relay channels" s.t. capacity is achieved in both
simultaneously. From a practical aspect, this observation gives a strong motivation to employ a combination of relaying and feedback in wireless networks since a relatively
small number of relay stations can optimally assist several nodes simultaneously.}
\end{remark}

\begin{remark}
\em{Note that since the capacity achieving channel inputs are mutually independent, adding relay nodes to the existing wireless networks \em{does not require any modifications}\em{} in the
transmitters codebooks. Hence, these techniques (relaying with feedback) can be incorporated into current designs in a relatively simple manner.}
\end{remark}

\begin{remark}
\label{rem:Impact_of_no_CSI_in FB}
\em{}
    We now discuss the implication of having feedback of only channel outputs without CSI.
    Recall that in Comment \ref{rem:Comparing_FB_with_NFB_VSI} it is noted that, since feedback includes Rx-CSI as well as channel outputs, then, when feedback from
    both destinations to the relay is available, we can  employ the DF scheme to obtain a characterization of the capacity region for any quality of links from the transmitters to the relay.
    When feedback does not include CSI from the receivers, then decoding the sources' messages at the relay leads to additional restrictions on the channel coefficients,
    which are needed in order to arrive to a capacity characterization using DF.
    These restrictions decrease the set of channel coefficients for which the capacity region of the ICRF is achieved by the DF scheme.
    It should be emphasized that when the channel coefficients satisfy the additional restrictions, then the SI/VSI conditions are the same as those obtained with feedback that includes both Rx-CSI as well as channel output,
    and so are the rate constraints.
\end{remark}

\begin{remark}
\em{}Note that the capacity result in Theorem \ref{thm:ICRF VS} holds also when there is no independent receiver at the relay, i.e., when $Y_3 = \varnothing$,
in the VSI regime. This observation holds only in scenarios where there are two noiseless feedback links, one from each receiver to the relay and not with
partial feedback at the relay which will be studied in section \ref{sec:Partial FB}. This is because the feedback turns each component relay channel into a degraded channel in the sense of \cite{CoverElGamal:IT79}.
In the next sections \ref{sec:FB to tx}, \ref{sec:Partial FB}, where we consider feedback to the transmitters and partial feedback, degradedness does not occur.
\end{remark}


\section {ICRFs in the Strong Interference Regime}
\label{sec:FULL FB SI}
In this section, we characterize the capacity region of ICRFs in the SI regime. We consider two noiseless feedback links, one from each receiver to the relay. This capacity region is characterized in the following theorem:

\begin{theorem}
\label{thm:ICRF SI}
                Consider the fading ICRF with Rx-CSI. Assume that the channel coefficients are independent in time and independent of each other s.t. their phases are i.i.d. and distributed uniformly over $[0,2\pi)$. Let the additive noises be i.i.d. circularly symmetric complex Normal processes, $\CN(0,1)$, and let the sources have power constraints $\E\big\{|X_k|^2\big\} \le P_k$, $k\in\{1,2,3\}$. Assume noiseless feedback links from both receivers to the relay. If
                \begin{subequations}
                \label{eq:SI Conditions}
                \begin{eqnarray}
                \label{eq:first con in SI con}
                I(X_1,X_3;Y_1|X_2,\tH_1) &\le& I(X_1;Y_2|X_2,\tH_2)\\
                I(X_2,X_3;Y_2|X_1,\tH_2) &\le& I(X_2;Y_1|X_1,\tH_1),
                \end{eqnarray}
                \end{subequations}
                where the mutual information expressions are evaluated with $X_k\sim \CN(0,P_k), k\in\{1,2,3\}$, mutually independent,
                then the capacity region is given by all the nonnegative rate pairs s.t.
                \begin{subequations}
                \label{eq:SI_region}
                \begin{eqnarray}
                R_1 &\le& I(X_1,X_3;Y_1|X_2,\tH_1)\\
                R_2 &\le& I(X_2,X_3;Y_2|X_1,\tH_2)\\
                R_1+R_2&\le& \min \big\{I(X_1,X_2,X_3;Y_1|\tH_1),\nonumber\\
                & & \qquad \;\;\;\;\;I(X_1,X_2,X_3;Y_2|\tH_2)\big\},
                \end{eqnarray}
                \end{subequations}
                and it is achieved with $X_k\sim \CN(0,P_k), k\in\{1,2,3\}$, mutually independent and with DF strategy at the relay.
\end{theorem}
\subsection{Proof of Theorem \ref{thm:ICRF SI}}
The proof consists of the following steps:
\begin{itemize}
    \item From the ICRF we obtain the enhanced MARC (EMARC) as a MARC whose message destination is one of the destinations of the ICRF, but the relay receives feedback from both receivers.
    Therefore, EMARC$_1$ is defined by equations \eqref{eqn:ICR_model} and its receiver is Rx$_1$ and EMARC$_2$ is defined by equations \eqref{eqn:ICR_model} and its
    receiver is Rx$_2$.\footnote{
            Note that this definition is different from the usual definition of MARC, since in the present scenario feedback comes from both receivers but only one receiver
            is decoding. Thus, for each EMARC $(y_{1,1}^{i-1},y_{2,1}^{i-1},\hvec_{1,1}^{i-1},\hvec_{2,1}^{i-1})$ denotes the available feedback at the relay at time $i$ prior to the transmission of the $i$'th symbol.
        }
    \item We derive the capacity region of EMARC$_1$ and EMARC$_2$.
    \item We show that the \em{same} \em{coding strategy at the sources and at the relay achieves capacity for both EMARCs} \em{simultaneously}\em{.}
    \begin{itemize}
        \item We therefore obtain an achievable rate region for the ICRF as the intersection of capacity regions of EMARC$_1$ and EMARC$_2$.
    \end{itemize}
    \item We show that in the SI regime the intersection of the capacity regions of EMARC$_1$ and EMARC$_2$ contains the capacity region of the ICRF .
    \item We characterize the SI conditions for the ICRF.
    \item We conclude the capacity region of ICRF in the SI regime is equal to the intersection of the capacity regions of EMARC$_1$ and EMARC$_2$.
\end{itemize}
The first three steps are detailed in section \ref{sec: Full FB SI Achievable region} and the last three steps are detailed in section \ref{sec: SI, CONVERSE full FB}.
\subsubsection{An Achievable Rate Region}
\label{sec: Full FB SI Achievable region}
Define $\uavec \triangleq(a_{11},a_{21},a_{31},a_{12},a_{22},a_{32},a_{13},a_{23})$. Let $y_{1,1}^{i-1},y_{2,1}^{i-1},\hvec_{1,1}^{i-1},\hvec_{2,1}^{i-1}$ denote the available feedback at the relay at time
    $i$ in EMARC$_1$ and  EMARC$_2$ and let $\mathcal{C}_{\scriptsize{\mbox{EMARC$_1$}}}(\uavec)$ and $\mathcal{C}_{\scriptsize{\mbox{EMARC$_2$}}}(\uavec)$ denote their capacity region,
    respectively. Let $t_k(R_1,R_2)$ denote the coding strategy (codebooks, encoders and decoders) for EMARC$_k$ that achieves rate pair
    $(R_1,R_2)$. The capacity regions of the EMARCs are shown in Appendix \ref{App:MARC+feedback - capacity region} to be:
\begin{subequations}
    \label{eq:EMARC1}
    \begin{eqnarray}
     & & \!\!\!\!\!\!\!\!\!\!\!\!   \mathcal{C}_{\scriptsize{\mbox{EMARC$_1$}}}(\uavec)\nonumber\\
     & & \!\!\!\! \!\!\!\! =  \bigg\{ (R_1, R_2)\in \Rset^2_+:\nonumber\\
        \label{eq:EMARC1 capacity R1}
     & & \;\;  R_1 \le \min\big\{I(X_1;Y_1,Y_2,Y_3|X_2,X_3,\utH),\nonumber\\
     &  & \qquad \qquad \qquad I(X_1,X_3;Y_1|X_2,\tH_1)\big\}\\
        \label{eq:EMARC1 capacity R2}
     & & \;\; R_2 \le \min\big\{I(X_2;Y_1,Y_2,Y_3|X_1,X_3,\utH),\nonumber\\
     & & \qquad \qquad \qquad I(X_2,X_3;Y_1|X_1,\tH_1)\big\}\\
     & &  \;\; R_1+R_2 \le \min\big\{I(X_1,X_2;Y_1,Y_2,Y_3|X_3,\utH),\nonumber\\
     &  & \qquad \qquad \qquad I(X_1,X_2,X_3;Y_1|\tH_1)\big\} \bigg\}
    \end{eqnarray}
\end{subequations}
\begin{subequations}
    \label{eq:EMARC2}
    \begin{eqnarray}
& & \!\!\!\!\!\!\!\!\!\!\!\!        \mathcal{C}_{\scriptsize{\mbox{EMARC$_2$}}}(\uavec)  \nonumber\\
     & & \!\!\!\! \!\!\!\! =  \bigg\{ (R_1, R_2)\in \Rset^2_+:\nonumber\\
        \label{eq:EMARC2 capacity R1}
     & & \;\;  R_1 \le  \min\big\{I(X_1;Y_1,Y_2,Y_3|X_2,X_3,\utH),\nonumber\\
     &  & \qquad \qquad \qquad I(X_1,X_3;Y_2|X_2,\tH_2)\big\}\\
        \label{eq:EMARC2 capacity R2}
     & & \;\; R_2 \le \min\big\{I(X_2;Y_1,Y_2,Y_3|X_1,X_3,\utH),\nonumber\\
     & & \qquad \qquad \qquad I(X_2,X_3;Y_2|X_1,\tH_2)\big\}\\
     & &  \;\; R_1+R_2 \le \min\big\{I(X_1,X_2;Y_1,Y_2,Y_3|X_3,\utH),\nonumber\\
     & & \qquad \qquad \qquad I(X_1,X_2,X_3;Y_2|\tH_2)\big\} \bigg\},
    \end{eqnarray}
\end{subequations}
where $X_k\sim \CN(0,P_k), k\in\{1,2,3\}$ are mutually independent and DF is used at the relay.

Next, we have the following proposition:
\begin{proposition}
\label{pro:same strategy}
The same coding strategy at the sources and at the relay achieves capacity for both EMARCs simultaneously, i.e.,
\begin{equation}
    \mathcal{C}_{\scriptsize{\mbox{EMARC$_1$}}}(\uavec) \cap \mathcal{C}_{\scriptsize{\mbox{EMARC$_2$}}}(\uavec) \subseteq \mathcal{C}_{ICRF}(\uavec).
\end{equation}
\end{proposition}
\begin{proof}
In Appendix \ref{App:MARC+feedback - capacity region} it is shown that the capacity region of each EMARC is achieved with DF strategy at the relay and codebooks generated according to independent circularly
symmetric complex Normal distribution at the sources and at the relay (the same distributions are used in both EMARCs). In both EMARCs, for all rate pairs $(R_1,R_2)$, the relay
codebook has $2^{n(R_1+R_2)}$ codewords generated i.i.d. according to $\CN(0,P_3)$, independent of the codewords at the sources. For all rate pairs $(R_1,R_2)$ the same scheme is used
at the relay in both EMARCs: at block $b$ the relay decodes the messages $(m_{1,b},m_{2,b})$ via a joint-typicality decoder using $\big(\yvec_1(b),\yvec_2(b),\yvec_3(b), \mathbf{\uth}(b)\big)$,
and transmits $\xvec_3(m_{1,b},m_{2,b})$. Thus, all rate pairs $(R_1,R_2)$ s.t. $(R_1,R_2)\in\mathcal{C}_{\scriptsize{\mbox{EMARC$_1$}}}(\uavec) \cap \mathcal{C}_{\scriptsize{\mbox{EMARC$_2$}}}(\uavec)$
are achieved at both EMARCs simultaneously with $t_1(R_1,R_2)=t_2(R_1,R_2)=t_{DF}(R_1,R_2)$, where $t_{DF}(R_1,R_2)$ is the coding strategy detailed in Appendix \ref{App:MARC+feedback - capacity region},
for achieving the rate pair $(R_1,R_2)$.
\end{proof}
From Proposition \ref{pro:same strategy} it follows that an achievable rate region for the ICRF, $\R_{\scriptsize{ICRF}}(\uavec,t_{DF})$ (here $t_{DF}$ should be understood as the DF strategy appropriate for
each rate pair in the achievable region, see  Appendix \ref{App:MARC+feedback - capacity region}), can be obtained by:
 \begin{eqnarray}
    \label{eqn:ach_region_ICRF}
    &  &\!\!\!\!\!\!\! \R_{\scriptsize{ICRF}}(\uavec,t_{DF})\nonumber\\
    & & \quad =\mathcal{C}_{\scriptsize{\mbox{EMARC$_1$}}}(\uavec) \cap \mathcal{C}_{\scriptsize{\mbox{EMARC$_2$}}}(\uavec)\subseteq\mathcal{C}_{\scriptsize{\mbox{ICRF}}}(\uavec),
 \end{eqnarray}
and it is achieved with $X_k\sim \CN(0,P_k), k\in\{1,2,3\}$, independent of each other and with DF strategy at the relay.
\subsubsection{Converse}
\label{sec: SI, CONVERSE full FB}
By definition of the SI regime, in this regime both receivers can decode both messages without reducing the capacity region, i.e., any achievable rate pair
$(R_1,R_2)\in\mathcal{C}_{\scriptsize{\mbox{ICRF}}}(\uavec)$ is also achievable in EMARC$_1$ and EMARC$_2$, hence
$\mathcal{C}_{\scriptsize{\mbox{ICRF}}}(\uavec)\subseteq \mathcal{C}_{\scriptsize{\mbox{EMARC$_1$}}}(\uavec)\cap\mathcal{C}_{\scriptsize{\mbox{EMARC$_2$}}}(\uavec)$.
Combined with equation \eqref{eqn:ach_region_ICRF} we conclude that in the SI regime
$\mathcal{C}_{\scriptsize{\mbox{ICRF}}}(\uavec)= \mathcal{C}_{\scriptsize{\mbox{EMARC$_1$}}}(\uavec)\cap\mathcal{C}_{\scriptsize{\mbox{EMARC$_2$}}}(\uavec)$.
Hence, the only problem left open is to determine the SI conditions for the ICRF. Note that from proposition \ref{pro:same strategy} we obtain
that $\mathcal{C}_{\scriptsize{\mbox{EMARC$_1$}}}(\uavec)$ and $\mathcal{C}_{\scriptsize{\mbox{EMARC$_2$}}}(\uavec)$ are achieved with $t_{DF}$, thus {for the rest of the proof we only consider mutually independent, circularly symmetric complex Normal channel inputs with zero mean}. The rest of the proof consists
of the following steps:
\begin{itemize}
    \item We assume an achievable rate pair $(R_1,R_2)$ in the ICRF.

    \item We characterize the maximal rate at which each receiver can decode its desired message. We conclude that this rate is achieved with independent channel inputs generated i.i.d. according to the circularly symmetric complex Normal distribution (see Theorem \ref{thm: capacity of EMARC} in Appendix \ref{App:MARC+feedback - capacity region}).

    \item We characterize the worst case conditions for each receiver to decode the interfering message.

    \item We derive the conditions for which decoding both messages at each receiver does not reduce the capacity region.
\end{itemize}

For the first two steps note that the maximal rates for decoding at the destinations are given by the cut-set bounds in \eqref{eq:cut-set bound}, i.e.,
\begin{subequations}
\begin{eqnarray}
    \mbox{max }R_1&=&\!\!\!\!\!\!\!\! \sup_{f(x_1,x_2,x_3)}\Big\{\min \big\{I(X_1;Y_1,Y_2,Y_3|X_2,X_3,\utH),\nonumber\\
    & & \qquad \qquad I(X_1,X_3;Y_1|X_2,\tH_1)\big\}\Big\}\\
    \mbox{max }R_2&=&\!\!\!\!\!\!\!\! \sup_{f(x_1,x_2,x_3)}\Big\{\min \big\{I(X_2;Y_1,Y_2,Y_3|X_1,X_3,\utH),\nonumber\\
    & & \qquad \qquad I(X_2,X_3;Y_2|X_1,\tH_2)\big\}\Big\},
\end{eqnarray}
\end{subequations}
and they are achieved with mutually independent, circularly symmetric complex Normal channel inputs with zero mean and with DF at the relay (see Appendix \ref{App:MARC+feedback - capacity region} for a detailed proof). The worst case scenario for decoding at the destinations, however, is when the signal from the relay degrades the performance of the receivers. Given two vectors, $\mathbf{a}$ and $\mathbf{b}$, define the notation $\mathbf{a}\cdot\mathbf{b}\triangleq(a_1b_1,a_2b_2,...,a_nb_n)$. Assume the rate pair $(R_1,R_2)\in\mathcal{C}_{\scriptsize{\mbox{ICRF}}}(\uavec)$. If Rx$_1$ can decode $m_1$ from the signal
\begin{equation*}
    \yvec_1=\hv_{11}\cdot\xvec_1+\hv_{21}\cdot\xvec_2+\hv_{31}\cdot\xvec_3+\zvec_1,
\end{equation*}
then it can create the signal
\begin{equation*}
    \yvec_1'=\hv_{21}\cdot\xvec_2+\hv_{31}\cdot\xvec_3+\zvec_1,
\end{equation*}
from which it can decode $m_2$ by treating $\hv_{31}\cdot\xvec_3$ as additive noise\footnote{Note that for this step we use the fact that the codebooks are generated independently, hence
the relay signal can be treated as additive noise.} if
\begin{equation*}
    R_2\le I(X_2;Y_1'|\tH_1)\triangleq R_2'.
\end{equation*}
Similarly, Rx$_2$ can decode $m_1$ if
\begin{equation*}
    R_1\le I(X_1;Y_2'|\tH_2)\triangleq R_1'.
\end{equation*}
Next, we should guarantee that decoding both messages at each receiver does not reduce the capacity region of the ICRF.
This is achieved if $\max R_1\le R_1'$
and $\max R_2\le R_2'$, i.e.,
\begin{eqnarray*}
  & &\!\!\! \!\!\!\!\!\!\!\!\!\!\!\!  \sup_{f(x_1,x_2,x_3)}\Big\{\min \big\{I(X_1;Y_1,Y_2,Y_3|X_2,X_3,\utH),\\
  & & \qquad \qquad\qquad \qquad I(X_1,X_3;Y_1|X_2,\tH_1)\big\}\Big\}\\
    & &            \le I(X_1;Y_2'|\tH_2)\\
    & & \stackrel{(a)}{=}I(X_1;Y_2|X_2,\tH_2),\\
 & &  \!\!\!\!\!\!\!\!\!\!\!\!\!\!\!  \sup_{f(x_1,x_2,x_3)}\Big\{\min \big\{I(X_2;Y_1,Y_2,Y_3|X_1,X_3,\utH),\\
    & & \qquad  \qquad\qquad\qquad    I(X_2,X_3;Y_2|X_1,\tH_2)\big\}\Big\} \\
    & & \le I(X_2;Y_1'|\tH_1)\\
    & & =I(X_2;Y_1|X_1,\tH_1),
\end{eqnarray*}
where (a) follows from the fact that the capacity region of the ICRF in the SI regime,  as well as the supremum on the left-hand side (l.h.s.) of the inequality, are achieved with
$X_k\sim \CN(0,P_k), k\in\{1,2,3\}$, \em{independent of each other and of the channel coefficients}\em{;}
thus all mutual information expressions are evaluated with the same distribution.
Note that from arguments similar to those used in section \ref{sec: Full FB CR for VSI}, we also obtain that
\begin{eqnarray*}
    I(X_1;Y_2|X_2,\tH_2)&\le& I(X_1;Y_1,Y_2,Y_3|X_2,X_3,\utH)\\
    I(X_2;Y_1|X_1,\tH_1)&\le& I(X_2;Y_1,Y_2,Y_3|X_1,X_3,\utH).
\end{eqnarray*}
Thus, to guarantee that $\max R_1\le R_1'$ and $\max R_2\le R_2'$, it is enough to require
\begin{subequations}
\label{eq: Full FB SI Filtered conditions}
\begin{eqnarray}
    I(X_1,X_3;Y_1|X_2,\tH_1) &\le& I(X_1;Y_2|X_2,\tH_2)\\
    I(X_2,X_3;Y_2|X_1,\tH_2) &\le& I(X_2;Y_1|X_1,\tH_1),
\end{eqnarray}
\end{subequations}
which give \eqref{eq:SI Conditions}. Hence, when \eqref{eq: Full FB SI Filtered conditions} holds,
$\mathcal{C}_{\scriptsize{\mbox{ICRF}}}(\uavec)= \mathcal{C}_{\scriptsize{\mbox{EMARC$_1$}}}(\uavec_1)\cap\mathcal{C}_{\scriptsize{\mbox{EMARC$_2$}}}(\uavec_2)$.
\begin{remark}
{\em
    Note that the argument presented here uses only local Rx-CSI, as opposed to the argument of Sato \cite{Sato:81}.
}
\end{remark}
\subsubsection{Simplification of the Capacity Region}
Consider the constraints on $R_1$ in \eqref{eq:EMARC1 capacity R1} and \eqref{eq:EMARC2 capacity R1}. Note that if \eqref{eq: Full FB SI Filtered conditions} holds,
since the channel inputs are independent of each other and of the channel coefficients, we get
\begin{eqnarray*}
     I(X_1,X_3;Y_1|X_2,\tH_1) &\le& I(X_1;Y_2|X_2,\tH_2)\\
    & \le & I(X_1,X_3;Y_2|X_2,\tH_2),\\
    I(X_1,X_3;Y_1|X_2,\tH_1) &\le& I(X_1;Y_2|X_2,\tH_2)\\
    & \le& I(X_1;Y_1,Y_2,Y_3|X_2,X_3,\utH).
\end{eqnarray*}
Thus, the constraints on $R_1$ in \eqref{eq:EMARC1 capacity R1} and \eqref{eq:EMARC2 capacity R1}  can be reduced to
\begin{subequations}
\begin{equation}
    R_1\le I(X_1,X_3;Y_1|X_2,\tH_1).
\end{equation}
Following the same steps, the constraints on $R_2$ in \eqref{eq:EMARC1 capacity R2}  and \eqref{eq:EMARC2 capacity R2}  can be reduced to
\begin{equation}
    R_2\le I(X_2,X_3;Y_2|X_1,\tH_2).
\end{equation}
\end{subequations}
Finally, note that since the channel inputs are independent of each other and of the channel coefficients, then
\begin{eqnarray*}
    I(X_1;Y_2|\tH_1)&\le& I(X_1;Y_2|X_2,\utH)\\
        & \le & I(X_1;Y_1,Y_2,Y_3|X_2,X_3,\utH)\\
    I(X_2;Y_1|\tH_1)&\le& I(X_2;Y_1,Y_2,Y_3|X_3,\utH).
\end{eqnarray*}
Hence, when \eqref{eq: Full FB SI Filtered conditions} is satisfied,
\begin{eqnarray*}
    & &\!\!\!\!\!\!\!\!\!\!\!\!\!\!\!\!\!\!\!\!\!\!\!\!\!I(X_1,X_2;Y_1,Y_2,Y_3|X_3,\utH)\\
    &=&I(X_2;Y_1,Y_2,Y_3|X_3,\utH)\\
    & &\qquad  +I(X_1;Y_1,Y_2,Y_3|X_2,X_3,\utH)\\
    &\ge&I(X_2;Y_1|\tH_1)+I(X_1,X_3;Y_1|X_2,\tH_1)\\
    &=&I(X_1,X_2,X_3;Y_1|\tH_1),
\end{eqnarray*}
and
\begin{equation*}
    I(X_1,X_2;Y_1,Y_2,Y_3|X_3,\utH)\ge I(X_1,X_2,X_3;Y_2|\tH_2),
\end{equation*}
implying that in the SI regime, the sum-rate conditions for decoding at the relay is always satisfied. This shows that when \eqref{eq:SI Conditions} holds,  the capacity region is characterized in \eqref{eq:SI_region}.
\vspace{-0.5cm}
\begin{flushright}
$\blacksquare$
\end{flushright}
\vspace{-0.5cm}
\subsection{Ergodic Phase Fading}
When the channel is subject to ergodic phase fading, we obtain the following explicit result:
\begin{corollary}
\label{Cor:SI Phase fading}
            Consider the phase fading ICRF with Rx-CSI and noiseless feedback links from both receivers to the relay, s.t. $y_{1,1}^{i-1},y_{2,1}^{i-1},\hvec_{1,1}^{i-1}$ and $\hvec_{2,1}^{i-1}$ are available at the relay at time $i$. If the channel coefficients satisfy
            \begin{subequations}
            \label{eq:SIPhaseFadingConditions}
            \begin{eqnarray}
            a_{11}^2P_1+a_{31}^2P_3 &\le& \frac{a_{12}^2P_1}{1+a_{32}^2P_3}\\
            a_{22}^2P_2+a_{32}^2P_3 &\le& \frac{a_{21}^2P_2}{1+a_{31}^2P_3},
            \end{eqnarray}
            \end{subequations}
            then the capacity region is characterized by all the nonnegative rate pairs s.t.
            \begin{subequations}
            \label{eq:SIPhaseFadingICRF_region}
            \begin{eqnarray}
            R_1 &\le& \log_2\big(1+a_{11}^2P_1+a_{31}^2P_3\big)\\
            R_2 &\le& \log_2\big(1+a_{22}^2P_2+a_{32}^2P_3\big)\\
            R_1+R_2&\le& \min\bigg\{\log_2(1+a_{11}^2P_1+a_{21}^2P_2+a_{31}^2P_3),\nonumber\\
            & &  \log_2(1+a_{12}^2P_1+a_{22}^2P_2+a_{32}^2P_3)\bigg\},
            \end{eqnarray}
            \end{subequations}
            and it is achieved with $X_k\sim \CN(0,P_k), k\in\{1,2,3\}$, mutually independent and with DF strategy at the relay.
\end{corollary}
\begin{proof}
The result follows from the expressions of Theorem \ref{thm:ICRF SI}. In order to obtain the conditions on the channel coefficients in \eqref{eq:SIPhaseFadingConditions} we first
evaluate $I(X_1,X_3;Y_1|X_2,\tH_1)$ and $I(X_2,X_3;Y_2|X_1,\tH_2)$ as in Corollary
\ref{Cor:ICRF VSI Phase fading}, by using the r.h.s. of \eqref{eq:eq6}. Also note that $\{H_{31,i}X_{3,i}\}_{i=1}^n$ can be considered as additive Gaussian
noise\footnote{Here we follow the same arguments as in Corollary \ref{Cor:ICRF VSI Phase fading}.} at Rx$_1$ and $\{H_{32,i}X_{3,i}\}_{i=1}^n$ can be considered as additive
Gaussian noise at Rx$_2$. Therefore from the independence of the channel inputs we obtain
    \begin{eqnarray*}
    I(X_2;Y_1|X_1,\tH_1)&=&\log_2\Big(1+\frac{a_{21}^2P_2}{1+a_{31}^2P_3}\Big)\\
    I(X_1;Y_2|X_2,\tH_2)&=&\log_2\Big(1+\frac{a_{12}^2P_1}{1+a_{32}^2P_3}\Big).
    \end{eqnarray*}
Thus, by evaluating \eqref{eq:SI Conditions} we obtain the conditions in \eqref{eq:SIPhaseFadingConditions}. Finally, we evaluate $I(X_1,X_2,X_3;Y_1|\tH_1)$ and $I(X_1,X_2,X_3;Y_2|\tH_2)$ by using the r.h.s. of \eqref{eq:eq12}:
    \begin{eqnarray*}
    & & I(X_1,X_2,X_3;Y_1|\tH_1)\\
    & & \qquad \qquad\qquad= \log_2(1+a_{11}^2P_1+a_{21}^2P_2+a_{31}^2P_3)\\
    & & I(X_1,X_2,X_3;Y_2|\tH_2)\\
    & & \qquad \qquad \qquad=\log_2(1+a_{12}^2P_1+a_{22}^2P_2+a_{32}^2P_3).
    \end{eqnarray*}
\end{proof}

\subsection{Ergodic Rayleigh Fading}
Define $\tU_k \triangleq(U_{kk},U_{3k}),k\in\{1,2\}$ and $\utU \triangleq(U_{11},U_{12},U_{21},U_{22},U_{31},U_{32})$. When the channel is subject to ergodic Rayleigh fading, we obtain the following explicit result:
\begin{corollary}
            Consider the Rayleigh fading ICRF with Rx-CSI and noiseless feedback links from both receivers to the relay, s.t. $y_{1,1}^{i-1},y_{2,1}^{i-1},\hvec_{1,1}^{i-1}$ and $\hvec_{2,1}^{i-1}$ are available at the relay at time $i$. If the channel coefficients satisfy
            \begin{subequations}
            \label{eq:SIRayleighFadingConditions}
            \begin{eqnarray}
            \!\!\!\!\!\!\!\!\!\!\!\!\frac{\frac{a_{12}^2P_1}{1+a_{32}^2P_3}}{e^{ \frac{1+ a_{32}^2P_3}{a_{12}^2P_1}}E_1\left(\frac{1+ a_{32}^2P_3}{ a_{12}^2P_1}\right)}&\ge&(1 + a_{11}^2 P_{1} + a_{31}^2 P_{3}  )\\
            \!\!\!\!\!\!\!\!\!\!\!\!\frac{\frac{a_{21}^2P_2}{1+a_{31}^2P_3}}{e^{ \frac{1+a_{31}^2P_3}{a_{21}^2P_2}}E_1\left(\frac{1+a_{31}^2P_3}{a_{21}^2  P_2}\right)}&\ge&(1 +a_{22}^2P_{2}+a_{32}^2P_{3}),
            \end{eqnarray}
            \end{subequations}
            then the capacity region is characterized by all the nonnegative rate pairs s.t.
            \begin{subequations}
            \label{eq:SIRayleighFadingICRF_region}
            \begin{eqnarray}
            & &\!\!\!\!\!\!\!\!\!\!\!\!\!\!\!\! R_1    \le \E_{\tU_1}\! \big\{\log_2(1+a_{11}^2|U_{11}|^2P_1+a_{31}^2|U_{31}|^2P_3)\big\}\\
            & &\!\!\!\!\!\!\!\!\!\!\!\!\!\!\!\! R_2    \le \E_{\tU_2}\! \big\{\log_2(1+a_{22}^2|U_{22}|^2P_2+a_{32}^2|U_{32}|^2P_3)\big\}\\
            & &\!\!\!\!\! \!\!\!\!\!\!\!\!\!\!\! R_1+R_2 \le  \min_{k\in\{1,2\}}\Big\{\E_{\utU}\!\big\{\log_2(1+a_{1k}^2|U_{1k}|^2P_1\nonumber\\
            & & \qquad \qquad +a_{2k}^2|U_{2k}|^2P_2+a_{3k}^2|U_{3k}|^2P_3)\big\}\Big\},
            \end{eqnarray}
            \end{subequations}
            and it is achieved with $X_k\sim \CN(0,P_k), k\in\{1,2,3\}$, mutually independent and with DF strategy at the relay.
\end{corollary}
\begin{proof}
The result follows from the expressions of Theorem \ref{thm:ICRF SI} and follows the same approach as in the proof of Corollary \ref{Cor:SI Phase fading}. Here \eqref{eq:SIRayleighFadingConditions} follows from \cite[Proposition 1]{Dabora:10}.
\end{proof}

\begin{figure}[t]
  \centering
  \subfloat[The range of $a_{12}$ and $a_{21}$ which satisfy the conditions for the VSI regime (gray area only) and for the SI regime (both black and gray areas) when DF achieves capacity in the phase fading scenario.]
  {\label{fig:fig2}\includegraphics[scale=0.5]{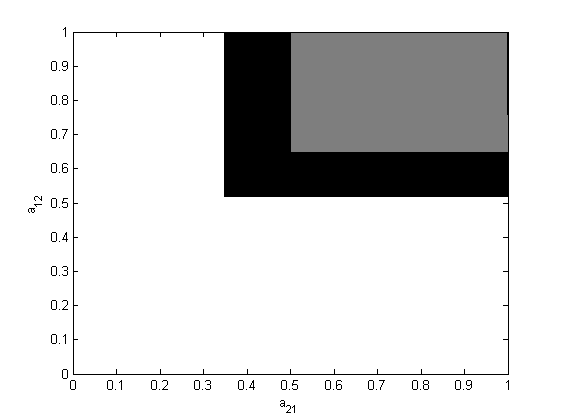}}
  \qquad
  \subfloat[The capacity region of the ICRF for the phase fading scenario in the VSI regime (both black and gray areas, $a_{12}=a_{21} = 0.7$) and the SI regime (black area only, $a_{12} = 0.53$, $a_{21} = 0.36$)] 
  {\label{fig:fig1}\includegraphics[scale=0.5]{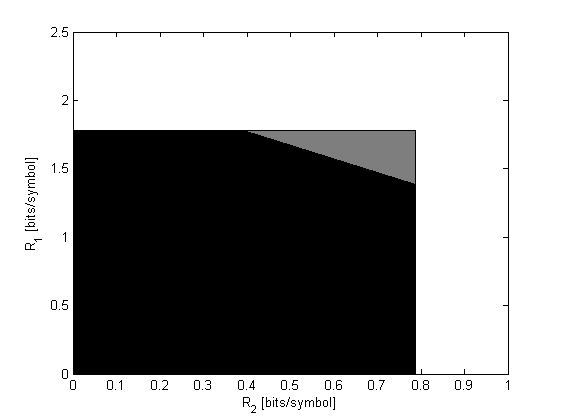}}
  \caption{SI and VSI for $P_1=P_2=P_3=10, a_{11}=0.42, a_{22}=0.25, a_{31}=0.26$ and $a_{32}=0.1.$}
  \label{fig:fig1+2}
\end{figure}
\subsection{Comments}
\begin{remark}
\label{rem:Compare_FB_with_NFB_SI}
\em{}

 In order to compare the feedback capacity region of Corollary \ref{Cor:SI Phase fading}  to that obtained without feedback, let $\mD_k$ be the set of channel coefficients
    $\avec_k =\big(a_{1k}, a_{2k}, a_{3k}, a_{13}, a_{23} \big) \in \setR^5_+$ that satisfy
    \begin{subequations}
          \begin{eqnarray}
              a_{1k}^2 P_1 + a_{3k}^2 P_3 & \le & a_{13}^2 P_1\\
              a_{2k}^2 P_2 + a_{3k}^2 P_3 & \le & a_{23}^2 P_2\\
              \label{eqn:SI_dec_relay_coeff_sum}
              a_{1k}^2 P_1 + a_{2k}^2 P_2 + a_{3k}^2 P_3 & \le & a_{13}^2 P_1 + a_{23}^2 P_2,
          \end{eqnarray}
    \end{subequations}
    $k = 1,2$. Let $\avec \triangleq \big(a_{11}, a_{12}, a_{13}, a_{21}, a_{22}, a_{23},  a_{31}, a_{32}\big) \in \setR^8_+$ and let
    $\avec \in \mD_1 \cap \mD_2$ be a short form notation to denote that $\avec_1$ and $\avec_2$
    satisfy  $\big\{  \avec_1 \cup \avec_2 = \avec \big\} \bigcap \big\{\avec_1 \in \mD_1\big\} \bigcap \big\{ \avec_2 \in \mD_2\big\} $.

    \cite[Theorem 2]{Dabora:101}  states that when the channel coefficients satisfy  \eqref{eq:SIPhaseFadingConditions} and also $\avec \in \mD_1 \cap \mD_2$, then the capacity region is
    given by \eqref{eq:SIPhaseFadingICRF_region}.  Similar to VSI (see Comment \ref{rem:Comparing_FB_with_NFB_VSI}),
    observe that the rate constraints are the same for both the feedback and the no-feedback cases, thus when capacity is achieved without feedback, then feedback does not enlarge the
    capacity region. Using similar arguments as in the discussion in Comment \ref{rem:Comparing_FB_with_NFB_VSI},
    it is possible to show that when \eqref{eq:SIPhaseFadingConditions} hold, then there are situations in which feedback enlarges
    the capacity region. The same conclusion applies to Rayleigh fading as well.
\end{remark}
\begin{remark}
\em{}Since the optimal codewords are generated independent of each other and of the channel coefficients, we obtain
\begin{eqnarray*}
    I(X_1;Y_2|\tH_2) &\le& I(X_1;Y_2|X_2,\tH_2)\\
    I(X_2;Y_1|\tH_1) &\le& I(X_2;Y_1|X_1,\tH_1),
\end{eqnarray*}
thus, its easy to see that the SI conditions in \eqref{eq:SI Conditions} are weaker than the VSI conditions in \eqref{eq:ICRF VSI Con}, as depicted in Fig. \ref{fig:fig2}. The capacity region
are compared in Fig. \ref{fig:fig1}.
\end{remark}
\begin{remark}
    \em{}Although in the SI regime the resulting model can be thought of as  a ``compound EMARC", it is important to note that both EMARCs share \em{the same relay}\em{} and thus they are not separate,
    contrary to ICs without relay. Note that the strategy at the relay is \em{optimal for both EMARCs}\em{} s.t. capacity is achieved for both \em{simultaneously}.
\end{remark}

\section{ICRs with Feedback to the Relay and Transmitters}
\label{sec:FB to tx}
In this section we study the scenarios in which feedback is available both at the relay and at the transmitters. We consider two configurations:
(1) feedback from each receiver to the relay and to its opposite transmitter, (2) feedback from each receiver to the relay and to its corresponding transmitter.

\subsection{Feedback to the Opposite Transmitters}
First, we study how the capacity region is affected if there are two noiseless feedback links from each receiver, both to the relay and to its opposite transmitter,
s.t. $y_{1,1}^{i-1},\hvec_{1,1}^{i-1}$ are available at Tx$_2$, $y_{2,1}^{i-1},\hvec_{2,1}^{i-1}$ are available at Tx$_1$, and $y_{1,1}^{i-1},y_{2,1}^{i-1},\hvec_{1,1}^{i-1},\hvec_{2,1}^{i-1}$ are
available at the relay at time $i$, prior to the transmission at each node. for this scenario, the definitions of the encoders at the transmitters
in Definition \ref{def:code} are modified as follows:
\begin{subequations}
\label{eq:source sig tx fb 1}
\begin{eqnarray}
    x_{1,i}&=&e_{1,i}(m_1, y_{2,1}^{i-1},\tilde{h}_{2,1}^{i-1})\\
    x_{2,i}&=&e_{2,i}(m_2, y_{1,1}^{i-1},\tilde{h}_{1,1}^{i-1});
\end{eqnarray}
\end{subequations}
the rest of the definitions remain unchanged and they are the same as in section \ref{sec:Model}. This model can represent scenarios where each transmitter is close to its opposite receiver, e.g., when VSI occurs in the ICRF. This configuration is depicted in Fig. \ref{fig:ICRF2}.
\begin{figure*}[!t]
    \centering
    \includegraphics[scale=0.35]{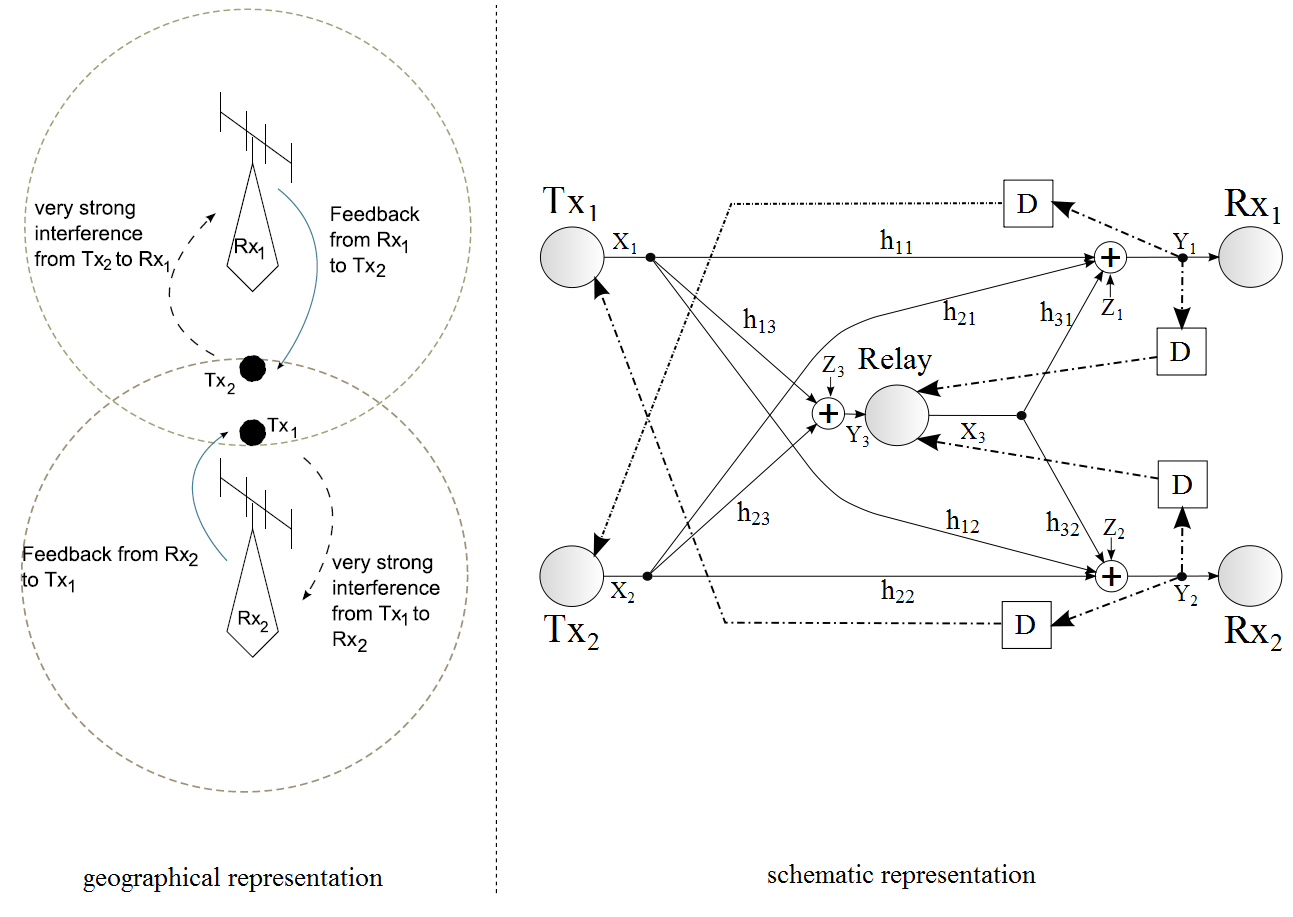}
\caption{The interference channel with a relay and feedback from both receivers to the relay and to their opposite transmitters.
    The `D' block represents a single-symbol delay.}
\label{fig:ICRF2}
\end{figure*}
\begin{proposition}
\label{thm:feedback doesnt improve capacity}
Consider the ICRF in which there is a noiseless feedback link from each receiver to the relay. Then, additional feedback links from each receiver to its opposite transmitter (see Fig. \ref{fig:ICRF2}), do not provide any further enlargement  to the capacity region in the VSI regime.
\end{proposition}
\begin{proof}
Let $m_1\in\M_1, m_2\in\M_2$ denote the messages that Tx$_1$ and Tx$_2$ send to Rx$_1$ and Rx$_2$, respectively. Let the encoders at Tx$_1$ and Tx$_2$ map their messages and the information
received from their feedback links into the channel input symbols $x_{1,i}$ and $x_{2,i}$, respectively. Thus, the encoders at the transmitters are given in \eqref{eq:source sig tx fb 1}.
The encoder at the relay remains unchanged, i.e., it is the causal function given in \eqref{eq:relay encoding}. Consider the cut-set bound expressions in \eqref{eq:cut-set bound}.
Observe that in the cut-set bounds on $R_1$, Rx$_1$ and Tx$_2$ belong to $\mathcal{S}^C$ while Tx$_1$ and Rx$_2$ belong to $\mathcal{S}$. Hence, by inspection of the proof of the cut-set bound
\cite[Theorem 15.10.1]{cover-thomas:it-book}, it is evident that the encoders at Tx$_1$ and Tx$_2$ used for the cut-set expressions are exactly those in \eqref{eq:source sig tx fb 1} and therefore the
cut-set expressions for rates $R_1$ and $R_2$ in \eqref{eq:cut-set bound} remain unchanged when feedback is also sent from each receiver to its opposite transmitter.

Finally, note that Theorem \ref{thm:ICRF VS} proves that in the VSI regime, if feedback from both receivers is available at the relay then the cut-set bounds \eqref{eq:cs bound R12} and
\eqref{eq:cs bound R22} are achievable and there is no constraint on the sum-rate. Hence, we conclude that when feedback from both receivers is available at the relay then additional
feedback links from each receiver to its opposite transmitter do not enlarge the capacity region of the ICRF in the VSI regime.
\end{proof}
\subsubsection{Comments}
\begin{remark}
\em{} In \cite{Wolf:75} it was shown that feedback can increase the capacity region of the discrete memoryless MAC by allowing the sources to coordinate their transmissions.
In the ICRF with additional feedback links from each receiver to its opposite transmitter, since the cut-set bound expressions are maximized with mutually independent channel
inputs then such coordination is not beneficial and in fact it is not possible.
\end{remark}
\begin{remark}
    \em{}We conclude that if, due network limitations, each receiver may send feedback either to the relay or to its opposite transmitter (when in the VSI regime \eqref{eq:ICRF VSI Con}),
    then its preferable to send feedback to the relay,
    since the relay can exploit the additional information to achieve the capacity in the VSI regime.
\end{remark}

\subsection{Feedback to the Corresponding Transmitters}
In this section, we study how the capacity region is affected if there are two noiseless feedback links from each receiver, both to the relay and to its corresponding transmitter, s.t.
$y_{1,1}^{i-1},\hvec_{1,1}^{i-1}$ are available at Tx$_1$, $y_{2,1}^{i-1},\hvec_{2,1}^{i-1}$ are available at Tx$_2$, and $y_{1,1}^{i-1},y_{2,1}^{i-1},\hvec_{1,1}^{i-1},\hvec_{2,1}^{i-1}$ are available at the relay,
at time $i$, prior to the transmission at each node. For this scenario the encoders at the transmitters in Definition \ref{def:code} are changed to
\begin{subequations}
\label{eq:source sig tx fb 2}
\begin{eqnarray}
    x_{1,i}&=&e_{1,i}(m_1, y_{1,1}^{i-1},\tilde{h}_{1,1}^{i-1})\\
    x_{2,i}&=&e_{2,i}(m_2, y_{2,1}^{i-1},\tilde{h}_{2,1}^{i-1}),
\end{eqnarray}
\end{subequations}
the rest of the definitions are the same as in section \ref{sec:Model}. This configuration is depicted in Fig. \ref{fig:ICRF3}.
\begin{figure}[h]
    \centering
    \includegraphics[scale=0.30]{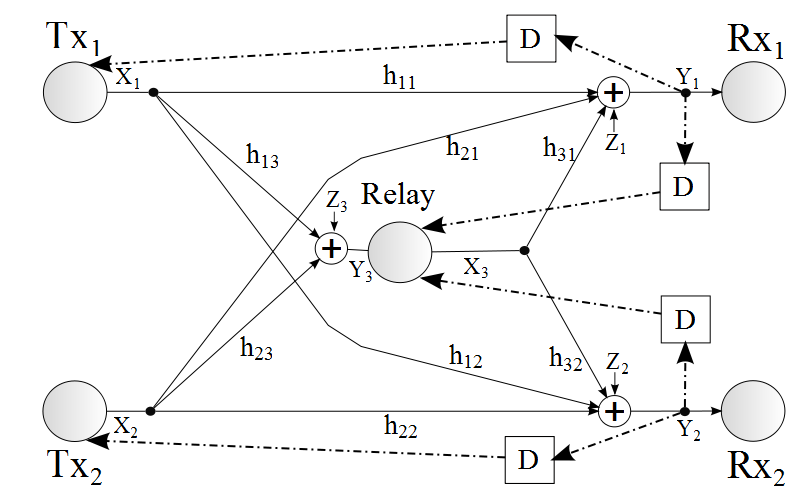}
\caption{The interference channel with a relay and feedback from both receivers to the relay and to their corresponding transmitters.
    The `D' block represents a single-symbol delay.}
\label{fig:ICRF3}
\end{figure}

    Let $\tmC_{MAC-FB}$ and $\tmR_{OB}$ be defined as
    \begin{subequations}
    \label{eqn:cap_region_mac_fb}
    \begin{eqnarray}
       & & \!\!\!\!\!\!\!\!\! \!\!\!\!\!\!\! \tmC_{MAC-FB} \nonumber \\
       & & \!\!\!\!\!\!\!\!\!\!\!\! \triangleq \bigg\{ (R_1, R_2) \in \setR_+^2:\nonumber\\
       & &   R_1  \le  I(X_1;Y_1,Y_2,Y_3|X_2, X_3, \utH) \\
       & &   R_2  \le  I(X_2;Y_1,Y_2,Y_3|X_1,  X_3, \utH)\\
       & & R_1+R_2  \le  I(X_1,X_2;Y_1,Y_2,Y_3|X_3, \utH)\; \bigg\},
    \end{eqnarray}
    \end{subequations}
    \begin{subequations}
    \label{eqn:outer_bound_2}
    \begin{eqnarray}
       & & \!\!\!\!\!\!\!\!\! \!\!\!\!\!\!\!  \tmR_{OB} \nonumber\\
       & & \!\!\!\!\!\!\!\!\!\!\!\!\triangleq \bigg\{ (R_1, R_2) \in \setR_+^2:\nonumber\\
       & &  R_1  \le  I(X_1,X_3;Y_1,Y_2|X_2, \utH) \\
       & &  R_2  \le  I(X_2,X_3;Y_1,Y_2|X_1, \utH)\\
       & & R_1+R_2  \le  I(X_1,X_2,X_3;Y_1,Y_2| \utH)\quad \bigg\},
    \end{eqnarray}
    \end{subequations}
    where all mutual information expressions in \eqref{eqn:cap_region_mac_fb} and \eqref{eqn:outer_bound_2}  are evaluated with $X_k \sim \CN(0,P_k)$, $k=1,2,3$, mutually independent.
    Next, define the region $\tmR_{ICRF}^{VSI}$ as follows:
    \begin{eqnarray*}
      & & \!\!\!\!\!\!\!\!\! \!\!\!\!\!\!\!   \tmR_{ICRF}^{VSI} \\
      & & \!\!\!\!\!\!\!\!\!\!\!\! \!\!\!\!\!\!\!\triangleq \bigg\{ (R_1,R_2) \in \setR_+^2:\\
      & & R_2  \le  I(X_2,X_3;Y_2|X_1,\tH_2)\\
      & & R_1  \le  I(X_1,X_2,X_3;Y_1|\tH_1) \\
      & & \qquad \qquad - \frac{I(X_2;Y_1|\tH_1)}{I(X_2,X_3;Y_2|X_1,\tH_2)}R_2\bigg\},
    \end{eqnarray*}
     where all mutual information expression are evaluated with $X_k \sim \CN(0,P_k)$, $k=1,2,3$, mutually independent. We now state the inner and outer bound in the following proposition:

    \smallskip

    \begin{proposition}
    \label{prop:capacity_increased_VSI_fb_corrs_Txs}
    The capacity region of the ICRF with  noiseless feedback links from each receiver to the relay and to its corresponding
    transmitter, denoted $\tmC_{ICRF}$, is outer bounded by $\tmC_{ICRF} \subseteq \tmC_{MAC-FB}\cap\tmR_{OB}$.
    Furthermore, if the VSI conditions (3) hold and also
    \begin{equation}
    \label{eq:VSI hardened con with feedback to own transmitter}
       I(X_1,X_2,X_3;Y_1|\tH_1)\le I(X_1;Y_2|X_2,X_3,\tH_2),
    \end{equation}
    holds for mutually independent Gaussian inputs, $X_k \sim \CN(0,P_k)$, $k=1,2,3$.
    Then, the corresponding capacity region of the ICRF in the VSI regime, denoted $\tmC_{ICRF}^{VSI}$, satisfies $\tmC_{ICRF}^{VSI} \supseteq \tmR_{ICRF}^{VSI}$.
    \end{proposition}

    \begin{proof}
    See Appendix \ref{Appndx:proof_prop_feedback_enlarges}.
    \end{proof}

    As a direct consequence of Proposition \ref{prop:capacity_increased_VSI_fb_corrs_Txs} we have the following corollary:

\begin{corollary}
\label{thm:feedback does improve capacity}
Consider the ICRF with two noiseless feedback links from the receivers to the relay. Then, additional feedback links from each receiver to its corresponding transmitter (see Fig. \ref{fig:ICRF3}), enlarge the capacity
region in the SI and VSI regimes.
\end{corollary}

\subsubsection{Comments}

    \begin{remark}
    \em{}
        Note that feedback to the corresponding transmitters also increases the capacity region of the ICRF in the SI regime. The proof is identical
        to the one used in Proposition \ref{prop:capacity_increased_VSI_fb_corrs_Txs} subject to \eqref{eq:VSI hardened con with feedback to own transmitter} and
        conditions \eqref{eq:SI Conditions}. In particular, the outer bound is identical to that in the proof of
        Proposition \ref{prop:capacity_increased_VSI_fb_corrs_Txs}, and the achievable rate region is obtained by time sharing between
        $\big(I(X_1,X_2,X_3;Y_1|\tH_1),0\big)$ and the rate points of the SI region \eqref{eq:SI_region}.
    \end{remark}

    \begin{remark}
        \em{} Note that in the coding scheme described in Proposition \ref{prop:capacity_increased_VSI_fb_corrs_Txs}, Tx$_2$ behaves like a second relay node for Tx$_1$, i.e., the ICRF is transformed into a multiple relay channel.
        It should also be noted that when $\Tbad$ cooperates with $\Tgood$ and with the relay in sending $m_1$, this decreases the maximal rate of information that could be sent from $\Tbad$ to $\Rbad$.
        However, as this cooperation increases the maximal achievable rate from  $\Tgood$ to $\Rgood$, compared to the case where feedback is available only at the relay, the capacity region is increased.
    \end{remark}

    \begin{remark}
        \em{} Recall that in the classic relay channel Rx-Tx feedback does not enlarge the capacity region once feedback from the receiver is available at the relay node. In ICRF, in contrary to the classic relay
        channel, Rx-Tx feedback can enlarge the capacity region beyond what is achieved with Rx-relay feedback. Thus, not all of the insights from the study of the classic relay channel hold for the ICRF.
    \end{remark}

    \begin{remark}
    \label{rem:Tightness_Outer_Bound}
    \em{}
        The boundaries of $\tmR_{ICRF}^{VSI}$ and  $\tmC_{MAC-FB}\cap\tmR_{OB}$, together with the capacity region of the ICRF in the VSI regime are depicted in Figure \ref{fig:RxTx benefit}.
        Observe that adding feedback to corresponding transmitters increases the capacity region of the ICRF in the SI and the VSI regimes. Also observe  from the figure  that the rate point
        $\big(I(X_1,X_2,X_3;Y_1|\tH_1), I(X_1,X_2,X_3;Y_2|\tH_2) \big)$ is outside the outer bound. This shows that the outer bound is not trivial.
        Since $\tmR_{OB}\nsubseteq\tmC_{MAC-FB}$ and $\tmC_{MAC-FB}\nsubseteq\tmR_{OB}$, then both regions are needed in the outer bound.

        Next, we note that when \eqref{eq:VSI hardened con with feedback to own transmitter} holds for mutually independent Gaussian inputs,
        then the achievable rate pair $(R_{1,B}, R_{2,B}) = \big(I(X_1,X_2,X_3;Y_1|\tH_1),0\big)$  is clearly on the boundary of the capacity region of the ICRF with additional
        feedback links from each receiver to its corresponding transmitter. Therefore, for this rate pair our achievability scheme is tight.
        We note that in all expressions in the outer bound, both signals $(Y_1,Y_2)$ appear together. Therefore, we do not expect the outer bound to be tight. However, the outer bound is not trivial
        as it {\em excludes the rate point $(R_1,R_2) = (I(X_1,X_2,X_3;Y_1|\tH_1), I(X_1,X_2,X_3;Y_2|\tH_2))$}.

    \end{remark}

    \begin{remark}
    \label{rem:Not_Possible_Cut-Set}
    \em{}
        We note that it is not possible to apply  directly the cut-set bound  \cite[Thm. 15.10.1]{cover-thomas:it-book} to the present case. To demonstrate this, consider the rate from $\Tgood$ to
        $\Rgood$. To obtain the corresponding bound using the cut-set theorem one should assign $\Tgood$ and $\Rbad$ to $\mS$ and $\Tbad$ and $\Rgood$ to $\mS^c$. Now, to generate $X_{\mS^c}$
        we need both $W_2$ and $Y_{2,1}^n$ (see, e.g. \cite[Eq. (15.330)]{cover-thomas:it-book}). But as $\Rbad \in \mS$, then this is not a valid assignment.
        In order to handle feedback to corresponding transmitters, we treat $(Y_1,Y_2)$ as a single MIMO receiver when deriving $\tmR_{OB}$.
    \end{remark}

    \begin{remark}
        \em{} In \cite{Xie:04} and \cite{Xie:05}, Xie and Kumar derived achievable rates for relay channels with $k$ different relay nodes where $B-k$ messages are sent in $B$ transmission blocks. Xie and Kumar proposed
        a scheme where the $l$'th relay  node transmits only after the transmission of the source and the first $l-1$ relays are finished. Note that in general the coding scheme proposed in \cite{Xie:04} and \cite{Xie:05}
        achieves higher rates for the relay channels, however, in the SI and VSI regimes as defined in Theorems \ref{thm:ICRF VS} and \ref{thm:ICRF SI}, there is no such improvement.
    \end{remark}

\begin{figure}[t]
    \centering
     \centering
    \includegraphics[scale=0.62]{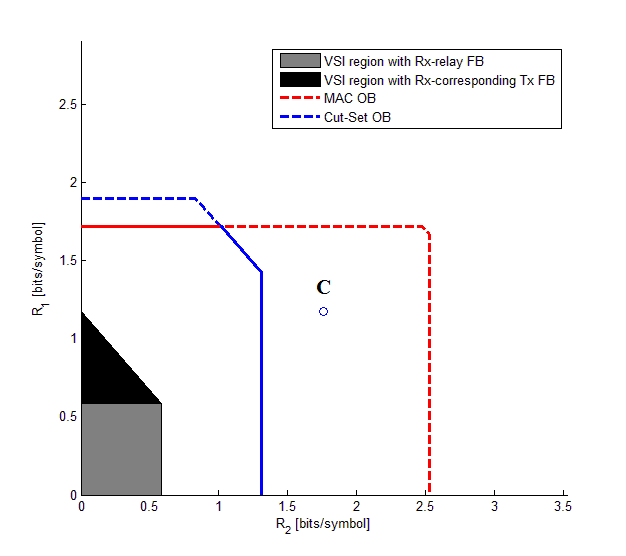}
\caption{The capacity region of the ICRF in the VSI regime with Rx-relay feedback and without Rx-Tx feedback (the gray area) and the achievable region
of Proposition \ref{prop:capacity_increased_VSI_fb_corrs_Txs} (gray and black area), when the channel is subject to phase fading
and $P_1=P_2=P_3=10$, $a_{11}=0.2$, $a_{12}=0.44$, $a_{22}=0.2$, $a_{21}=0.27$ ,
$a_{13} = 0.01$, $a_{23} = 0.6$, $a_{31}=0.1$, and $a_{32}=0.1$. Point C corresponds to the rate pair $(R_1,R_2) = (I(X_1,X_2,X_3;Y_1|\tH_1), I(X_1,X_2,X_3;Y_2|\tH_2))$.
Observe that this rate pair is outside the outer bound. }
\label{fig:RxTx benefit}
\end{figure}

    \begin{remark}
    \em{}
    Recall that in \cite{Cadambe:08}  it was shown that feedback can provide an unbounded gain as the SNR and INR increase to infinity. In \cite{Tse:09} it was shown that an unbounded capacity
    gain can be obtained also for the weak interference regime. We note that these results deal with the degrees of freedom of the channel, thus the conclusion
    holds only when the SNR and INR increase to infinity. As to the present case, we show in Proposition \ref{prop:capacity_increased_VSI_fb_corrs_Txs}
    that the rate pair $(R_1^*, R_2^*) = (I(X_1,X_2,X_3;Y_1|\tH_1),0)$
    is achievable when $I(X_1,X_2,X_3;Y_1|\tH_1) \le I(X_1;Y_2|X_2,X_3,\tH_2)$ holds. In the following we show that this implies an unbounded capacity gain over the no-feedback case for Rayleigh fading on the VSI regime.

    \smallskip

    Let $P_1 = P_2 = P_3 = SNR$, Let $a_{11}$, $a_{13}$, $a_{22}$, $a_{23}$, $a_{31}$, $a_{32}$ be constants, and let $a_{12} = \alpha_{12}SNR^{\frac{b-1}{2}}$, $a_{21} = \alpha_{21}SNR^{\frac{a-1}{2}}$.
    Note that under these definitions
    \begin{eqnarray*}
        INR_{12} & = & a_{12}^2P_1 = \alpha_{12}^2SNR^{b-1}SNR = \alpha_{12}SNR^b\\
        INR_{21} & = & a_{21}P_2 = \alpha_{21}SNR^a.
    \end{eqnarray*}
    Now consider the condition $I(X_1,X_2,X_3;Y_1|\tH_1) \le I(X_1;Y_2|X_2,X_3,\tH_2)$: Using the above definitions we obtain
    \begin{eqnarray*}
     & & \!\!\!\!\! \dsE_{\tU_1}\Big\{\log_2\big(1+ a_{11}^2|U_{11}|^2SNR \\
     & & \quad + \alpha_{21}^2 |U_{21}|^2SNR^a + a_{31}^2|U_{31}|^2SNR\big)\Big\}\\
                    & & \qquad \qquad \quad \le  \dsE_{\tU_2}\left\{\log_2\left(1 + \alpha_{12}^2|U_{12}|^2SNR^b\right)\right\}.
    \end{eqnarray*}
    Taking $SNR\rightarrow\infty$ and restricting $a>1$ and $b>1$ we arrive to the equivalent relationship
    \begin{eqnarray*}
   & &  \dsE_{\tU_1}\Big\{\log_2\big(\alpha_{21}^2 |U_{21}|^2SNR^a\big)\Big\} \\
   & & \qquad    \le \dsE_{\tU_2}\left\{\log_2\big(\alpha_{12}^2|U_{12}|^2SNR^b\big) \right\},
    \end{eqnarray*}
    which requires  $b\ge a$.  When this holds, the asymptotic sum-rate (we consider only the maximal $R_1$  when $R_2 =0$) is given by
    \begin{eqnarray*}
       & &  \!\! \!\!\!\!\!\!\!\!\!\!\!\!\!\!\!\!\!\!\!\!\!\!\!\!\!\! \mC_{\mbox{\scriptsize sum}}^{\mbox{\scriptsize FB-Tx+relay}}(R_2=0,SNR,a,b)\\
        & = & I(X_1,X_2,X_3;Y_1|\tH_1)\\
        & = & \dsE_{\tU_1}\Big\{\log_2\big(1+ a_{11}^2|U_{11}|^2SNR\\
        & & \quad + \alpha_{21}^2 |U_{21}|^2SNR^a + a_{31}^2|U_{31}|^2SNR\big)\Big\}\\
        &\stackrel{SNR\rightarrow\infty}{\longrightarrow}& \dsE_{\tU_1}\Big\{\log_2\big(\alpha_{21}^2 |U_{21}|^2SNR^a\big)\Big\} + O(1)\\
        & = & a\log_2SNR + O(1),
    \end{eqnarray*}
    where $O(1)$ means that for some $SNR$ large enough, the term is bounded by a constant, see, e.g. \cite{Zheng:03}.
    Next, we consider the sum-rate with feedback only to the relay, starting with the VSI regime.
    Recall that the same sum-rate is achieved without feedback when relay reception is good in the sense that the channel coefficients satisfy \cite[Eqns. (8)]{Dabora:101}.
    Consider first the VSI condition \eqref{eq:first con in VSI con2}:
    $I(X_2,X_3;Y_2|X_1,\tH_2) \le I(X_2;Y_1|\tH_1)$. Writing this explicitly we obtain
    \begin{eqnarray*}
       &  & \!\!\!\! \dsE_{\tU_1}\!\!\left\{\log_2\left(\!\!1 + \frac{\alpha_{21}^2|U_{21}|^2SNR^a}{1 + a_{11}^2|U_{11}|^2SNR + a_{31}^2|U_{31}|^2SNR}\right)\!\!\right\}\\
       &  &       \ge \dsE_{\tU_2}\!\Big\{\!\! \log_2\big(1 + a_{22}^2|U_{22}|^2SNR +  a_{32}^2|U_{32}|^2SNR\big)\Big\},
    \end{eqnarray*}
    which, as $SNR\rightarrow\infty$, becomes
    \begin{eqnarray*}
       &  &\!\!\!\! \dsE_{\tU_2}\Big\{\log_2\big( (a_{22}^2|U_{22}|^2 + a_{32}^2|U_{32}|^2)SNR\big)\Big\} \\
       &  &\quad   \le \dsE_{\tU_1}\left\{\log_2\left( \alpha_{21}^2|U_{21}|^2SNR^a\right)\right\}\\
      & &\qquad \quad        - \dsE_{\tU_1}\left\{\log_2\left( (a_{11}^2|U_{11}|^2 + a_{31}^2|U_{31}|^2)SNR\right)\right\}.
    \end{eqnarray*}
    This inequality holds asymptotically when $a >2$. Similarly we can show that $I(X_1,X_3;Y_1|X_2,\tH_1) \le I(X_1;Y_2|\tH_2)$ holds for $SNR \rightarrow \infty$ when $b>2$.

    Recall that at asymptotically high  SNR and INR, the VSI regime is defined as $a \ge 2$ and $b\ge  2$ (see  \cite{Cadambe:08}, \cite{Tse:09}.
    We thus conclude that with feedback only at the relay, the maximal achievable sum-rate at asymptotically high SNR in the VSI regime is
    \begin{eqnarray*}
        & & \!\!\! \mC_{\mbox{\scriptsize sum}}^{\mbox{\scriptsize FB}}(SNR,a,b) \\
        & & =  I(X_1,X_3;Y_1|X_2,\tH_1) + I(X_2,X_3;Y_2|X_1,\tH_2)\\
        & & = \dsE_{\tU_1}\Big\{\log_2\big(1+ a_{11}^2|U_{11}|^2SNR+  a_{31}^2|U_{31}|^2SNR\big)\Big\} \\
        & &  \quad  + \dsE_{\tU_2}\!\Big\{\!\log_2\big(\!1+ a_{22}^2|U_{22}|^2SNR+  a_{32}^2|U_{32}|^2SNR\big)\!\Big\} \\
        & &\stackrel{SNR\rightarrow\infty}{\longrightarrow}  \dsE_{\tU_1}\Big\{ \log_2\big((a_{11}^2|U_{11}|^2 +  a_{31}^2|U_{31}|^2)SNR\big)\Big\} \\
        & & \quad\;\;  \qquad + \dsE_{\tU_2}\Big\{\log_2\big((a_{22}^2|U_{22}|^2 +  a_{32}^2|U_{32}|^2)SNR\big)\Big\} \\
        & & =  2 \log_2 SNR + O(1).
    \end{eqnarray*}
    Comparing the sum-capacity with and without  feedback  to the transmitters we observe that in VSI
    \[
        \frac{\mC_{\mbox{\scriptsize sum}}^{\mbox{\scriptsize FB-Tx+relay}}(R_2=0,SNR,a,b)}{\mC_{\mbox{\scriptsize sum}}^{\mbox{\scriptsize FB}}(SNR,a,b)} = \frac{a}{2}.
    \]
    We conclude that adding feedback links from each receiver to the its corresponding transmitter allows an unbounded rate gain in the VSI regime.
    This follows directly from our capacity results.

    \end{remark}


\section{ICRs with Partial Feedback at the Relay}
\label{sec:Partial FB}
In this section we study the scenarios in which only partial feedback is available at the relay. We consider the case where feedback is available only from Rx$_1$,
the case where feedback is available only from Rx$_2$ is symmetric. This scenario is described in Fig. \ref{fig:ICPRF}.
\begin{figure}[h]
    \centering
    \includegraphics[scale=0.35]{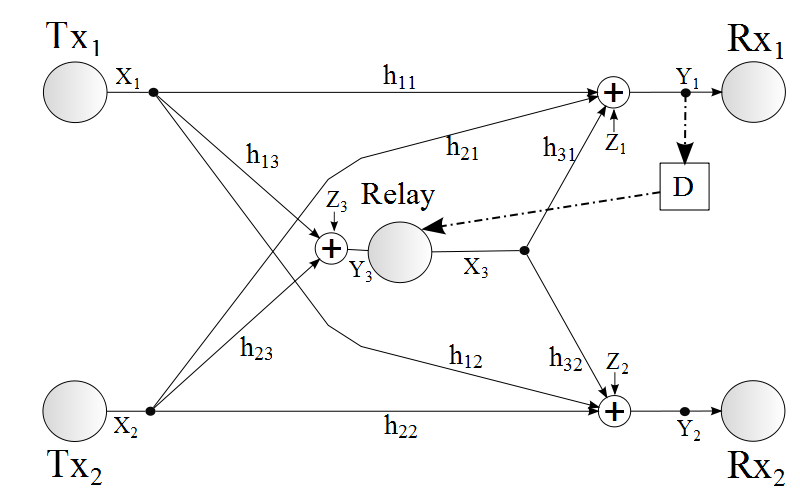}
\caption{The interference channel with a relay and feedback only from Rx$_1$ to the relay. The `D' block represents a single-symbol delay.}
\label{fig:ICPRF}
\end{figure}
\subsection{Partial Feedback in the Very Strong Interference Regime}
First, we characterize the capacity region of the ICRF in the VSI regime for the case where the relay receives feedback only from Rx$_1$, i.e, $y_{1,1}^{i-1}$ and $\hvec_{1,1}^{i-1}$
are available at the relay at time $i$ prior to transmission. In this scenario, the CSI at the relay is represented by $\utH=\big(H_{13}, H_{23},\tH_1\big)\in\setC^5 \triangleq \underline{\tilde{\mathfrak{H}}}$.
Thus, the encoder at the relay in \eqref{eq:relay encoding} in Definition \ref{def:code} is replaced by
\begin{equation}
\label{eq:relay encoding Partial FB}
x_{3,i} = t_i\big(y_{1,1}^{i-1},y_{3,1}^{i-1},h_{13,1}^{i-1},h_{23,1}^{i-1},\tilde{h}_{1,1}^{i-1}\big) \in \setC,
\end{equation}
$i = 1, 2, ..., n$. The other definitions remain unchanged, as described in section \ref{sec:Model}. Next, we have the following theorem:
\begin{theorem}
\label{thm:Partial FB VSI}
        Consider the fading ICRF with Rx-CSI. Assume that the channel coefficients are independent in time and independent of each other s.t. their phases are i.i.d. and distributed
        uniformly over $[0,2\pi)$. Let the additive noises be i.i.d. circularly symmetric complex Normal processes, $\CN(0,1)$, and let the sources have power constraints
        $\E\big\{|X_k|^2\big\} \le P_k$, $k\in\{1,2,3\}$. Assume that there is only one noiseless feedback link -- from Rx$_1$ to the relay (see Fig. \ref{fig:ICPRF}). If
        \begin{subequations}
        \label{eq:ICR partial FB VSI Con}
        \begin{eqnarray}
        I(X_1,X_3;Y_1|X_2,\tH_1) &\le& \min\big\{I(X_1;Y_1,Y_3|X_2,X_3,\utH),\nonumber\\
            & &  \qquad \quad I(X_1;Y_2|\tH_2)\big\}\\
        I(X_2,X_3;Y_2|X_1,\tH_2) &\le& I(X_2;Y_1|\tH_1),
        \end{eqnarray}
        \end{subequations}
        where the mutual information expressions are evaluated with $X_k\sim \CN(0,P_k), k\in\{1,2,3\}$, mutually independent,
        then the capacity region is given by all the nonnegative rate pairs s.t.
        \begin{subequations}
        \label{eq:ICR partial FB_region}
        \begin{eqnarray}
        &&R_1\le I(X_1,X_3;Y_1|X_2,\tH_1)\\
        &&R_2\le I(X_2,X_3;Y_2|X_1,\tH_2),
        \end{eqnarray}
        \end{subequations}
        and it is achieved with $X_k\sim \CN(0,P_k), k\in\{1,2,3\}$, mutually independent and with DF strategy at the relay.
\end{theorem}
\subsubsection{Proof of Theorem \ref{thm:Partial FB VSI}}
The proof consists of the following steps:
\begin{itemize}
    \item We obtain an outer bound on the capacity region using the cut-set bound.
    \item We show that the input distribution that maximizes the outer bound is zero-mean, circularly symmetric complex Normal with channel inputs independent of each other and with maximum allowed power.
    \item We derive an achievable rate region based on DF with partial feedback at the relay,  using codebooks generated according to mutually independent, zero-mean circularly symmetric complex Normal input distributions.
         \begin{itemize}
         \item We derive an achievable rate region for decoding at the relay using steps similar to \cite[Sec. 4.D]{Kramer:05}.
        \item We obtain an achievable rate region for decoding at the destinations by decoding the interference first, while treating the relay signal and the desired signal as additive i.i.d. noises, followed by using a backward decoding scheme for decoding the desired message.
        \end{itemize}
    \item We derive the VSI conditions which guarantee that decoding the interference first at each receiver, does not constrain the rate of the other pair.
    \item We obtain conditions on the channel coefficients that guarantee that the achievable rate region coincides with the cut-set bound and thus it is the capacity region of the ICR with partial feedback in the VSI regime.
\end{itemize}
We follow steps similar to the case in which full feedback is available at the relay, so we only provide a sketch of the proof. \paragraph{An Outer Bound}
An outer bound on the capacity region is given by the cut-set bound in \eqref{eq:cut-set bound}. Following similar steps as in
section \ref{sec:Full FB VSI upper bounds } and
Appendix \ref{app:maximzing dist for cut set bound}, we conclude that the outer bound is maximized by mutually independent, zero-mean, circularly symmetric complex Normal channel inputs with maximum allowed power.
\paragraph{An Achievable Rate Region}
The code construction and encoding process are similar to sections \ref{sec:ICRF Code Book} and \ref{sec:ICRF Encoding}. Hence, following similar steps as in \cite[Sec. 4.D]{Kramer:05}, we conclude that an achievable rate region for decoding at the relay is given by
\begin{subequations}
\label{eq:Partial FB VSI Decoing at Relay}
\begin{eqnarray}
& & \!\!\!\!\!\!\!\!  \!\!\!\!\!\! \!\!\!\!\!\!\R_{\mbox{\scriptsize{Relay Decoding}}}\nonumber\\
& &\!\!\!\!\!\!  \!\!\!\!\!\!   =  \bigg\{ (R_1, R_2)\in \Rset^2_+:\nonumber\\
& & R_1  \le  I(X_1;Y_1,Y_3|X_2,X_3,\utH)\\
& &    R_2  \le  I(X_2;Y_1,Y_3|X_1,X_3,\utH)\\
& &   R_1 + R_2  \le  I(X_1,X_2;Y_1,Y_3|X_3,\utH) \bigg\}.
\end{eqnarray}
\end{subequations}
At the destinations, Rx$_1$ can decode the interference if
    \begin{subequations}
    \label{eq:Partial FB VSI condition}
    \begin{equation}
    R_2 \le I(X_2;Y_1|\tH_1),
    \end{equation}
and Rx$_2$ can decode the interference if
    \begin{equation}
    R_1 \le I(X_1;Y_2|\tH_2).
    \end{equation}
    \end{subequations}
Thus, decoding the interference first, we obtain an achievable rate region for decoding at the destinations:
    \begin{subequations}
    \label{eq:Partial FB VSI Decoing at Destination}
    \begin{eqnarray}
& & \!\!\!\!\!\!\!\!\!\!  \!\!\!\!\! \!\!    \R'_{\mbox{\scriptsize{Destination Decoding}}} \nonumber\\
& & \!\!\! \!\!\!\!\!\!\!\! \!\! \!\!=  \bigg\{ (R_1, R_2)\in \Rset^2_+:\nonumber\\
& & \!\!\!\!\!\!\!\!\!\!\!\!\!\!\!\! R_1\! \le \! \min\big\{I(X_1,X_3;Y_1|X_2,\tH_1), I(X_1;Y_2|\tH_2)\big\}\\
& & \!\!\!\! \! \!\!\!\! \! \!\!\!\! \!\! R_2 \!\le \! \min\big\{I(X_2,X_3;Y_2|X_1,\tH_2),I(X_2;Y_1|\tH_1)\big\}\!\! \bigg\}.
    \end{eqnarray}
    \end{subequations}
Hence, an achievable rate region for the ICR with partial feedback is given by
\begin{equation}
\label{eq:VSI PFB achievable region }
    \R_{\mbox{\scriptsize{Achievable}}}=\R_{\mbox{\scriptsize{Relay Decoding}}} \cap \R'_{\mbox{\scriptsize{Destination Decoding}}}.
\end{equation}
\paragraph{The Capacity Region}
Next, we should guarantee that decoding the interference does not constrain the rates at the destinations, this is satisfied if
    \begin{subequations}
    \label{eq:Partial FB - hardened VSI conditons}
    \begin{eqnarray}
    & & \!\!\!\!\!\!\!\!\!\!\!\!\!\!\!\!\!\!\!\!\!\!\!\!\!\!\!\!\!\!\!\!\!\min\big\{I(X_1;Y_1,Y_3|X_2,X_3,\utH),I(X_1,X_3;Y_1|X_2,\tH_1)\big\} \nonumber\\
    \phantom{xxxxxxx}&\le& I(X_1;Y_2|\tH_2)\\
    \label{eq:eq18}
    & & \!\!\!\!\!\!\!\!\!\!\!\!\!\!\!\!\!\!\!\!\!\!\!\!\!\!\!\!\!\!\!\!\!\min\big\{I(X_2;Y_1,Y_3|X_1,X_3,\utH),I(X_2,X_3;Y_2|X_1,\tH_2)\big\}\nonumber\\
    \phantom{xxxxxxx} &\le& I(X_2;Y_1|\tH_1).
    \end{eqnarray}
    \end{subequations}
Note that as the channel inputs are mutually independent, we obtain that $I(X_2;Y_1|\tH_1)\le I(X_2;Y_1,Y_3|X_1,X_3,\utH)$. Hence, the conditions in \eqref{eq:Partial FB - hardened VSI conditons} can be reduced to
    \begin{subequations}
    \label{eq:Partial FB - hardened VSI conditons2}
    \begin{eqnarray}
   & &\!\!\!\!\!\!\!\!\!\!\!\!\!\!\!\!\!\!\!\!\!\!\!\!\!\!\!\!\!\!\!\!\!\!\!\!\!\!\!\!\!\!\!\!\!\!\!\!\!\!\!\!\!\!\!\!\!\!\!\!\!\!\!\!\min\big\{I(X_1;Y_1,Y_3|X_2,X_3,\utH),I(X_1,X_3;Y_1|X_2,\tH_1)\big\}\nonumber\\
    &\le& I(X_1;Y_2|\tH_2)\\
    \label{eq:eq182}
    I(X_2,X_3;Y_2|X_1,\tH_2) &\le& I(X_2;Y_1|\tH_1).
    \end{eqnarray}
    \end{subequations}

Finally, in order to achieve capacity, we should guarantee that,
whenever the destinations can reliably decode their messages, the
relay can decode both messages reliably. This can be done if
    \begin{subequations}
    \label{eq:RELAY hardened conditions VSI}
    \begin{eqnarray}
    \label{eq:eq19}
  & &\!\!\! I(X_1,X_3;Y_1|X_2,\tH_1) \le I(X_1;Y_1,Y_3|X_2,X_3,\utH)\phantom{xx}\\
    \label{eq:eq192}
  & & \!\!\! I(X_2,X_3;Y_2|X_1,\tH_2) \le I(X_2;Y_1,Y_3|X_1,X_3,\utH)\phantom{xx}\\
    \label{eq:eq193}
     & &\!\!\!  I(X_1,X_3;Y_1|X_2,\tH_1) \nonumber\\
     & & \!\!\! +I(X_2,X_3;Y_2|X_1,\tH_2)  \le  I(X_1,X_2;Y_1,Y_3|X_3,\utH)\phantom{xx}.
    \end{eqnarray}
    \end{subequations}
Recall that the channel inputs are independent, hence when \eqref{eq:eq182} holds, \eqref{eq:eq192} is always satisfied since
    \begin{eqnarray*}
    I(X_2,X_3;Y_2|X_1,\tH_2) & \le & I(X_2;Y_1|\tH_1)\\
    & \le & I(X_2;Y_1,Y_3|X_3,\utH)\\
    & \le & I(X_2;Y_1,Y_3|X_1,X_3,\utH).
    \end{eqnarray*}
Similar arguments show that when \eqref{eq:eq182} and \eqref{eq:eq19} hold, \eqref{eq:eq193} is always satisfied:
    \begin{eqnarray*}
  & &\!\! \!\!\!\!  I(X_1,X_2;Y_1,Y_3|X_3,\utH)\\
     & & = I(X_2;Y_1,Y_3|X_3,\utH)+I(X_1;Y_1,Y_3|X_2,X_3,\utH)\\
     & & \ge I(X_2,X_3;Y_2|X_1,\tH_2) +I(X_1,X_3;Y_1|X_2,\tH_1).
    \end{eqnarray*}
Therefore, if \eqref{eq:eq182} holds, \eqref{eq:eq19} is enough to guarantee reliable decoding at the relay (i.e., \eqref{eq:Partial FB VSI Decoing at Relay} is satisfied).

Finally note that by combining \eqref{eq:Partial FB - hardened VSI conditons2} with \eqref{eq:eq19} we obtain conditions which coincide with \eqref{eq:ICR partial FB VSI Con} and under these
conditions \eqref{eq:VSI PFB achievable region } specialize to \eqref{eq:ICR partial FB_region}. Comparing with the cut-set bound in \eqref{eq:cut-set bound},  we conclude that if
\eqref{eq:ICR partial FB VSI Con} holds, the achievable rate region  \eqref{eq:ICR partial FB_region}, coincides with the cut-set bounds and hence it is the capacity region.
\tend

\subsubsection{Ergodic Phase Fading}
When the channel is subject to ergodic phase fading, we obtain the following explicit result:
\begin{corollary}
\label{Cor:Partial FB VSI phase fading}
            Consider the phase fading ICR with Rx-CSI and partial feedback s.t. $y_{1,1}^{i-1}$ and $\hvec_{1,1}^{i-1}$ are available at the relay at time $i$. If the channel coefficients satisfy
            \begin{subequations}
            \label{eq:VSPhaseFadingICR Partial FB CONDITIONS}
            \begin{eqnarray}
            a_{11}^2P_1+a_{31}^2P_3 &\le& \min\Big\{\frac{a_{12}^2P_1}{1+a_{22}^2P_2+a_{32}^2P_3},\nonumber\\
            & & \qquad \qquad (a_{11}^2+a_{13}^2)P_1\Big\}\\
            a_{22}^2P_2+a_{32}^2P_3 &\le& \frac{a_{21}^2P_2}{1+a_{11}^2P_1+a_{31}^2P_3},
            \end{eqnarray}
            \end{subequations}
            then the capacity region is characterized by all the nonnegative rate pairs s.t.
            \begin{subequations}
            \label{eq:VSPhaseFadingICR Partial FB_region}
            \begin{eqnarray}
            R_1 &\le& \log_2\big(1+a_{11}^2P_1+a_{31}^2P_3\big)\\
            R_2 &\le& \log_2\big(1+a_{22}^2P_2+a_{32}^2P_3\big),
            \end{eqnarray}
            \end{subequations}
            and it is achieved with $X_k\sim \CN(0,P_k), k\in\{1,2,3\}$, mutually independent and with DF strategy at the relay.
\end{corollary}
\begin{proof}
The result follows from the expressions of Theorem \ref{thm:Partial FB VSI}. In order to obtain the conditions on the channel coefficients in \eqref{eq:VSPhaseFadingICR Partial FB CONDITIONS},
we evaluate $I(X_1;Y_1,Y_3|X_2,X_3,\utH)$ using mutually independent, zero-mean circularly symmetric complex Normal channel inputs. This leads to
\begin{eqnarray}
\label{eq:eq13}
& &\!\!\!\!\!\!\!\! I(X_1;Y_1,Y_3 | X_2,X_3,\utH)\nonumber\\
& & \quad =\E_{\utH}\bigg\{\log_2\Big(1+P_1(|H_{11}|^2+|H_{13}|^2)\Big)\bigg\}.
\end{eqnarray}
Next, note that under the phase fading model $|H_{11}|^2 = a_{11}^2$ and $|H_{13}|^2 = a_{13}^2$, thus \eqref{eq:eq13} can be rewritten as
\begin{eqnarray*}
 I(X_1;Y_1,Y_3|X_2,X_3,\utH)  =\log_2\big(1+P_1(a_{11}^2+a_{13}^2)\big).
\end{eqnarray*}
The rest of the expressions in Theorem \ref{thm:Partial FB VSI} have been already evaluated for the phase fading model in Corollary \ref{Cor:ICRF VSI Phase fading}.
\end{proof}

\subsubsection{Ergodic Rayleigh Fading}
Define $\tU_k \triangleq(U_{kk},U_{3k}), \hat{U}_k \triangleq(U_{1k},U_{2k},U_{3k}), k\in\{1,2\}$. If the channel is subject to ergodic Rayleigh fading, we obtain the following explicit result:
\begin{corollary}
            Consider the Rayleigh fading ICR with Rx-CSI and partial feedback s.t. $y_{1,1}^{i-1}$ and $\hvec_{1,1}^{i-1}$ are available at the relay at time $i$. If the channel coefficients satisfy
            \begin{subequations}
            \begin{eqnarray}
            &   & \!\!\!\!\!\!\!\!\!\!\!\!\!\!\E_{\tU_1}\Big\{\!\log_2\big(1\!+\!a_{11}^2|U_{11}|^2P_1\!+\!a_{31}^2|U_{31}|^2P_3\big)\Big\} \nonumber\\
            &  & \!\!\!\!\!\!\!\!\!\!\!\!\!\le     \E_{U_{11},U_{13}}\Big\{\!\log_2\big(1\!+\!(a_{11}^2|U_{11}|^2\!+\!a_{13}^2|U_{13}|^2)P_1\big)\!\Big\}\\
            &   & \!\!\!\!\!\!\!\!\!\!\!\!\!\!\E_{\tU_1}\Big\{\log_2\big(1+a_{11}^2|U_{11}|^2P_1+a_{31}^2|U_{31}|^2P_3\big)\Big\} \nonumber\\
            &  & \!\!\!\!\!\!\!\!\!\!\!\!\!\le   \E_{\hat{U}_2} \!\!\left\{\! \log_2\!\left( \!1\! +\! \frac{ a_{12}^2 |U_{12}|^2 P_1 }
            {1\! +\!  a_{22}^2|U_{22}|^2P_2\! +\! a_{32}^2 |U_{32}|^2 P_3}  \right) \!\right\}\\
            &   & \!\!\!\!\!\!\!\!\!\!\!\!\!\!\E_{\tU_2}\!\Big\{\!\log_2\big(1\!+\!a_{22}^2|U_{22}|^2P_2\!+\!a_{32}^2|U_{32}|^2P_3\big)\!\Big\} \nonumber\\
            &  &\!\!\!\!\!\!\!\!\!\!\!\!\! \le  \E_{\hat{U}_1}\!\! \left\{\! \log_2\!\left(\! 1\! +\! \frac{ a_{21}^2 |U_{21}|^2 P_2 }
            {1\! +\!  a_{11}^2|U_{11}|^2P_1\! +\! a_{31}^2 |U_{31}|^2 P_3}  \right) \!\right\}
            \end{eqnarray}
            \end{subequations}
            then the capacity region is characterized by all the nonnegative rate pairs s.t.
            \begin{subequations}
            \label{eq:VSRayleighFadingICR Partial FB_region}
            \begin{eqnarray}
            & & \!\!\!\!\!\!\!\!\!\!\!\!\!\!\!\!\! R_1 \le \E_{\tU_1}\Big\{\!\log_2\big(1\!+\!a_{11}^2|U_{11}|^2P_1\!+\!a_{31}^2|U_{31}|^2P_3\big)\!\Big\}\\
            & & \!\!\!\!\!\!\!\!\!\!\!\!\!\!\!\!\! R_2 \le \E_{\tU_2}\Big\{\!\log_2\big(1\!+\!a_{22}^2|U_{22}|^2P_2\!+\!a_{32}^2|U_{32}|^2P_3\big)\!\Big\},
            \end{eqnarray}
            \end{subequations}
            and it is achieved with $X_k\sim \CN(0,P_k), k\in\{1,2,3\}$, mutually independent and with DF strategy at the relay.
\end{corollary}
\begin{proof}
The proof follows similar arguments to those used in the proof of Corollary \ref{Cor:Partial FB VSI phase fading}.
\end{proof}

\subsubsection{Comments}
\begin{remark}
\em{}Comparing to Theorem \ref{thm:ICRF VS}, we observe that in Theorem \ref{thm:Partial FB VSI} there is an additional condition in \eqref{eq:ICR partial FB VSI Con}.
This is due to the fact that with partial feedback at the relay, the cut-set bound cannot be achieved at the destination without guaranteeing reliable decoding at the relay.
\end{remark}

\begin{remark}
\em{} For the configuration described in Theorem \ref{thm:Partial FB VSI}, then from Proposition \ref{thm:feedback doesnt improve capacity} it is clear that adding a noiseless feedback link,
from Rx$_1$ to Tx$_2$ (partial Rx-opposite Tx feedback) does not enlarge the capacity region in the VSI regime.
\end{remark}
\begin{remark}
\em{} Consider the configuration described in Theorem \ref{thm:Partial FB VSI} with an additional  noiseless feedback link from Rx$_1$ to Tx$_1$ (partial Rx-corresponding Tx feedback).
Then, following the same arguments as in Proposition \ref{prop:capacity_increased_VSI_fb_corrs_Txs}, we conclude that if
        \begin{subequations}
        \label{eq:ICR partial FB VSI Con2}
        \begin{eqnarray}
        \!\!\!\!\!\!\!\!\!\!I(X_1,X_3;Y_1|X_2,\tH_1) &\!\le\!& \min\big\{I(X_1;Y_2|\tH_2),\nonumber\\
                &\!\! & I(X_1;Y_1,Y_3|X_2,X_3,\utH)\!\big\}\\
        \!\!\!\!\!\!\!\!\!\!I(X_1,X_2,X_3;Y_2|\tH_2) &\!\le\!& I(X_2;Y_1|\tH_1),
        \end{eqnarray}
        \end{subequations}
hold,
then $(R_1,R_2)=\big(0,I(X_1,X_2,X_3;Y_2|\tilde{H}_2)\big)$ is achievable.
Note that \eqref{eq:ICR partial FB VSI Con2} guarantees \eqref{eq:ICR partial FB VSI Con}. Hence, when partial feedback (only from Rx$_1$) is available at the relay,
then an additional feedback link from Rx$_1$ to Tx$_1$ increases the capacity region in the VSI regime.
Figure \ref{fig:RxTx partial benefit} demonstrates the corresponding capacity region.
\begin{figure}[h]
    \centering
    \includegraphics[scale=0.60]{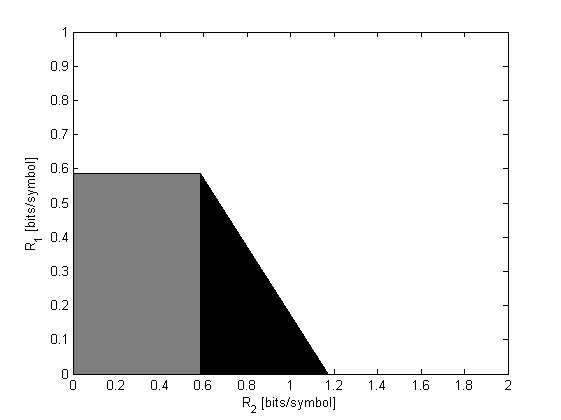}
\caption{The capacity region of the ICRF in the VSI regime with Rx$_1$-relay feedback and without Rx-Tx feedback (the gray area) and the achievable region with
Rx$_1$-relay and Rx$_1$-Tx$_1$ feedback (gray and black area), when the channel is subject to phase fading and $P_1=P_2=P_3=10, a_{11}=0.2, a_{12}=0.27, a_{21}=0.44 ,a_{22}=0.2, a_{31}=0.1,a_{32}=0.1$ and $a_{13}=0.3$.}
\label{fig:RxTx partial benefit}
\end{figure}
\end{remark}

\subsection{Partial Feedback in the Strong Interference Regime}
In this section, we characterize the capacity region of the ICR with partial feedback in the strong interference regime. We consider the case in which feedback is available only from Rx$_1$.
\begin{theorem}
\label{thm:ICR Partial FB SI}
                For the scenario of Theorem \ref{thm:Partial FB VSI}, if the channel coefficients satisfy
                \begin{subequations}
                \label{eq:Partial FB SI Conditions}
                \begin{eqnarray}
        \!\!\!\!\!\!\!\!\!  \!\!      I(X_1,X_3;Y_2|X_2,\tH_2) &\!\le\!& I(X_1;Y_1,Y_3|X_2,X_3,\utH)\\
        \!\!\!\!\!\!\!\!\!    \!\!    I(X_2,X_3;Y_1|X_1,\tH_1) &\!\le\!& I(X_2;Y_1,Y_3|X_1,X_3,\utH)\\
        \!\!\!\!\!\!\!\!\!      \!\!  I(X_1,X_2,X_3;Y_2|\tH_2) &\!\le\!& I(X_1,X_2;Y_1,Y_3|X_3,\utH)\\
        \!\!\!\!\!\!\!\!\!    \!\!    I(X_1,X_3;Y_1|X_2,\tH_1) &\!\le\!& I(X_1;Y_2|X_2,\tH_2)\\
        \!\!\!\!\!\!\!\!\!      \!\!  I(X_2,X_3;Y_2|X_1,\tH_2) &\!\le\!& I(X_2;Y_1|X_1,\tH_1),
                \end{eqnarray}
                \end{subequations}
                where the mutual information expressions are evaluated with $X_k\sim \CN(0,P_k), k\in\{1,2,3\}$, mutually independent,
                then the capacity region is given by all the nonnegative rate pairs s.t.
                \begin{subequations}
                \label{eq:Partial FB SI_region}
                \begin{eqnarray}
                \!\!\!\!R_1     &\le& I(X_1,X_3;Y_1|X_2,\tH_1)\\
                \!\!\!\!R_2     &\le& I(X_2,X_3;Y_2|X_1,\tH_2)\\
                \!\!\!\!R_1+R_2 &\le& \min\big\{I(X_1,X_2,X_3;Y_1|\tH_1),\nonumber\\
                & & \quad \quad \qquad I(X_1,X_2,X_3;Y_2|\tH_2)\big\},
                \end{eqnarray}
                \end{subequations}
                and it is achieved with $X_k\sim \CN(0,P_k), k\in\{1,2,3\}$, mutually independent and with DF strategy at the relay.
\end{theorem}
\subsubsection{Proof of Theorem \ref{thm:ICR Partial FB SI}}
The proof consists of the following steps:
\begin{itemize}
    \item From the ICRF we obtain two component channels: the first is the MARC with feedback (MARCF) which is a MARC whose message destination is Rx$_1$ and the relay
        receives feedback only from Rx$_1$.     The MARCF is defined by equations \eqref{eqn:Rx_1_sig} and \eqref{eqn:Relay_sig}.
        The second is the partially enhanced MARC (PEMARC) defined as a MARC whose message destination is Rx$_2$, while the relay receives feedback
        from Rx$_1$.
        The PEMARC is defined by equations \eqref{eqn:ICR_model} and $(y_{1,1}^{i-1},\hvec_{1,1}^{i-1})$ denotes the available feedback at the relay at
        time $i$ in both components.
        Note that contrary to Theorem \ref{thm:ICRF SI}, in the present case {\em the component channels are not symmetric}.
    \item We obtain the conditions on the channel coefficients s.t. the capacity of the MARCF and the PEMARC is achieved with DF at the relay.
            We show that for each component, capacity is achieved with zero-mean, circularly symmetric complex Normal channel inputs.
    \item We show that the same coding strategy at the sources and at the relay achieves capacity for both the MARCF and the PEMARC simultaneously.
    \item We therefore provide an achievable rate region for the ICR with partial feedback as the intersection of the capacity regions of the MARCF and the PEMARC.
    \item We show that when the conditions for SI are satisfied, the intersection of the capacity regions of the MARCF and the PEMARC contains the capacity region of the ICR with partial feedback.
    \item We conclude the capacity region of the ICR with partial feedback in the SI regime, to be the intersection of the capacity regions of the MARCF and the PEMARC.
    \item We explicitly characterize the SI conditions for the ICR with partial feedback.
\end{itemize}
\paragraph{An Achievable Rate Region}
Recall the achievable rate region for decoding at the destination as in Theorem \ref{thm:ICRF SI}.
Define  $\uavec$ as in section \ref{sec: Full FB SI Achievable region} and let $\mathcal{C}_{\scriptsize{\mbox{MARCF}}}(\uavec)$ and
$\mathcal{C}_{\scriptsize{\mbox{PEMARC}}}(\uavec)$ denote the capacity regions of the MARCF and the PEMARC, respectively. Moreover, let $t_1(R_1,R_2)$ and $t_2(R_1,R_2)$ denote the coding strategy for
the MARCF and the PEMARC, respectively.
From the derivation in Appendix \ref{App:MARC+feedback - capacity region} it follows that when only partial feedback is available at the relay,
 achievable rate regions for the MARCF and the PEMARC obtained with DF at the relay are given by
    \begin{subequations}
    \label{eq:eq22}
    \begin{eqnarray}
     & &   \!\!\!\!\!\!\!\!\!\! \R'_{\scriptsize{\mbox{MARCF/PEMARC}}}\nonumber\\
     & &  \!\!\!\! =  \bigg\{ (R_1, R_2)\in \Rset^2_+:\nonumber\\
     & & \qquad R_1 \le \min\big\{I(X_1;Y_1,Y_3|X_2,X_3,\utH),\nonumber\\
     & & \qquad\qquad \qquad I(X_1,X_3;Y_m|X_2,\tH_m)\big\}\\
     & &  \qquad R_2 \le \min\big\{I(X_2;Y_1,Y_3|X_2,X_3,\utH),\nonumber\\
     & & \qquad \qquad\qquad I(X_2,X_3;Y_m|X_1,\tH_m)\big\}\\
     & &  \qquad R_1+R_2 \le \min\big\{I(X_1,X_2;Y_1,Y_3|X_3,\utH),\nonumber\\
     & & \qquad \qquad\qquad I(X_1,X_2,X_3;Y_m|\tH_m)\big\} \bigg\},
    \end{eqnarray}
    \end{subequations}
where $m=1$ for the MARCF and $m=2$ for the PEMARC, and decoding at the relay at block $b$ is done using rule \eqref{eq:VSI Decoing at Relay Rule} without considering $\yvec_2(b)$,
and $X_k\sim\CN(0,P_k), k\in\{1,2,3\}$, mutually independent.
Let $t_{DF}(R_1,R_2)$ denote the coding scheme with mutually independent complex Normal inputs, achieving rates $R_1$ and $R_2$ in each component channel.
Following steps similar to Appendix \ref{app:maximzing dist for cut set bound}, we can show that
in order for $t_{DF}(R_1,R_2)$ to achieve the capacity of each component channel, we should guarantee that the relay decodes both messages reliably, i.e., we should guarantee that
    \begin{subequations}
    \label{eq:EMARC hardened conditions}
    \begin{eqnarray}
    \label{eq:eq21}
    & & \!\!\!\! \!\!\!\!\!\!\!\! \!\!\!\!  \max_{k\in\{1,2\}}\left\{I(X_1,X_3;Y_k|X_2,\tH_k)\right\}\nonumber\\
    & & \qquad \qquad \le I(X_1;Y_1,Y_3|X_2,X_3,\utH)\\
    \label{eq:eq26}
    & &\!\!\!\! \!\!\!\! \!\!\!\! \!\!\!\! \max_{k\in\{1,2\}}\left\{I(X_2,X_3;Y_k|X_1,\tH_k)\right\}\nonumber\\
     & &\qquad \qquad \le I(X_2;Y_1,Y_3|X_1,X_3,\utH)\\
    \label{eq:eq25}
    & & \!\!\!\! \!\!\!\! \!\!\!\! \!\!\!\! \max_{k\in\{1,2\}}\left\{I(X_1,X_2,X_3;Y_k|\tH_k)\right\}\nonumber\\
    & &\qquad \qquad \le I(X_1,X_2;Y_1,Y_3|X_3,\utH).
    \end{eqnarray}
    \end{subequations}
Note that if \eqref{eq:EMARC hardened conditions} holds then the capacity regions of the component channels are given by:
    \begin{subequations}
    \label{eq:eq23}
    \begin{eqnarray}
     & & \!\!\!\!\!\!\!\!\!\!\!\!\!\!\!\!    \mathcal{C}_{\scriptsize{\mbox{MARCF/PEMARC}}}(\uavec) \nonumber\\
     & & \!\!\!\!\!\!\!\!  =  \bigg\{ (R_1, R_2)\in \Rset^2_+:\nonumber\\
     & & \quad \quad   R_1 \le I(X_1,X_3;Y_m|X_2,\tH_m)\\
     & & \quad \quad  R_2 \le I(X_2,X_3;Y_m|X_1,\tH_m)\\
     & & \quad \quad  R_1+R_2 \le I(X_1,X_2,X_3;Y_m|\tH_m)\bigg\},
    \end{eqnarray}
    \end{subequations}
where $m=1$ for the MARCF and $m=2$ for the PEMARC, and they are achieved with $X_k\sim\CN(0,P_k), k\in\{1,2,3\}$, mutually independent and with DF strategy at the relay.

Note that when the conditions in \eqref{eq:EMARC hardened conditions} hold, by following similar steps as in the proof of Proposition \ref{pro:same strategy} we conclude that the same coding strategy
achieves capacity for both component channels simultaneously. Let $\R_{\scriptsize{\mbox{MARCF}}}\big(\uavec,t_1(R_1,R_2)\big)$ and $\R_{\scriptsize{\mbox{PEMARC}}}\big(\uavec,t_2(R_1,R_2)\big)$ denote
the achievable rate regions
of the MARCF and the PEMARC, respectively. Hence, when \eqref{eq:EMARC hardened conditions} holds and by choosing $t_1=t_2=t_{DF}$, any achievable rate pair
$(R_1,R_2)\in\R_{\scriptsize{\mbox{MARCF}}}\big(\uavec,t_{DF}(R_1,R_2)\big)\cap\R_{\scriptsize{\mbox{PEMARC}}}\big(\uavec,t_{DF}(R_1,R_2)\big)$ is also achievable in the ICR with partial
feedback. Thus, if \eqref{eq:EMARC hardened conditions} holds then an achievable rate region for the ICR with partial feedback is given by
\begin{eqnarray}
\label{eq:ICRF Partial SI Achievable}
    \R_{\scriptsize{ICRF}}(\uavec,t_{DF}) & = &
            \R_{\scriptsize{\mbox{MARCF}}}(\uavec,t_{DF})\cap\R_{\scriptsize{\mbox{PEMARC}}}(\uavec,t_{DF})\nonumber\\
    & = &
    \mathcal{C}_{\scriptsize{\mbox{MARCF}}}(\uavec) \cap \mathcal{C}_{\scriptsize{\mbox{PEMARC}}}(\uavec)\nonumber\\
    & \subseteq & \mathcal{C}_{\scriptsize{\mbox{ICRF}}}(\uavec) .
 \end{eqnarray}
\paragraph{Converse}
The proof of the converse follows similar arguments to those used in section \ref{sec: SI, CONVERSE full FB}. Note that in the SI regime, both receivers can decode both messages without
reducing the capacity region. Thus, any achievable rate pair $(R_1,R_2)\in\mathcal{C}_{\scriptsize{\mbox{ICRF}}}(\uavec)$ is also achievable in the component channels, i.e.,
$\mathcal{C}_{\scriptsize{\mbox{ICRF}}}(\uavec)\subseteq \mathcal{C}_{\scriptsize{\mbox{MARCF}}}(\uavec)\cap\mathcal{C}_{\scriptsize{\mbox{PEMARC}}}(\uavec)$. Thus, combined with
\eqref{eq:ICRF Partial SI Achievable} we conclude that in the SI regime
\begin{equation}
\label{eq:SI PF Capacity Region}
\mathcal{C}_{\scriptsize{\mbox{ICRF}}}(\uavec)= \mathcal{C}_{\scriptsize{\mbox{MARCF}}}(\uavec)\cap\mathcal{C}_{\scriptsize{\mbox{PEMARC}}}(\uavec).
\end{equation}
Recall that when
decoding at the relay does not constrain the rates, then
$\mathcal{C}_{\scriptsize{\mbox{MARCF}}}(\uavec)$ and $\mathcal{C}_{\scriptsize{\mbox{PEMARC}}}(\uavec)$ are given in \eqref{eq:eq23}. Next, we determine the SI conditions in the ICR with
partial feedback. Note that since $\mathcal{C}_{\scriptsize{\mbox{MARCF}}}(\uavec)$ and $\mathcal{C}_{\scriptsize{\mbox{PEMARC}}}(\uavec)$ in \eqref{eq:eq23} are achieved with
$t_{DF}(R_1,R_2)$, {for the rest of the proof we only consider mutually independent, circularly symmetric complex Normal channel inputs with zero mean}. Recall the converse proof in
\ref{sec: SI, CONVERSE full FB} and consider any rate pair $(R_1,R_2)\in\mathcal{C}_{\scriptsize{\mbox{ICRF}}}(\uavec)$. If Rx$_1$ can decode $m_1$ from the signal
\begin{equation*}
    \yvec_1=\hv_{11}\cdot\xvec_1+\hv_{21}\cdot\xvec_2+\hv_{31}\cdot\xvec_3+\zvec_1,
\end{equation*}
then it can create the signal
\begin{equation*}
    \yvec_1'=\hv_{21}\cdot\xvec_2+\hv_{31}\cdot\xvec_3+\zvec_1,
\end{equation*}
from which it can decode $m_2$ by treating $\hv_{31}\cdot\xvec_3$ as additive i.i.d. noise\footnote{Note that for this step we use the fact that the capacity-achieving codebooks are generated independently.} if
\begin{equation*}
    R_2\le I(X_2;Y_1'|\tH_1)\triangleq R_2',
\end{equation*}
and similarly Rx$_2$ can decode $m_1$ if
\begin{equation*}
    R_1\le I(X_1;Y_2'|\tH_2)\triangleq R_1'.
\end{equation*}
In order to guarantee that decoding both messages at each receiver does not reduce the capacity region we should require: $\max R_1\le R_1'$ and $\max R_2\le R_2'$. This is satisfied if
\begin{subequations}
\label{eq:Partial FB SI2}
\begin{eqnarray}
    & &\!\!\!\!\!\!\!\!\!\! \sup_{f(x_1,x_2,x_3)}\Big\{I(X_1,X_3;Y_1|X_2,\tH_1)\Big\} \!\le\! I(X_1;Y_2'|\tH_2)\nonumber\\
    & &\qquad \qquad \qquad \qquad \!\stackrel{(a)}{=}  I(X_1;Y_2|X_2,\tH_2)\\
    & &\!\!\!\!\!\!\!\!\!\!\sup_{f(x_1,x_2,x_3)}\Big\{I(X_2,X_3;Y_2|X_1,\tH_2)\Big\} \!\le\! I(X_2;Y_1'|\tH_1)\nonumber\\
    & & \qquad \qquad\qquad \qquad \!= \! I(X_2;Y_1|X_1,\tH_1).
\end{eqnarray}
\end{subequations}
Note that all the above mutual information expressions are evaluated using the same channel input distribution. Here, (a) follows from the fact that the l.h.s. of \eqref{eq:Partial FB SI2} is
maximized by $X_k\sim \CN(0,P_k), k\in\{1,2,3\}$, \em{mutually independent and independent of the channel coefficients}\em{.} Thus, when \eqref{eq:EMARC hardened conditions} holds,
the conditions for the strong interference are given by
    \begin{subequations}
    \label{eq:PartialFB SI Hardened conditons}
    \begin{eqnarray}
    \label{eq:eq24}
    I(X_1,X_3;Y_1|X_2,\tH_1) &\le& I(X_1;Y_2|X_2,\tH_2)\\
    \label{eq:eq20}
    I(X_2,X_3;Y_2|X_1,\tH_2) &\le& I(X_2;Y_1|X_1,\tH_1).
    \end{eqnarray}
    \end{subequations}
We conclude that if the conditions for reliable decoding at the relay \eqref{eq:EMARC hardened conditions} and the conditions for SI \eqref{eq:PartialFB SI Hardened conditons} are satisfied,
then the capacity region is given in \eqref{eq:SI PF Capacity Region} where the rate expressions for the component channels are given in \eqref{eq:eq23}.

\paragraph{Simplification of the Capacity Region}
Next, note that when \eqref{eq:PartialFB SI Hardened conditons} is satisfied, we obtain
\begin{eqnarray*}
    I(X_1,X_3;Y_1|X_2,\tH_1)&\le& I(X_1;Y_2|X_2,\tH_2)\\
    & \le & I(X_1,X_3;Y_2|X_2,\tH_2)\\
    I(X_2,X_3;Y_2|X_1,\tH_2)&\le& I(X_2;Y_1|X_1,\tH_1)\\
    & \le & I(X_2,X_3;Y_1|X_1,\tH_1),
\end{eqnarray*}
hence, \eqref{eq:SI PF Capacity Region} can be reduced to
\begin{subequations}
\label{eq:eq41}
\begin{eqnarray}
        \label{eq:eq411}
        R_1 &\le& I(X_1,X_3;Y_1|X_2,\tH_1)\\
        \label{eq:eq412}
        R_2 &\le& I(X_2,X_3;Y_2|X_1,\tH_2)\\
        \label{eq:eq413}
        R_1+R_2 &\le& \min\big\{I(X_1,X_2,X_3;Y_1|\tH_1),\nonumber\\
        & &\qquad \quad I(X_1,X_2,X_3;Y_2|\tH_2)\big\},
\end{eqnarray}
\end{subequations}
and \eqref{eq:EMARC hardened conditions} can be reduced to
    \begin{subequations}
    \label{eq:EMARC hardened conditions2}
    \begin{eqnarray}
    \label{eq:eq212}
    & &\!\!\!\!\!\!\!\!\!\! I(X_1,X_3;Y_2|X_2,\tH_2) \nonumber\\
     & &\qquad\qquad\le I(X_1;Y_1,Y_3|X_2,X_3,\utH)\\
    \label{eq:eq262}
    & & \!\!\!\!\!\!\!\!\!\!I(X_2,X_3;Y_1|X_1,\tH_1) \nonumber\\
    & &\qquad\qquad\le I(X_2;Y_1,Y_3|X_1,X_3,\utH)\\
    \label{eq:eq252}
    & & \!\!\!\!\!\!\!\!\!\!\max_{k\in\{1,2\}}\left\{I(X_1,X_2,X_3;Y_k|\tH_k)\right\}\nonumber\\
    & &\qquad\qquad\le I(X_1,X_2;Y_1,Y_3|X_3,\utH).
    \end{eqnarray}
    \end{subequations}
Next, note that when \eqref{eq:eq24} and \eqref{eq:eq212} hold, we obtain
    \begin{eqnarray*}
    & &\!\!\!\!\!\!\!\!\!\!\!\!\! I(X_1,X_2,X_3;Y_1|\tH_1)\\
    &=& I(X_2;Y_1|\tH_1)+I(X_1,X_3;Y_1|X_2,\tH_1)\\
    &\le&I(X_2;Y_1|\tH_1)+I(X_1,X_3;Y_2|X_2,\tH_2)\\
    &\le&I(X_2;Y_1,Y_3|X_3,\utH)+I(X_1;Y_1,Y_3|X_2,X_3,\utH)\\
    &=&I(X_1,X_2;Y_1,Y_3|X_3,\utH).
    \end{eqnarray*}
Therefore, if \eqref{eq:PartialFB SI Hardened conditons} holds,
\eqref{eq:eq252} reduces to
    \begin{eqnarray}
    \label{eq:eq2522}
    \!\!\!\!\!\!\!\!I(X_1,X_2,X_3;Y_2|\tH_2) &\le& I(X_1,X_2;Y_1,Y_3|X_3,\utH).
    \end{eqnarray}
Observe that in the SI regime with partial feedback, reliable decoding at the relay does not constrain the sum-rate in the MARCF. In the PEMARC, however,
\eqref{eq:eq2522} constrains the sum-rate to guarantee reliable decoding at the relay. Hence, decoding at the relay imposes an additional condition on the channel
coefficients in the PEMARC onto those required in the MARCF.
Note that \eqref{eq:eq41} gives \eqref{eq:Partial FB SI_region} and by combining \eqref{eq:PartialFB SI Hardened conditons} with \eqref{eq:eq212}, \eqref{eq:eq262} and \eqref{eq:eq2522},
we obtain the conditions in \eqref{eq:Partial FB SI Conditions}. This completes the proof.
\tend

\subsubsection{Ergodic Phase Fading}
Define $\tilde{\underline{\theta}}\triangleq(\theta_{11},\theta_{13},\theta_{21},\theta_{23})$. When the channel is subject to ergodic phase fading, we obtain the following explicit result:
\begin{corollary}
\label{cor:Partial FB Phase}
            Consider the phase fading ICR with Rx-CSI and partial feedback s.t. $y_{1,1}^{i-1}$ and $\hvec_{1,1}^{i-1}$ are available at the relay at time $i$. If the channel coefficients satisfy
            \begin{subequations}
            \label{eq:SIPhaseFadingICR Partial FB CONDITIONS}
            \begin{eqnarray}
            \label{eqn:part_cond1}
            a_{12}^2P_1+a_{32}^2P_3 &\le& (a_{11}^2+a_{13}^2)P_1\\
            a_{21}^2P_2+a_{31}^2P_3 &\le& (a_{21}^2+a_{23}^2)P_2\\
            \label{eqn:part_cond3}
            \!\!\!\!\!\!\!\!\!\!\!\!\log_2\big(1\!+\!a_{12}^2P_1\!+\!a_{22}^2P_2\!+\!a_{32}^2P_3\big) &\le&\nonumber\\
            & & \!\!\!\!\!\!\!\!\!\!\!\!\!\!\!\!\!\!\!\!\!\!\!\!\!\!\!\!\!\!\!\!\!\!\!\!\!\!\!\!\!\!\!\!\!\!\!\!\!\!\!\!\! \!\!\!\!\!\!\!\!\!\!\!\!\!\!\!\!\!\!\!\!\!\!\!\!\!\!\!\!\!\!\!\! \E_{\tilde{\underline{\theta}}}\bigg\{\log_2\Big(1+\big(\sum_{k=1}^2 P_k\cdot(a_{k1}^2+a_{k3}^2)\big)\nonumber\\
            & &\!\!\!\!\!\!\!\!\!\!\!\!\!\!\!\!\!\!\!\!\!\!\!\!\!\!\!\!\!\!\!\!\!\!\!\!\!\!\!\!\!\!\!\!\!\!\!\!\!\!\! \!\!\!\!+P_1P_2\cdot\big(a_{11}^2a_{23}^2+a_{13}^2a_{21}^2 \nonumber\\
            & &\!\!\!\!\!\!\!\!\!\!\!\!\!\!\!\!\!\!\!\!\!\!\!\!\!\!\!\!\!\!\!\!\!\!\!\!\!\!\!\!\!\!\!\!\!\!\!\!\!\!\! \!\!\!\! -2\cdot a_{13}a_{21}a_{11}a_{23}\cdot\cos(\theta_{13}+\theta_{21}\nonumber\\
            & &\!\!\!\!\!\!\!\!\!\!\!\!\!\!\!\!\!\!\!\!\!\!\! -\theta_{11}-\theta_{23})\big)\Big) \bigg\}\\
            a_{11}^2P_1+a_{31}^2P_3 &\le& \frac{a_{12}^2P_1}{1+a_{32}^2P_3}\\
            a_{22}^2P_2+a_{32}^2P_3 &\le& \frac{a_{21}^2P_2}{1+a_{31}^2P_3},
            \end{eqnarray}
            \end{subequations}
            then the capacity region is characterized by all the nonnegative rate pairs s.t.
            \begin{subequations}
            \label{eq:SIPhaseFadingICR Partial FB_region}
            \begin{eqnarray}
            R_1 &\le& \log_2\big(1+a_{11}^2P_1+a_{31}^2P_3\big)\\
            R_2 &\le& \log_2\big(1+a_{22}^2P_2+a_{32}^2P_3\big)\\
            R_1+R_2 &\le& \min\Big\{\log_2\big(1+a_{11}^2P_1+a_{21}^2P_2+a_{31}^2P_3\big),\nonumber\\
            & & \log_2\big(1+a_{12}^2P_1+a_{22}^2P_2+a_{32}^2P_3\big)\Big\},
            \end{eqnarray}
            \end{subequations}
            and it is achieved with $X_k\sim \CN(0,P_k), k\in\{1,2,3\}$, mutually independent and with DF strategy at the relay.
\end{corollary}
\begin{proof}
The result follows from the expressions of Theorem \ref{thm:ICR Partial FB SI}.
\end{proof}
\subsubsection{Ergodic Rayleigh Fading}
Define $\tU_k\triangleq(U_{kk},U_{3k}), \utU_k\triangleq(U_{1k},U_{2k},U_{3k}), k\in\{1,2\}$, and
$\hat{U}_1\triangleq(U_{12},U_{32})$, $\hat{U}_2\triangleq(U_{11},U_{13}),\hat{U}_3\triangleq(U_{21},U_{31})$,
$\hat{U}_4\triangleq(U_{21},U_{23})$,
$\hat{U}_5\triangleq(U_{11},U_{21},U_{13},U_{23})$. If the channel is subject to ergodic Rayleigh fading, then we obtain the following explicit result:
\begin{corollary}
            Consider the Rayleigh fading ICR with Rx-CSI and partial feedback s.t. $\yvec_{1}^{i-1}$ and $\hvec_{1,1}^{i-1}$ are available at the relay at time $i$. If the channel coefficients satisfy
            \begin{subequations}
            \begin{eqnarray}
            &   &\!\!\!\!\!\! \E_{\hat{U}_1}\Big\{\log_2\big(1+a_{12}^2|U_{12}|^2P_1+a_{32}^2|U_{32}|^2P_3\big)\Big\} \nonumber\\
            & & \le         \E_{\hat{U}_2}\Big\{\!\log_2\big(1+a_{11}^2|U_{11}|^2P_1+a_{13}^2|U_{13}|^2P_1\big)\!\Big\}\\
            &   &\!\!\!\!\!\! \E_{\hat{U}_3}\Big\{\log_2\big(1+a_{21}^2|U_{21}|^2P_2+a_{31}^2|U_{31}|^2P_3\big)\Big\} \nonumber\\
            &   &  \le  \E_{\hat{U}_4}\Big\{\!\log_2\big(1+a_{21}^2|U_{21}|^2P_2+a_{23}^2|U_{23}|^2P_2\big)\!\Big\}\\
            &   & \!\!\!\!\!\!\E_{\utU_2}\Big\{\log_2\big(1+a_{12}^2|U_{12}|^2P_1+a_{22}^2|U_{22}|^2P_2\nonumber\\
            & & \qquad\qquad \qquad\qquad \qquad  \qquad +a_{32}^2|U_{32}|^2P_3\big)\Big\} \nonumber\\
            &  &         \le \E_{\hat{U}_5}\Bigg\{\log_2\Bigg(1+\bigg(\sum_{k=1}^2 P_k\cdot(a_{k1}^2|U_{k1}|^2\nonumber\\
            & & \qquad \qquad \qquad\qquad\qquad\qquad +a_{k3}^2|U_{k3}|^2)\bigg)\nonumber\\
            &  &   + \big(P_1P_2\cdot(a_{11}^2|U_{11}|^2 a_{23}^2|U_{23}|^2 +a_{13}^2|U_{13}|^2a_{21}^2|U_{21}|^2 \nonumber\\
            & &\quad - 2\cdot \Real\{a_{13}U_{13}a_{21}U_{21}a_{11}U_{11}^*a_{23}U_{23}^*\})\big)\Bigg) \Bigg\}
            \end{eqnarray}
            and
            \begin{eqnarray}
            1+a_{11}^2P_1+a_{31}^2P_3 &\le& \frac{\frac{a_{12}^2P_1}{1+a_{32}^2P_3}}{e^{ \frac{1+ a_{32}^2P_3}{a_{12}^2P_1}}E_1\left(\frac{1+ a_{32}^2P_3}{ a_{12}^2P_1}\right)}\\
            1+a_{22}^2P_2+a_{32}^2P_3 &\le& \frac{\frac{a_{21}^2P_2}{1+a_{31}^2P_3}}{e^{ \frac{1+ a_{31}^2P_3}{a_{21}^2P_2}}E_1\left(\frac{1+a_{31}^2P_3}{a_{21}^2P_2}\right)},
            \end{eqnarray}
            \end{subequations}
            then the capacity region is characterized by all the nonnegative rate pairs s.t.
            \begin{subequations}
            \label{eq:SIRayleighFadingICR Partial FB_region}
            \begin{eqnarray}
            R_1 &\le& \E_{\tU_1}\Big\{\log_2\big(1+a_{11}^2|U_{11}|^2P_1\nonumber\\
            & & \qquad \qquad \qquad \quad +a_{31}^2|U_{31}|^2P_3\big)\Big\}\\
            R_2 &\le& \E_{\tU_2}\Big\{\log_2\big(1+a_{22}^2|U_{22}|^2P_2\nonumber\\
                & & \qquad \qquad \quad \qquad +a_{32}^2|U_{32}|^2P_3\big)\Big\}\\
            R_1+R_2 &\le& \min_{k\in\{1,2\}}\Big\{\E_{\utU_k}\big\{\log_2(1+a_{1k}^2|U_{1k}|^2P_1\nonumber\\
            & & \quad \!\!+a_{2k}^2|U_{2k}|^2P_2+a_{3k}^2|U_{3k}|^2P_3)\big\}\!\Big\}.
            \end{eqnarray}
            \end{subequations}
            and it is achieved with $X_k\sim \CN(0,P_k), k\in\{1,2,3\}$, mutually independent and with DF strategy at the relay.
\end{corollary}
\begin{proof}
The proof follows similar arguments to those used in the proof of Corollary \ref{cor:Partial FB Phase}.
\end{proof}
\begin{figure}[t]
  \centering
    {\includegraphics[scale=0.65]{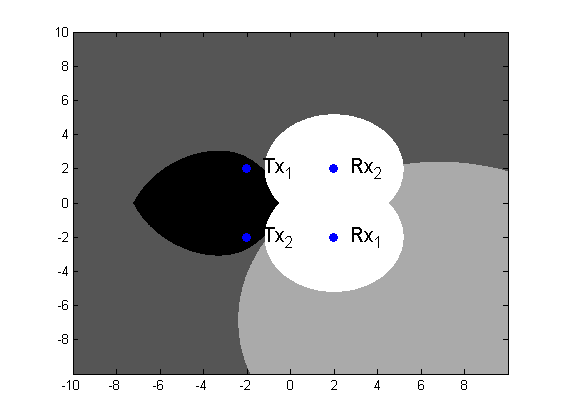}}
    \caption{The geographical position in the 2D-plane in which the VSI conditions hold for the ICR subject to
        phase fading. The black region shows the location of the relay for ICRs without feedback in which DF at the relay achieves capacity at VSI. The union of the black and dark gray regions shows the location of the
        relay for ICRFs with partial feedback, only from $\Rgood$ to the relay, in which DF at the relay achieves capacity at VSI.
        The union of the black, dark gray and light gray corresponds to ICRF with feedback from both receivers to the relay.
        The scenario parameters are detailed in Comment \ref{rem:fig:Commnet_on_VSI-FULLNO}.}
  \label{fig:fig5+6}
\end{figure}

\subsection{Comments}
\begin{remark}
\em{}
    Comparing partial feedback (Corollary \ref{cor:Partial FB Phase}) with no-feedback \cite[Theorem 2]{Dabora:101}, we note that the rate constraints at the destinations are the
    same. Also note that  $\avec \in \mD_1 \cap \mD_2$ guarantees \eqref{eqn:part_cond1}-\eqref{eqn:part_cond3}, hence as in Comment \ref{rem:Comparing_FB_with_NFB_VSI},
    feedback allows obtaining capacity characterization for a larger set of channel
    coefficients, and there are scenarios in which feedback strictly {\em enlarges} the capacity region (e.g., when $a_{23}$ is very small).

    Fig. \ref{fig:fig5+6} was created using the same parameters as those used for generating Fig. \ref{fig:VSI-FULLNO}, see Comment \ref{rem:fig:Commnet_on_VSI-FULLNO} for details.
    The figure demonstrates most clearly the benefits of combining feedback with relaying for interference management.
    Observe that without feedback, achieving capacity in the VSI regime requires the relay to be close to the transmitters, while partial feedback and moreover
    full feedback allow achieving capacity for a significantly larger geographical region.
\end{remark}
\begin{remark}
\em{}Note that if we assume a unidirectional noiseless cooperation link from one of the receivers to the other one, then the conditions of the SI and the VSI regime can not be satisfied.
Without loss of generality assume a noiseless cooperation link from Rx$_1$ to Rx$_2$, then the achievable rate at Rx$_2$ will always exceed the achievable rate at Rx$_1$. Hence decoding
the interference at Rx$_1$ will always decrease the capacity region. The same conclusion also holds for decoding both messages at Rx$_1$. Thus if there is a unidirectional noiseless
cooperation link from one of the receivers to the other one, then the SI and the VSI conditions can not be satisfied. Note that this conclusion does not hold for the scenarios where
the link between the receivers is noisy or if the receiver first compresses its channel observations prior to forwarding them to the other receiver (see \cite{Goldsmith:03} and \cite{Goldsmith:07}).
\end{remark}

\section{Conclusions}
\label{sec:conclusions}
In this paper we characterized the capacity region of the fading interference channel with a relay for different feedback configurations.
The capacity region was characterized explicitly for the phase fading and Rayleigh fading scenarios in both SI and VSI regimes. We showed that the
capacity is achieved with zero-mean, circularly symmetric complex Normal channel inputs, independent of each other with all transmitters using their
maximum available power. It was also shown that when feedback is available at the relay, then the best strategy for the relay in these regimes is to
decode both messages and forward them to both receivers and thus assist both receivers simultaneously. We showed that with such a strategy at the relay,
when VSI occurs the ICRF behaves like two parallel relay channels. We also showed that when SI occurs, the ICRF behaves like two EMARCs and the same coding
strategy achieves capacity for both simultaneously. We next showed that if feedback from both receivers is available at the relay, then additional feedback
from each receiver to its opposite transmitter provides no further improvement to the capacity region. However,  additional feedback links from
each receiver to its corresponding transmitter can enlarge the capacity region. By comparing the scenario where there is partial feedback (from one of the receivers only)
or full feedback (from both receivers) at the relay, versus the scenario with no feedback at all, we showed that partial feedback and moreover full feedback increase the
range of the channel coefficients which allow achieving capacity in both VSI and SI regimes significantly. With no feedback however, the relay reception must be good in order
to achieve capacity (Recall \cite[Theorem 6]{Kramer:05}: in the relay channel under the phase fading assumption, DF achieves capacity when the relay is closer to the source
than to the destination).

The fact that the capacity achieving channel inputs are mutually independent allows a relatively simple integration of relaying into existing wireless networks.
Also note that since the relay can be optimal for both communicating pairs simultaneously, then a relatively small number of relay stations can optimally assist several
nodes simultaneously. These observations  support the deployment of relay nodes to assist in managing interference in practical wireless systems such as cellular and WiFi networks.
We note however, that additional research, especially on non-fading scenarios is still required to obtain a complete assessment of cost-benefit tradeoff.

\section*{Acknowledgements}
We would like to thank the associate editor and the anonymous reviewers for their comments, which greatly improved the results of this manuscript.


\appendices
\setcounter{equation}{0}
\numberwithin{equation}{section}
\numberwithin{theorem}{section}
\numberwithin{proposition}{section}

\section{Maximizing Distribution for the Cut-Set Bound}
\label{app:maximzing dist for cut set bound}

\subsection{The Upper Bound on $R_1$}
\label{Bounds On R_1}
Starting with the upper bound on $R_1$, we first consider $I(X_1;Y_1,Y_2,Y_3 | X_2,X_3,\utH)$. Note that by fixing the side information $\utH=\underline{\tilde{h}}$, we obtain
    \begin{eqnarray}
    \label{eq:p1}
    & & \!\!\!\!\!\!\!\!\!\!\!\!\!\!\!h(Y_1,Y_2,Y_3 | X_2,X_3,\utH=\underline{\tilde{h}})\nonumber\\
    & & =h(h_{11}X_1 + Z_1,h_{12}X_1 + Z_2 , \nonumber\\
    & & \qquad \qquad h_{13}X_1+ Z_3 | X_2,X_3,\utH=\underline{\tilde{h}}).
    \end{eqnarray}
Define $\Hv\triangleq(H_{11},H_{12},H_{13})^T, \Zvec\triangleq(Z_1,Z_2,Z_3)^T$. Following these definitions we define ($\hY_1,\hY_2,\hY_3)^T \triangleq \Hv \cdot X_1 + \mathbf{Z}$.
For $\Xvec\triangleq(X_1,X_2,X_3)$, zero-mean, define $\C_1=\E\{\Xvec \Xvec^H\}$. Letting $\alpha_{ij}\triangleq \E[X_{i} X_{j}^*]$ we get
\begin{eqnarray}
\label{eq:c1 def}
\C_1 & \triangleq  & \E \left\{ \left[ \begin{array}{c}
X_1\\
X_2\\
X_3\end{array} \right] [X_1^*, X_2^*, X_3^*] \right\}\\
&  = & \E \left\{ \left[ \begin{array}{ccc}
X_1 X_1^* & X_1 X_2^* & X_1 X_3^*\\
X_2 X_1^* & X_2 X_2^* & X_2 X_3^*\\
X_3 X_1^* & X_3 X_2^* & X_3 X_3^* \end{array} \right] \right\} \nonumber \\
& = &\left[ \begin{array}{ccc}
P_1 & \alpha_{12} & \alpha_{13}\\
\alpha_{21} & P_2 & \alpha_{23}\\
\alpha_{31} & \alpha_{32} & P_3 \end{array} \right]\nonumber \\
 & \triangleq & \left[ \begin{array}{cc}
t_{11}    & \tvec_{21}^H\\
\tvec_{21} & \Tmat_{22} \end{array} \right],
\end{eqnarray}
where $t_{11} \triangleq P_1$ and $ \Tmat_{22}$ is $2 \times 2$ and p.d.. Next we obtain
\begin{eqnarray}
& &\!\!\!\!\!\!\!\!\!\!I(X_1;Y_1,Y_2,Y_3 | X_2,X_3,\utH)\nonumber\\
&=& \E_{\utH}\bigg\{h(Y_1,Y_2,Y_3 | X_2,X_3,\uth)\nonumber\\
& & \qquad \qquad \qquad \qquad -h(Y_1,Y_2,Y_3 | X_1,X_2,X_3,\uth)\bigg\}\nonumber\\
&\stackrel{(a)}{=}& \E_{\utH}\bigg\{h(\hY_1,\hY_2,\hY_3 | X_2,X_3,\uth)-h(Z_1,Z_2,Z_3)\bigg\}\nonumber\\
&\stackrel{(b)}{\le}&\E_{\utH}\bigg\{\log_2\Big((\pi e)^3 \det\big(\cov(\hY_1,\hY_2,\hY_3 | X_2,X_3,\uth)\big)\Big)\nonumber\\
& & \qquad \qquad \qquad \qquad\qquad \qquad  -\log_2(\pi e)^3\bigg\}\nonumber\\
&=& \E_{\utH}\bigg\{\log_2\Big(\det\big(\cov(\hY_1,\hY_2,\hY_3 | X_2,X_3,\uth)\big)\Big)\bigg\}\nonumber\\
&\stackrel{(c)}{=}&\E_{\utH}\bigg\{\log_2\Big(\det\big(\cov(\Hv \cdot X_1 + \mathbf{Z} | X_2,X_3,\uth)\big)\Big)\bigg\}\nonumber\\
&\stackrel{(d)}{=}& \E_{\utH}\bigg\{\log_2\Big(\det\big(\cov(\Hv \cdot X_1|X_2,X_3,\uth)\nonumber\\
& & \qquad \qquad \qquad \qquad \qquad + \cov(\mathbf{Z}|X_2,X_3,\uth) \big) \Big)\bigg\} \nonumber\\
&=& \E_{\utH}\bigg\{\!\log_2\Big(\!\det\big(\Imat+\Hv \cdot \cov(X_1|X_2,X_3) \cdot \Hv^H \big)\Big)\!\bigg\} \nonumber\\
&\stackrel{(e)}{=}& \label{eq:eq2} \E_{\utH}\bigg\{\!\log_2\Big(\!\det\big(\Imat\nonumber\\
& & \quad \qquad \qquad +\Hv \cdot (t_{11}-\tvec_{21}^H \Tmat_{22}^{-1}\tvec_{21}) \cdot \Hv^H\big)\!\Big)\!\bigg\}
\end{eqnarray}
\begin{eqnarray}
&\stackrel{(f)}{=}& \E_{\utH}\bigg\{\log_2\Big(\det\big(\Imat\nonumber\\
& & +\Hv \cdot t_{11}^{\frac{1}{2}}(1- t_{11}^{-\frac{1}{2}} \tvec_{21}^H \Tmat_{22}^{-\frac{1}{2}}
    \Tmat_{22}^{-\frac{1}{2}} \tvec_{21} t_{11}^{-\frac{1}{2}}) t_{11}^{\frac{1}{2}} \cdot \Hv^H\big)\Big)\bigg\} \nonumber \\
&\stackrel{(g)}{=}&\E_{\utH}\bigg\{\! \log_2\Big(\!\det\big(\Imat+\Hv \cdot t_{11}^{\frac{1}{2}} (1 - \mathds{A}\mathds{A}^H) t_{11}^{\frac{1}{2}} \cdot \Hv^H\big)\Big)\!\bigg\} \nonumber \\
&\stackrel{(h)}{\le}&\E_{\utH}\bigg\{\! \log_2\Big(\!\det\big(\Imat+\Hv \cdot t_{11}^{\frac{1}{2}} \cdot(1-\rho^2)\cdot t_{11}^{\frac{1}{2}} \cdot \Hv^H\big)\Big)\!\bigg \} \nonumber \\
&\stackrel{(i)}{\le}& \E_{\utH}\bigg\{\log_2\Big(\det\big(\Imat+\Hv\cdot t_{11}\cdot \Hv^H\big)\Big)\bigg\}.\nonumber
\end{eqnarray}
Thus, by using \cite[Fact 3.7.19]{MatrixMathematics}\footnote{Given $\xvec,\yvec\in\mathcal{C}^n$, define $\Amat\triangleq\Imat-\xvec\yvec^H$. Then, $\det(\Amat)=1-\xvec^H\yvec$.}, we get the following inequality:
\begin{eqnarray}
& &\!\!\!\!\!\!\!\!\!\!\!\!\!\!\!\!\!I(X_1;Y_1,Y_2,Y_3 | X_2,X_3,\utH)\le\nonumber\\
\label{eq:eq1}
& & \!\!\!\!\!\!\!\!\!\!\!\!\!\!\E_{\utH}\bigg\{\!\log_2\Big(1\!+\!P_1(|H_{11}|^2+ |H_{12}|^2+|H_{13}|^2)\Big)\!\bigg\}.
\end{eqnarray}
In the above transitions,
\begin{itemize}
\item (a) follows from \eqref{eq:p1}.
\item (b) follows from Lemma \ref{lem:lem0} which show that  $h(\hat{Y}_1,\hat{Y}_2,\hat{Y}_3| X_2,X_3)$ is maximized by jointly circularly symmetric complex Normal RVs with zero mean and same
covariance matrix as $(\hat{Y}_1,\hat{Y}_2,\hat{Y}_3,X_2,X_3)$.
\item (c) follows from the definition of $\hat{Y}_1,\hat{Y}_2,\hat{Y}_3,\Hv$ and $\mathbf{Z}$.
\item (d) follows from the fact that $\mathbf{Z}$ is independent of $X_k, k\in\{1,2,3\}$.
\item (e) follows from \cite[Sec. VI]{Valerie:81} which shows that the conditional covariance matrix of jointly complex normal RVs is given by the Schur complement
    of $\Tmat_{22}$ in the covariance matrix $\C_1$.
\item (f) follows from \cite[Proposition 8.1.2 and Lemma 8.2.1]{MatrixMathematics}\footnote{Given a p.d. matrix $\Amat$, $\Amat^{-1}$ is also a p.d.
    matrix and $\Amat$ can be written as $\Amat=\Bmat^2$, where $\Bmat$ is also a p.d. matrix.}.
\item (g) follows from the definition $\Amat\triangleq t_{11}^{-\frac{1}{2}} \tvec_{21}^H \Tmat_{22}^{-\frac{1}{2}}$.
\item (h) follows from \cite[Lemma 3.1]{Wang:05} and from \cite[Theorem 7.7.2 and Theorem 7.7.4]{Matrix:99}\footnote{\cite[Theorem 7.7.2]{Matrix:99}:
    Given two Hermitian matrices, $\Amat,\Bmat$, if $\Bmat\preceq\Amat$ then $\Tmat^H\Bmat \Tmat\preceq \Tmat^H\Amat \Tmat$. Thus,
        $\Imat+\Tmat^H\Bmat \Tmat\preceq \Imat+\Tmat^H\Amat \Tmat$.\\\cite[Theorem 7.7.4]{Matrix:99}: Given two matrices, $\Amat,\Bmat$, if $\Bmat\preceq\Amat$ then $\det(\Bmat)\le\det(\Amat)$.}.
\item (i) follows from the range of $\rho$, as given in \cite[Lemma 3.1]{Wang:05}: $\rho\in[0,1]$ and from \cite[Theorem 7.7.2 and Theorem 7.7.4]{Matrix:99} which
    state that since $(1-\rho^2)t_{11}\le t_{11}$ then $\det(\Imat+\Hv(1-\rho^2)t_{11}\Hv^H) \le \det(\Imat+\Hv t_{11}\hv^H)$.
\end{itemize}

Next, note that \eqref{eq:eq1} does not depend on $\alpha_{12}, \alpha_{13}$ or $\alpha_{23}$ and it can be achieved with equality from \eqref{eq:eq2} by setting
$\alpha_{12}=\alpha_{13}=0$, irrespective of the value of $\alpha_{23}$, which can also be set to zero. We also note that since $\log_2(x)$ monotonically increases with respect to $x$, then
$I(X_1;Y_1,Y_2,Y_3 | X_2,X_3,\utH)$ is maximized when Tx$_1$ transmits at its maximum available power.
Finally, (b) is achieved with equality if $X_1, X_2, X_3$ are jointly Normal. As these variables are uncorrelated, they are also independent.
In conclusion, $I(X_1;Y_1,Y_2,Y_3 | X_2,X_3,\utH)$ is maximized by
mutually independent, circularly symmetric complex Normal channel inputs with zero mean.

Next we consider $I(X_1,X_3;Y_1|X_2,\tH_1)$. Following the same arguments as the previous rate bound, note that
\begin{eqnarray}
\label{eq:p2}
& & \!\!\!\!\!\!\!\!\!\!\!\!\!\!\! h(Y_1 | X_2,\tH_1=\tilde{h}_1)\nonumber\\
& & \;\; =h(h_{11}X_1+ h_{31}X_3+ Z_1| X_2,\tH_1=\tilde{h}_1).
\end{eqnarray}
Define
\begin{subequations}
\label{eq:DF2}
\begin{eqnarray}
\label{eq:DF21}\hY_1 &\triangleq &H_{11}X_1+ H_{31}X_3+ Z_1,
\end{eqnarray}
\begin{eqnarray}
\label{eq:DF22}
\C_2& \triangleq & \E \left\{ \left[ \begin{array}{c}
\hY_1\\
X_2 \end{array} \right] \left[ \begin{array}{cc} \hY_1^*, X_2^*  \end{array} \right] \Bigg|\tH_1 \right\}\nonumber\\
& = & \E \left\{ \left[ \begin{array}{cc}
\hY_1 \hY_1^* & \hY_1 X_2^*\\
X_2 \hY_1^*   & X_2 X_2^*\end{array} \right]\Bigg|\tH_1\right\}\nonumber\\
&=&\left[\begin{array}{cc}
T_{11} & T_{12}\\
T_{12}^* & t_{22}\end{array} \right].
\end{eqnarray}
\end{subequations}
Hence, we bound $I(X_1,X_3;Y_1|X_2,\tH_1)$ as follows:
\begin{eqnarray}
& &\!\!\!\!\!\!\!\!\!\!\!\!\!\!\!I(X_1,X_3;Y_1|X_2,\tH_1)\nonumber\\
&=& \E_{\tH_1}\bigg\{h(Y_1 | X_2,\th_1)-h(Y_1| X_1,X_2,X_3,\th_1)\bigg\}\nonumber\\
&\stackrel{(a)}{\le}&\E_{\tH_1}\bigg\{\log_2\Big(\det\big(\cov(\hY_1 | X_2,\th_1)\big)\Big)\bigg\}\nonumber\\
&=&\label{eq:eq3}\E_{\tH_1}\bigg\{\log_2\Big(T_{11}-T_{12} t_{22}^{-1}T_{12}^*\Big)\bigg\}\\
&\stackrel{(b)}{\le}&\E_{\tH_1}\bigg\{\log_2\Big(T_{11}\Big)\bigg\}\nonumber.
\end{eqnarray}
Thus, we get
\begin{eqnarray}
\label{eq:eq4} & &\!\!\!\! I(X_1,X_3;Y_1|X_2,\tH_1)\nonumber\\
& & \le\E_{\tH_1}\Big\{\log_2(|H_{11}|^2 P_1+H_{11}H_{31}^*\alpha_{13}\nonumber\\
& & \qquad \qquad \qquad +H_{31}H_{11}^*\alpha_{31} +|H_{31}|^2 P_3 + 1)\Big\}.\phantom{xxxx}
\end{eqnarray}
Here, (a) follows from the definition of $\hY_1$ in \eqref{eq:DF21}, from \eqref{eq:p2} and from Lemma \ref{lem:lem0}, and (b) follows from the fact that $T_{12}^*T_{12}=|T_{12}|^2\ge 0$. Next, note that:
\begin{eqnarray*}
& & \!\!\!\!\E_{\tH_1}\Big\{\log_2(T_{11})\Big\}\\
& & =\int_{\tilde{h}_1\in\setC^3}\!\!\!\!\!\!\log_2\big(|h_{11}|^2 P_1 + h_{11}h_{31}^*\alpha_{13} +  h_{31}h_{11}^*\alpha_{31} \\
& & \qquad\qquad\qquad\qquad \quad + |h_{31}|^2 P_3 + 1\big) f_{\tH_1}(\tilde{h}_1) d\tilde{h}_1.
\end{eqnarray*}
Hence, if we replace $H_{31}$ with $-H_{31}$, the result of the integral remains unchanged. This follows as $h_{ij}$'s are independent complex RVs with uniform phases on $[0,2\pi)$, independent of their magnitudes,
and therefore the value of the integration is the same for all initial phases. Hence, following the same technique as in \cite[Theorem 8]{Kramer:05} and from concavity of the logarithm function, we can apply
Jensen's inequality and rewrite the bound in \eqref{eq:eq4} as:
\begin{eqnarray}
\label{eq:eq6}
& & \!\!\!\!\!\!\!\!\!\!\!\!\!\!\!\!\!\!\!I(X_1,X_3;Y_1|X_2,\tH_1) \le \E_{\tH_1}\Big\{\log_2(T_{11})\Big\}\nonumber\\
 & &\!\!\!\! \!\!\!\!\le  \E_{\tH_1}\Big\{\log_2(|H_{11}|^2P_1+|H_{31}|^2P_3+1)\Big\}.
\end{eqnarray}
As the bound in \eqref{eq:eq6} does not depend on $\alpha_{12}, \alpha_{13}$ or $\alpha_{23}$, then \eqref{eq:eq6} is achieved with equality from \eqref{eq:eq3} by setting $\alpha_{12}=\alpha_{13}=\alpha_{23}=0$.

We conclude that the upper bound on $R_1$ in cut-set bound is maximized by mutually independent, zero-mean, circularly symmetric complex Normal channel inputs, and with all sources transmitting at their maximum available power.

\subsection{The Upper Bound on $R_2$}
\label{Bounds On R_2}
Following steps similar to those in section \ref{Bounds On R_1}, we conclude that the mutual expressions $I(X_2;Y_1,Y_2,Y_3|X_1, X_3,\utH)$ and $I(X_2,X_3;Y_2|X_1,\tH_2)$ are
both maximized by mutually independent, zero mean, circularly symmetric complex Normal channel inputs.

\subsection{The Upper Bound on $R_1+R_2$}
\label{Bounds On R_1+R_2}
First, note that
\begin{eqnarray}
&&h(Y_1,Y_2,Y_3|X_3,\utH=\tilde{\underline{h}})=\nonumber\\
&&\qquad h\big(h_{11} X_{1} + h_{21} X_{2} +Z_1, h_{12} X_{1} + h_{22} X_{2} + Z_{2},\nonumber\\
& &\qquad\quad \quad  h_{13} X_{1} + h_{23} X_{2} + Z_{3}|X_3,\utH=\tilde{\underline{h}}).
\end{eqnarray}
Define $\Hmat, \mathbf{X}, \mathbf{Z}$ and rewrite $\C_1$ from \eqref{eq:c1 def} as
\begin{subequations}
\begin{eqnarray}
\Hmat &\triangleq &\left[ \begin{array}{cc}
H_{11} & H_{21}\\
H_{12} & H_{22}\\
H_{13} & H_{23}\end{array}\right] \\
\mathbf{X} &\triangleq &\left[ \begin{array}{c}
X_1\\
X_2\end{array} \right]\\
 \mathbf{Z} & \triangleq & \left[ \begin{array}{c}
Z_1\\
Z_2\\
Z_3\end{array} \right] \\
\C_1 & \triangleq & \left[ \begin{array}{cc}
\Tmat_{11} & \tvec_{12}\\
\tvec_{12}^H & t_{22} \end{array} \right],
\end{eqnarray}
\end{subequations}
where $\Tmat_{11}$ is $2\times2$ and p.d. and $ t_{22}\triangleq P_3$. Using the above definitions we obtain
\begin{eqnarray}
& & \!\!\!\!\!\!\!\!\!\!\!\!\!\!\!\!\!\!\!\! I(X_1,X_2;Y_1,Y_2,Y_3|X_3,\utH)\nonumber\\
 &=& \E_{\utH}\bigg\{ h(Y_1,Y_2,Y_3|X_3,\uth)\nonumber\\
 & & \qquad  \qquad-h(Y_1,Y_2,Y_3|X_1,X_2,X_3,\uth)\bigg\}\nonumber\\
&\le& \E_{\utH}\bigg\{\label{eq:eq9}\log_2\Big(\det\big(\Imat+\Hmat \cdot (\Tmat_{11} \nonumber\\
& & \qquad  \qquad \qquad-\tvec_{12}t_{22}^{-1}\tvec_{12}^H) \cdot \Hmat^H\big)\Big)\bigg\} \\
&\stackrel{(a)}{\le}&  \E_{\utH}\bigg\{ \log_2\Big(\det\big(\Imat+\Hmat\cdot\Tmat_{11}\cdot\Hmat^H\big)\Big)\bigg\}.
\end{eqnarray}
Thus, we bound $I(X_1,X_2;Y_1,Y_2,Y_3|X_3,\utH)$ as follows:
\begin{eqnarray}
\label{eq:eq7}
& & \!\!\!\!I(X_1,X_2;Y_1,Y_2,Y_3|X_3,\utH)\nonumber\\
& & \le \E_{\utH}\bigg\{\log_2\Big(1+\big(\sum_{i=1}^2 P_i(|H_{i1}|^2+|H_{i2}|^2+|H_{i3}|^2)\big)\nonumber\\
& & \qquad + \big((P_1P_2-|\alpha_{12}|^2)(|H_{11}|^2|H_{22}|^2+|H_{11}|^2|H_{23}|^2\nonumber\\
& &\qquad +|H_{12}|^2|H_{21}|^2 +|H_{12}|^2|H_{23}|^2  +|H_{13}|^2|H_{21}|^2\nonumber\\
& &\qquad +|H_{13}|^2|H_{22}|^2-2V_1-2V_2-2V_3)\big)\nonumber\\
& & \qquad +\alpha_{12}\big(H_{11}H_{21}^*+H_{12}H_{22}^*+H_{13}H_{23}^*\big)\nonumber\\
& & \qquad + \alpha_{21} \big(H_{21}H_{11}^*+ H_{22}H_{12}^*+H_{23}H_{13}^*\big)\Big)\bigg\},
\end{eqnarray}
where
\begin{subequations}
\label{eq:v def}
\begin{eqnarray}
V_1&=&\Real\{H_{11}H_{22}H_{12}^*H_{21}^*\}\\
V_2&=&\Real\{H_{12}H_{23}H_{13}^*H_{22}^*\}\\
V_3&=&\Real\{H_{13}H_{21}H_{11}^*H_{23}^*\}.
\end{eqnarray}
\end{subequations}
The transitions used in the above derivation are similar to those used in section \ref{Bounds On R_1}. Here, (a) follows from \cite[Lemma 3.1]{Wang:05} and from \cite[Theorem 7.7.2]{Matrix:99}.
Following the same technique as in \cite[Theorem 8]{Kramer:05} we obtain that replacing $H_{11},H_{12},H_{13}$ with $-H_{11},-H_{12},-H_{13}$, respectively, does not change the expected value in
\eqref{eq:eq7}. This is equivalent to replacing $\alpha_{12}$ and $\alpha_{21}$ with $-\alpha_{12}$ and $-\alpha_{21}$. The above observation can be used to express \eqref{eq:eq7} as:
\begin{eqnarray}
& &\!\!\!\!\!\!\!\!\!\! \E_{\utH}\bigg\{\log_2\Big(\det\big(\Imat+\Hmat\Tmat_{11}\Hmat^H\big)\Big)\bigg\} \nonumber\\
&\stackrel{(a)}{\le}& \E_{\utH}\bigg\{\log_2\Big(1+\big(\sum_{i=1}^2 P_i(|H_{i1}|^2+|H_{i2}|^2+|H_{i3}|^2)\big)\nonumber\\
& & \quad +\big((P_1P_2- 2|\alpha_{12}|^2)(|H_{11}|^2|H_{22}|^2\nonumber\\
& &\quad\quad\label{eq:eq8} +|H_{11}|^2|H_{23}|^2+|H_{12}|^2|H_{21}|^2\nonumber\\
& &\quad\quad  +|H_{12}|^2|H_{23}|^2 +|H_{13}|^2|H_{21}|^2\nonumber\\
& & \quad \quad     +|H_{13}|^2|H_{22}|^2-2V_1-2V_2-2V_3)\big)\!\Big) \!\bigg\},
\end{eqnarray}
where (a) follows from the concavity of the logarithm function. Next, we have the following proposition:
\begin{proposition}
\label{pro:R1+R2 INDIPENDENT SOURCES PROOF}
The expression in \eqref{eq:eq8} is maximized when $\alpha_{12}=0$ and when Tx$_1$ and Tx$_2$ transmit at their maximum available power,  i.e.,
\begin{eqnarray}
& &\!\!\!\!\E_{\utH}\bigg\{\log_2\Big(\det\big(\Imat+\Hmat\Tmat_{11}\Hmat^H\big)\Big)\bigg\}\le\nonumber\\
& &\;\label{eq:eq10} \E_{\utH}\bigg\{\log_2\Big(1+\big(\sum_{i=1}^2 P_i(|H_{i1}|^2+|H_{i2}|^2+|H_{i3}|^2)\big)\nonumber\\
& & \qquad \quad +P_1P_2\big(|H_{11}|^2|H_{22}|^2+|H_{11}|^2|H_{23}|^2 \nonumber\\
& & \qquad  \quad + |H_{12}|^2|H_{21}|^2 +|H_{12}|^2|H_{23}|^2 +|H_{13}|^2|H_{21}|^2\nonumber\\
& & \qquad  \quad  +|H_{13}|^2|H_{22}|^2 -2V_1-2V_2-2V_3\big)\Big)\! \bigg\},
\end{eqnarray}
where $V_1,V_2,V_3$ are defined in \eqref{eq:v def}.
\begin{proof}
\em{}
Consider $H_{11},H_{22},H_{12},H_{21} \in \mathcal{C}$ and $V_1,V_2,V_3$ as defined in \eqref{eq:v def}. Without loss of generality, we can write:
\begin{eqnarray*}
H_{11}H_{22}&=&A+Bi\\
H_{12}H_{21}&=&C+Di,
\end{eqnarray*}
$A,B,C,D\in\mathcal{R}$. Using the above definitions we get: $2V_1=2\cdot\Real\{H_{11}H_{22}H_{12}^*H_{21}^*\}=2(AC+BD)$. Also note that: $A^2+C^2\ge2AC$ for all $A,C \in \mathcal{R}$. Thus
\begin{eqnarray*}
    |H_{11}|^2|H_{22}|^2+|H_{12}|^2|H_{21}|^2 & \ge & 2\cdot\Real\{H_{11}H_{22}H_{12}^*H_{21}^*\}\\
    & = & 2V_1.
\end{eqnarray*}
Repeating this argument for $V_2$ and $V_3$, we conclude that the multiplier of $(P_1P_2-2|\alpha_{12}|^2)$ in \eqref{eq:eq8} is non-negative and omitting $2|\alpha_{12}|^2$ from
\eqref{eq:eq8} does not reduce the expected value. This leads to \eqref{eq:eq10}. Also note that since $\log(x)$ monotonically increases with respect to $x$, then \eqref{eq:eq10} is maximized when Tx$_1$ and Tx$_2$ transmit at their maximum available power.
\end{proof}
\end{proposition}

Hence, \eqref{eq:eq10} can be obtained with equality from \eqref{eq:eq9} by setting $\alpha_{12}=\alpha_{13}=\alpha_{23}=0$.
We conclude that the mutual information expressions in the cut-set bound are maximized by zero-mean, complex Normal channel inputs independent of each other, and with all sources transmitting at their maximum available power.

\setcounter{equation}{0}

\section{The Capacity Region of the EMARC}
\label{App:MARC+feedback - capacity region}
The capacity of the EMARC is stated in the following theorem:
\begin{theorem}
\label{thm: capacity of EMARC}
Consider the fading EMARC with Rx-CSI derived from the ICRF, given by equations \eqref{eqn:ICR_model}  where its message destination is Rx$_1$ but the relay
receives feedback from both receivers s.t. $y_{1,1}^{i-1},y_{2,1}^{i-1},\hvec_{1,1}^{i-1}$ and $\hvec_{2,1}^{i-1}$ are available at the relay at time $i$, prior to transmission.
Assume that the channel coefficients are independent in time and independent of each other s.t.
their phases are i.i.d. and distributed uniformly over $[0,2\pi)$. Let the additive noises be i.i.d. circularly symmetric complex Normal processes, $\CN(0,1)$, and let the sources
have power constraints $\E\big\{|X_k|^2\big\} \le P_k$, $k\in\{1,2,3\}$. The capacity region is then given by all nonnegative rate pairs s.t.
\vspace{-0.3 cm}
    \begin{subequations}
    \label{eq:MARC-FeedBack_Capacity}
    \begin{eqnarray}
        R_1 &\le& \min\{I(X_1;Y_1,Y_2,Y_3|X_2,X_3,\utH),\nonumber\\
        & & \qquad \qquad I(X_1,X_3;Y_1|X_2,\tH_1)\}\\
        R_2 &\le& \min\{I(X_2;Y_1,Y_2,Y_3|X_1,X_3,\utH),\nonumber\\
        & & \qquad  I(X_2,X_3;Y_1|X_1,\tH_1)\}\\
        R_1+R_2&\le& \min \{I(X_1,X_2;Y_1,Y_2,Y_3|X_3,\utH),\nonumber\\
        & & \qquad  I(X_1,X_2,X_3;Y_1|\tH_1)\},
    \end{eqnarray}
    \end{subequations}
    and it is achieved with $X_k\sim \CN(0,P_k), k\in\{1,2,3\}$, mutually independent and with DF strategy at the relay.
\end{theorem}
\begin{proof}
The proof consists of the following steps:
\begin{itemize}
    \item We provide an outer bound on the capacity region using the cut-set bound.
    \item We show that the input distribution that maximizes the outer bound is zero-mean, circularly symmetric complex Normal with channel inputs independent of each other and with maximum allowed power.
    \item Assuming codebooks generated according to the maximizing distribution, we present an achievable rate region using the DF strategy at the relay:
    \begin{itemize}
        \item For decoding at the relay, we follow steps similar to \cite[Sec. 4.D]{Kramer:05}.
        \item We provide an achievable rate region for decoding at the destination by considering a backward decoding scheme.
    \end{itemize}
    \item We conclude that the intersection of the achievable rate regions for decoding at the relay and at the destination coincides with the cut-set bound.
\end{itemize}
\subsection{An Outer Bound}
The following bounds are the cut-set bounds of the EMARC rate region:
\begin{subequations}
\label{eq:cut-set bound EMARC}
\begin{eqnarray}
\mathcal{S} &\triangleq& \{\mbox{Tx}_1\}, \mathcal{S}^C\triangleq \{\mbox{Tx}_2, \mbox{Relay}, \mbox{Rx}_1, \mbox{Rx}_2\} :\phantom{xxxxxx}\nonumber\\
& & \quad  R_1\le I(X_1;Y_1,Y_2,Y_3|X_2,X_3,\utH)\\
\mathcal{S} &\triangleq& \{\mbox{Tx}_1, \mbox{Relay}, \mbox{Rx}_2\},\mathcal{S}^C \triangleq \{\mbox{Tx}_2, \mbox{Rx}_1\} :\phantom{xxxxxx}\nonumber\\
& & \quad R_1\le I(X_1,X_3;Y_1|X_2,\tH_1)\\
\mathcal{S} &\triangleq& \{\mbox{Tx}_2\}, \mathcal{S}^C \triangleq\{\mbox{Tx}_1, \mbox{Relay}, \mbox{Rx}_1, \mbox{Rx}_2\} :\phantom{xxxxxx}\nonumber\\
& & \quad R_2\le I(X_2;Y_1,Y_2,Y_3|X_1,X_3,\utH)\\
\mathcal{S} &\triangleq& \{\mbox{Tx}_2, \mbox{Relay}, \mbox{Rx}_1\},\mathcal{S}^C \triangleq \{\mbox{Tx}_1, \mbox{Rx}_2\} :\phantom{xxxxxx} \nonumber\\
& & \quad R_2\le I(X_2,X_3;Y_1|X_1,\tH_2)\\
\mathcal{S} &\triangleq& \{\mbox{Tx}_1, \mbox{Tx}_2\}, \mathcal{S}^C \triangleq \{ \mbox{Rx}_1, \mbox{Rx}_2, \mbox{Relay}\} :\phantom{xxxxxx} \nonumber\\
& & \quad R_1+R_2\le I(X_1,X_2;Y_1,Y_2,Y_3|X_3,\utH)\\
\mathcal{S} &\triangleq& \{\mbox{Tx}_1, \mbox{Tx}_2, \mbox{Relay}, \mbox{Rx}_2\}, \mathcal{S}^C \triangleq \{\mbox{Rx}_1\} :\phantom{xxxxxx} \nonumber\\
& & \quad R_1+R_2\le I(X_1,X_2, X_3;Y_1|\utH).
\end{eqnarray}
\end{subequations}
Following the same arguments as in sections \ref{Bounds On R_1}, \ref{Bounds On R_2} and \ref{Bounds On R_1+R_2} we conclude that the outer bounds on $R_1$ and $R_2$ are maximized by mutually
independent, zero-mean, circularly symmetric complex Normal channel inputs and with all sources transmitting at their maximum available power. Moreover, we obtain an upper bound on $I(X_1,X_2,X_3;Y_1|\tH_1)$:
\begin{eqnarray}
\label{eq:eq12}
& & \!\!\!\!\!\!\!\!\!\!\!\! I(X_1,X_2,X_3;Y_1|\tH_1)\nonumber\\
& & \!\!\!\!\!\!\!\!\!\!\!\! \le\E_{\tH_1}\!\bigg\{\!\!\log_2\Big(\! 1\!+\!|H_{11}|^2P_1\!+\!|H_{21}|^2P_2\!+\!|H_{31}|^2P_3\Big)\!\!\bigg\},
\end{eqnarray}
and we conclude that it is achieved by mutually independent, zero mean, circularly symmetric complex Normal channel inputs.
\subsection{An Achievable Rate Region}
The achievability is based on DF strategy at the relay. Fix the block length $n$ and the input distributions: $f_{X_1}(x_1),f_{X_2}(x_2), f_{X_3}(x_3)$ where
$f_{X_k}(x_k)\sim\CN(0,P_k), k\in\{1,2,3\}$. The following encoding and decoding methods are considered:
\subsubsection{Code Construction}
The code construction is similar to section \ref{sec:ICRF Code Book}.
\subsubsection{Encoding at Block b}
The encoding process is similar to section \ref{sec:ICRF Encoding}.
\subsubsection{Decoding at the Relay at Block b}
The decoding process at the relay is similar to Section \ref{sec:ICRF relay decoding}, leading to the rate constraints \eqref{eq:VSI Decoing at Relay}.
\subsubsection{Decoding at the Destination}
The receiver uses a backward block decoding method. Assume that the receiver has successfully decoded $m_{1,b+1}$ and $m_{2,b+1}$. Then
\begin{itemize}
    \item In the first step the receiver generates the sets:
    \begin{eqnarray*}
     \!\!\!\!\!\!\!\!\mathcal{E}_{0,b}&\triangleq& \Big \{(m_1,m_2)\in\mathcal{M}_1\times \mathcal{M}_2: \big(\xvec_1(m_{1,b+1}),\\
     & &\!\!\!\!\!\!\!\!\!\!\!\!\!\!\!\!\!\!\!\! \xvec_2(m_{2,b+1}), \xvec_3(m_1,m_2) ,\mathbf{y}_1(b+1), \mathbf{\hvec}_1(b+1)\big)\!\! \in\!\! A_\epsilon^{(n)}\Big\}.\\
     \!\!\!\!\!\!\!\! \mathcal{E}_{1,b}&\triangleq& \Big\{(m_1,m_2)\in\mathcal{M}_1\times \mathcal{M}_2: \big(\xvec_1(m_1), \\
     & & \qquad \qquad \xvec_2(m_2),\mathbf{y}_1(b), \mathbf{\hvec}_1(b)\big) \in A_\epsilon^{(n)} \Big\}.
     \end{eqnarray*}
     \item The receiver then decodes $(m_{1,b},m_{2,b})$ by finding a unique pair $(m_1,m_2)\in\mathcal{E}_{0,b}\cap \mathcal{E}_{1,b}$.
\end{itemize}
Let the decoded pair be $(\hat{m}_{1,b},\hat{m}_{2,b})$. A decoding error happens if one of the following error events, associated with the decoding rule at the destination occurs:
\begin{itemize}
\item $E_{0,b} \cup E_{1,b}$, where
    \begin{eqnarray*}
    \!\!\!\!\!\!\!\!E_{0,b} &\triangleq& \Big\{\big(\xvec_1(m_{1,b+1}),\xvec_2(m_{2,b+1}), \xvec_3(m_{1,b},m_{2,b}) ,\\
    & & \qquad \qquad \qquad \mathbf{y}_1(b+1), \mathbf{\hvec}_1(b+1)\big) \notin A_\epsilon^{(n)}\Big\}\\
    \!\!\!\!\!\!\!\!E_{1,b} &\triangleq& \Big\{\big( \xvec_1(m_{1,b}), \xvec_2(m_{2,b}),\mathbf{y}_1(b), \mathbf{\hvec}_1(b)\big) \notin A_\epsilon^{(n)}\Big\}.
    \end{eqnarray*}
    From joint-typicality, the probability of the above event can be arbitrarily small if $n$ is large enough.

\item $E_{2,b}\triangleq\big\{(\hat{m}_{1,b}\neq m_{1,b},m_{2,b}) \in \mathcal{E}_{0,b} \cap \mathcal{E}_{1,b}\big\}$. From \cite[Lemma 2]{CoverElGamal:IT79}:
    \begin{eqnarray*}
      & &  \!\!\!\!\!\! \Pr\{(\hat{m}_{1,b}\neq m_{1,b},m_{2,b}) \in \mathcal{E}_{0,b}\}\\
       & & \qquad \qquad \qquad \qquad \le 2^{-n(I(X_3;Y_1|X_1,X_2,\tH_1)-7\epsilon)}\\
      & &  \!\!\!\!\!\!\Pr\{(\hat{m}_{1,b}\neq m_{1,b},m_{2,b}) \in \mathcal{E}_{1,b}\}\\
       & & \qquad \qquad \qquad \qquad \le        2^{-n(I(X_1;Y_1|X_2,\tH_1)-7\epsilon)}.
    \end{eqnarray*}
    Note that since the codebooks are constructed independent of each other then $\mathcal{E}_{0,b}$ is independent of $\mathcal{E}_{1,b}$ and the probability of $E_{2,b}$ can be arbitrarily small if $n$ is large enough and
    \begin{subequations}
    \label{eq:eq42}
    \begin{eqnarray}
    \label{eq:eq421}
    R_1 & \le & I(X_3;Y_1|X_1,X_2,\tH_1)+I(X_1;Y_1|X_2,\tH_1)\nonumber\\
    & = &I(X_1,X_3;Y_1|X_2,\tH_1).
    \end{eqnarray}

\item $E_{3,b}\triangleq\big\{(m_{1,b},\hat{m}_{2,b}\neq m_{2,b}) \in \mathcal{E}_{0,b}\cap \mathcal{E}_{1,b}\big\}$. From \cite[Lemma 2]{CoverElGamal:IT79},
the probability of this event can be arbitrarily small if $n$ is large enough and
    \begin{equation}
    \label{eq:eq422}
    R_2\le I(X_2,X_3;Y_1|X_1,\tH_1).
    \end{equation}

 \item $E_{4,b}\triangleq\big\{(\hat{m}_{1,b}\neq m_{1,b}, \hat{m}_{2,b}\neq m_{2,b})\in \mathcal{E}_{0,b}\cap\mathcal{E}_{1,b}\big\}$. From \cite[Lemma 2]{CoverElGamal:IT79}:
    \begin{eqnarray*}
        & &\!\!\!\!\! \Pr\{(\hat{m}_{1,b}\neq m_{1,b}, \hat{m}_{2,b}\neq m_{2,b}) \in \mathcal{E}_{0,b}\} \\
        & &  \qquad \qquad \qquad \qquad  \le 2^{-n\big(I(X_3;Y_1|X_1,X_2,\tH_1)-7\epsilon\big)}\\
        & &\!\!\!\!\!  \Pr\{(\hat{m}_{1,b}\neq m_{1,b}, \hat{m}_{2,b}\neq m_{2,b}) \in \mathcal{E}_{1,b}\} \\
        & & \qquad \qquad \qquad \qquad \le  2^{-n\big(I(X_1,X_2;Y_1|\tH_1)-7\epsilon\big)},
    \end{eqnarray*}
     the probability of this event can be arbitrarily small if $n$ is large enough and
     \begin{eqnarray}
     \!\!\!\!\!\!\!\!\!\!R_1+R_2 & \le & I(X_3;Y_1|X_1,X_2,\tH_1)+I(X_1,X_2;Y_1|\tH_1)\nonumber\\
        & = & I(X_1,X_2,X_3;Y_1|\tH_1).
     \end{eqnarray}
     \end{subequations}
\end{itemize}
Combining \eqref{eq:VSI Decoing at Relay} and \eqref{eq:eq42}, we obtain the achievable rate region of the EMARC:
    \begin{subequations}
    \label{eq:MARC+Feedback Region}
    \begin{eqnarray}
      & & \!\!\!\!\!\!\!\!\!\!\!\! \R_{\scriptsize{\mbox{EMARC}}}\nonumber\\
      & & \!\!\!\!\!\! =  \bigg\{ (R_1, R_2)\in \Rset^2_+:\nonumber\\
      & &  R_1 \le \min\big\{I(X_1;Y_1,Y_2,Y_3|X_2,X_3,\utH),\nonumber\\
      & & \qquad \qquad \qquad \quad\; I(X_1,X_3;Y_1|X_2,\tH_1)\big\}\\
      & &  R_2 \le \min\big\{I(X_2;Y_1,Y_2,Y_3|X_2,X_3,\utH),\nonumber\\
      & & \qquad \qquad \qquad \quad\; I(X_2,X_3;Y_1|X_1,\tH_1)\big\}\\
      & &  R_1+R_2 \le \min\big\{I(X_1,X_2;Y_1,Y_2,Y_3|X_3,\utH),\nonumber\\
      & & \qquad \qquad \qquad \quad\; I(X_1,X_2,X_3;Y_1|\tH_1)\big\} \bigg\}.
    \end{eqnarray}
    \end{subequations}
Finally, note that \eqref{eq:MARC+Feedback Region} coincides with the cut-set bound in \eqref{eq:cut-set bound EMARC} and thus it is the capacity region of the EMARC.
\end{proof}

\setcounter{equation}{0}

\section{Proof of Proposition \ref{prop:capacity_increased_VSI_fb_corrs_Txs}}
\label{Appndx:proof_prop_feedback_enlarges}
    \paragraph*{Inner bound}:
    Note that the region $\tmR_{ICRF}^{VSI}$ can be obtained by time-sharing between two rate points: point $A$ is the rate pair $(R_{1,A}, R_{2,A}) = \big(I(X_1,X_3;Y_1|X_2,\tH_1),
    I(X_2,X_3;Y_2|X_1,\tH_2)\big)$, and point $B$ is the rate pair $(R_{1,B}, R_{2,B}) = \big(I(X_1,X_2,X_3;Y_1|\tH_1),0\big)$.
     Theorem 1 shows that in the VSI regime  $(R_{1,A}, R_{2,A})$ is achievable. We next show that $(R_{1,B}, R_{2,B})$ is achievable.
     This is done using the DF-based achievability  scheme described in the following:

    \smallskip

%
%
    Fix the blocklength $n$ and the input distribution $f_{X_1,X_2,X_3} (x_1,x_2,x_3) = f_{X_1}(x_1)\cdot f_{X_2}(x_2)\cdot f_{X_3}(x_3)$ where $X_k\sim\CN(0,P_k)\mbox{ }k\in\{1,2,3\}$.
   Note that as $R_2=0$, then $\Tbad$ acts as a second relay for sending $m_1$ from $\Tgood$ to $\Rgood$. We use $nB$ channel symbols for sending $B-1$ messages.

    \paragraph{Codebook Construction}
    For each $m_1 \in \mathcal{M}_1$ and $k\in\{1,2,3\}$, select a codeword $\xvec_k(m_1)$ according to the p.d.f. $ f_{\Xvec_k}\big(\xvec_k(m_1)\big) =\prod_{i=1}^n f_{X_k}\big(x_{k,i}(m_1)\big)$.

    \paragraph{Encoding at Block $b$}
    At block $b$, Tx$_1$ transmits $m_{1,b}$ using $\xvec_1(m_{1,b})$. Let $\hat{m}_{1,b-1}, \hat{\hat{m}}_{1,b-1}$ denote
    the decoded $m_{1,b-1}$ at the end of block $b-1$, at the relay and at Tx$_2$, respectively.
    At block $b$, Tx$_2$ transmits $\xvec_2(\hat{\hat{m}}_{1,b-1})$ and the relay transmits $\xvec_3(\hat{m}_{1,b-1})$.
    At block $b=1$, Tx$_2$ transmits  $\xvec_2(1)$, and the relay transmits $\xvec_3(1)$, and at block $b=B$, Tx$_1$ transmits $\xvec_1(1)$.

    \paragraph{Decoding at the Relay and at Tx$_2$ at Block $b$}
    The relay and $\Tbad$ each uses a joint-typicality decoder.
    We now find conditions for bounding the average probability of error averaged over all codebooks.
    The decoder at the relay looks for a unique $m_1\in\mathcal{M}_1$ which satisfies
    \begin{eqnarray*}
        & &\Big(\xvec_1(m_1),\xvec_2(m_{1,b-1}),\xvec_3(m_{1,b-1}),\yvec_1(b),\yvec_2(b),\\
        & & \qquad \yvec_3(b),\mathbf{\uth}(b)\Big) \in A_\epsilon^{(n)}(X_1,X_2,X_3,Y_1,Y_2,Y_3,\utH).
    \end{eqnarray*}
    We conclude that the relay can decode $m_{1,b}$ with an arbitrarily small probability of error if $n$ is large enough and
    \begin{subequations}
    \label{eq:Tx2 and Relay Decoding}
    \begin{equation}
        R_1 < I(X_1;Y_1,Y_2,Y_3|X_2,X_3,\utH).
    \end{equation}
    Following the same approach, we can show that $\Tbad$ can decode $m_{1,b}$ reliably if $n$ is large enough and
    \begin{equation}
    \label{eqnprop:cond_dec_at_Tbad}
        R_1 < I(X_1;Y_2|X_2,X_3,\tH_2).
    \end{equation}
    \end{subequations}
    Note that as $I(X_1;Y_1,Y_2,Y_3|X_2,X_3,\utH) \ge I(X_1;Y_2|X_2,X_3,\tH_1)$, then reliable decoding at $\Tbad$ guarantees reliable decoding at the relay.

    \paragraph{Decoding at the Destination}
   $\Rgood$  uses a backward decoding scheme based on a joint-typicality rule. Assume that $\Rgood$ correctly decoded $m_{1,b+1}$. Then
    \begin{itemize}
        \item Rx$_1$ generates the sets:
        \begin{eqnarray*}
         \!\!\!\!\!\!\!\!\!\!\mathcal{E}_{0,b} &\triangleq& \Big \{m_1\in\mathcal{M}_1: \big(\xvec_1(m_{1,b+1}),\xvec_2(m_1), \xvec_3(m_1),\\
         & & \qquad \qquad \qquad \mathbf{y}_1(b+1), \mathbf{\hvec}_1(b+1)\big) \in A_\epsilon^{(n)}\Big\}\\
         \!\!\!\!\!\!\!\!\!\!\mathcal{E}_{1,b}  &\triangleq& \Big\{m_1\in\mathcal{M}_1: \big( \xvec_1(m_1),\mathbf{y}_1(b),\mathbf{\hvec}_1(b)\big) \in A_\epsilon^{(n)} \Big\}.
        \end{eqnarray*}
        \item $\Rgood$ then decodes $m_{1,b}$ by finding a unique $m_{1}\in \mathcal{E}_{0,b} \cap \mathcal{E}_{1,b}$.
    \end{itemize}
        Note that since the codewords are independent of each other, $\mathcal{E}_{0,b}$ is independent of $\mathcal{E}_{1,b}$ and the probability of decoding error
        can be made arbitrarily small by taking $n$ large enough as long as
        \begin{eqnarray}
        \label{eq:VSI con with feedback to own transmitter}
             R_1 & < & I(X_2,X_3;Y_1|X_1,\tH_1) + I(X_1;Y_1|\tH_1)\nonumber\\
              & = & I(X_1,X_2,X_3;Y_1|\tH_1).
        \end{eqnarray}

        Note that since \eqref{eq:VSI hardened con with feedback to own transmitter} holds  then
        \eqref{eq:VSI con with feedback to own transmitter} guarantees \eqref{eqnprop:cond_dec_at_Tbad}. Therefore reliable decoding at $\Rgood$ implies reliable decoding at $\Tbad$, and thus
        reliable decoding at the relay.
        Hence, when  \eqref{eq:VSI hardened con with feedback to own transmitter}  holds the rate pair
             $(R_{1,B}, R_{2,B})=\big(I(X_1,X_2,X_3;Y_1|\tH_1),0\big)$
        is achievable.

        By time sharing between point A and point B we conclude that  the inner bound $\tmR_{ICRF}^{VSI}$ is achievable.


        \bigskip

        \paragraph*{Outer bound}
        Consider the following three modifications to the ICR scenario defined in the proposition:
        (M1) We let each receiver observe the instantaneous channel output and Rx-CSI at the relay, $\big(y_{3,i}, \uth(i)\big)$, and at the other receiver;
        (M2) We also let each receiver send a feedback signal which consists of its channel output and Rx-CSI, to the opposite
        transmitter (in addition to the corresponding transmitter); and (M3) We let the relay send causal feedback of its channel output and Rx-CSI to both transmitters.
        Under these three assumptions, each receiver observes the same channel output at time $i$, $(y_{1,i},y_{2,i},y_{3,i}, \uth(i))$, and each transmitter observes at time $i$
        the feedback $\big(y_{1,i-1},y_{2,i-1},y_{3,i-1}, \uth(i-1)\big)$. Due to (M1) and the data processing inequality, the relay does not need to send any channel input, and we can set $X_3=0$. Equivalently,
        we may assume that the relay channel input is available non-causally at the receivers and therefore they can subtract it from their received signal prior to decoding and to sending feedback
        (as the receivers know at time $i$ $y_{3,i}$ and the CSI at the relay, and as they know the encoding function at the relay, they can generate at time $i$ $x_{3,i}$).
        The resulting scenario is therefore equivalent to the fading vector MAC with a MIMO receiver and causal feedback, of both the channel outputs and the Rx-CSI,  to both transmitters.
        Clearly, the capacity region of this channel constitutes an outer bound on $\tmC_{ICRF}$. In the following we show that this capacity region is given by
        $\tmC_{MAC-FB}$ defined in \eqref{eqn:cap_region_mac_fb}.

        To show this, we first derive an outer bound on the capacity region of the fading vector Gaussian MAC with feedback and Rx-CSI, denoted $\mC_{MAC-FB}$. This outer bound can be obtained from the
        cut-set bound \cite[Theorem 15.10.1]{cover-thomas:it-book} (see also \cite{Jafar:06}), and is given by
        \[
            \mC_{MAC-FB} \subseteq \bigcup_{f(x_1,x_2)}\mR_{MAC-FB}\big(f(x_1,x_2)\big),
        \]
        where
        \begin{subequations}
        \label{eqn:MAC-FB-outer-bound}
            \begin{eqnarray}
            & & \!\!\!\!\!\!\!\!\!\!\!\!\!\!\!\! \mR_{MAC-FB}\big(f(x_1,x_2)\big)\nonumber\\
            & & \!\!\!\!\!\!\!\!\!\!\!\! \triangleq   \bigg\{(R_1,R_2)\in \setR_+^2: \nonumber\\
            & & \!\!\!\!\!\!\!\!\!\!R_1  \le  I(X_1;Y_1, Y_2, Y_3|X_2,X_3=0,\utH)\\
            & & \!\!\!\!\!\!\!\!\!\!R_2  \le  I(X_2;Y_1,Y_2,Y_3|X_1,X_3=0,\utH)\\
            & & \!\!\!\!\!\!\!\!\!\!R_1 \!+\! R_2  \le  I(X_1,X_2;Y_1,Y_2,Y_3|X_3=0,\utH)\!\bigg\},
            \end{eqnarray}
        \end{subequations}
        where all mutual information expressions are evaluated with the specified input distribution $f(x_1,x_2)$.
        Repeating the arguments in Appendix \ref{app:maximzing dist for cut set bound} we conclude that the mutual information expressions in  Eqns. \eqref{eqn:MAC-FB-outer-bound} are simultaneously
        maximized by mutually independent channel inputs $X_k \sim \CN(0,P_k)$, $k=1,2$.
        Denote the corresponding input distributions $g(x_k)$, $k=1,2$. Thus $\mC_{MAC-FB} \subseteq \mR_{MAC-FB}\big(g(x_1)g(x_2)\big)$.

        It is straightforward to conclude that when $X_k \sim \CN(0,P_k)$, $k=1,2$, mutually independent, then any rate pair in $\mR_{MAC-FB}\big(g(x_1)g(x_2)\big)$ is achievable. Thus
    $\mC_{MAC-FB} \supseteq \mR_{MAC-FB}\big(g(x_1)g(x_2)\big)$. Combined with the outer bound we conclude that $\mC_{MAC-FB} = \mR_{MAC-FB}\big(g(x_1)g(x_2)\big)$.
    Lastly, we note that letting $X_3 \sim \CN(0,P_3)$ independent of $X_1$, $X_2$ does not change the rate expressions, thus $\mC_{MAC-FB} = \tmC_{MAC-FB}$.

    \smallskip

    Next, consider $\tmR_{OB}$. The derivation of the rate constraints in $\tmR_{OB}$ uses similar steps as in the derivation of the constraints in the cut-set bound,
    with the exception that the individual rate constraints are derived while letting
    both $Y_{1,1}^n$ and $Y_{2,2}^n$ be available at each receiver for decoding.
    This is necessary in order to accommodate the feedback, see Comment \ref{rem:Not_Possible_Cut-Set}.
    Similarly to Appendix \ref{app:maximzing dist for cut set bound}, it can be shown that mutually independent Gaussian inputs simultaneously maximize the
    mutual information expressions on the right-hand side of all constraints in \eqref{eqn:outer_bound_2}.
    Note that the sum-rate is maximized by mutually independent Gaussian inputs as a consequence of \cite[Theorem 8]{Kramer:05}.

    \smallskip

    This completes the proof of the outer bound  on 
    $\tmC_{ICRF}$.


\begin{biography}{Daniel Zahavi}
 received the B.Sc. degree in electrical engineering in 2009 from Technion, Israel Institute of Technology,
 Israel. He is currently working toward the M.Sc. degree at Ben-Gurion University of the Negev, Israel. Since 2010 he has been working as a communication
 research engineer with Signal Corps of Israel Defense Forces.

\end{biography}

\begin{biography}{Ron Dabora}
received his B.Sc. and M.Sc. degrees in 1994 and 2000, respectively,
from Tel-Aviv University and his Ph.D. degree in 2007 from Cornell
University, all in electrical engineering. From 1994 to 2000 he
worked as an engineer at the Ministry of Defense of Israel, and from
2000 to 2003, he was with the Algorithms Group at Millimetrix
Broadband Networks, Israel. From 2007 to 2009  he is a postdoctoral
researcher at the Department of Electrical Engineering at Stanford
University. Since 2009 he is an Assistant Professor at the
Department of Electrical and Computer Engineering, Ben-Gurion
University, Israel. He currently serves as an associate editor for the IEEE Signal Processing Letters.
\end{biography}


\begin{thebibliography}{10}

    \bibitem{Shannon:61}
    C.~E.~Shannon.
    \newblock{``Two-way communication channels"}.
    \newblock{\em Proceedings of the Fourth Berkeley Symposium on Mathematics, Statistics and Probability}, Jun.-Jul., 1960,
    vol. 1, pp. 611--644, University of California Press, 1961.

    \bibitem{Carleial:75}
    A. B. Carleial.
    \newblock{``A case where interference does not reduce capacity"}.
    \newblock{\em IEEE Trans. Inform. Theory}, vol. 21, no. 5, pp. 569--570, Sep. 1975.

    \bibitem{Sato:81}
    H. Sato.
    \newblock{``The capacity of the Gaussian interference channel under strong interference"}.
    \newblock{\em IEEE Trans. Inform. Theory}, vol. 27, no. 6, pp. 786--788, Nov. 1981.

    \bibitem{Ahlswede:74}
    R. Ahlswede.
    \newblock{``The capacity region of a channel with two senders and two receivers"}.
    \newblock{\em Ann. Probab.}, vol. 2, no. 5, pp. 805--814, 1974.

    \bibitem{venderMeulen:68}
    E. C. van der Meulen.
    \newblock{\em Transmission of information in a T-terminal discrete memoryless channel}.
    \newblock{Ph.D. dissertation, Department of Statistics}, University of California, Berkeley, CA, Jun. 1968.

    \bibitem{CoverElGamal:IT79}
    T. M. Cover and A. A. El Gamal.
    \newblock {``Capacity theorems for the relay channel"}.
    \newblock {\em IEEE Trans. Inform. Theory}, vol. 25, no. 5, pp. 572--584, Sep. 1979.

    \bibitem{Kramer:05}
    G. Kramer, M. Gastpar, and P. Gupta.
    \newblock{``Cooperative strategies and capacity theorems for relay networks"}.
    \newblock{\em IEEE Trans. Inform. Theory}, vol. 51, no. 9, pp. 3037--3063, Sep. 2005.

    \bibitem{Kramer:00}
    G. Kramer and A. J. van Wijngaarden.
    \newblock{``On the white Gaussian multiple-access relay channel"}.
    \newblock{\em Proceedings of the IEEE International Symposium on Information Theory (ISIT)}, Jun. 2000, Sorrento, Italy, pg. 40.

    \bibitem{Sankaranarayanan:04}
    L. Sankaranarayanan, G. Kramer, and N. B. Mandayam.
    \newblock{``Hierarchical sensor networks: Capacity theorems and cooperative
strategies using the multiple-access relay channel model".}
    \newblock{\em Proceedings of First IEEE Conference on Sensor and Ad Hoc
Communications and Networks}, Oct. 2004, Santa Clara, CA.

    \bibitem{Sankar:09}
    L. Sankar, N. B. Mandayam, and H. V. Poor.
    \newblock{``On the sum-capacity of the degraded Gaussian multiaccess relay channel".}
    \newblock{\em IEEE Trans. Inform. Theory}, vol. 55, no. 12, pp. 5394--5411, Dec. 2009.

    \bibitem{DelCoso:071}
    A. del Coso and C. Ibars.
    \newblock{``The amplify-based multiple-relay multiple-access channel: Capacity region and MAC-BC duality".}
    \newblock{\em Proc. IEEE Information Theory Workshop}, Jul. 2007, Bergen, Norway.

    \bibitem{Sahin:07}
    O. Sahin and E. Erkip.
    \newblock{``Achievable rates for the Gaussian interference relay channel"}.
    \newblock{\em Proceedings of the IEEE GLOBECOM Communications Theory Symposium}, Nov. 2007, Washington D.C., pp. 1627--1631.

    \bibitem{Sahin:09}
    O.~Sahin,~E. Erkip, and O. Simeone.
    \newblock{``Interference channel with a relay: models, relaying strategies, bounds"}.
    \newblock{\em Proceedings of the UCSD  Information Theory and Applications Workshop (ITA)}, Feb. 2009, San Diego, CA, pp. 90--95.

    \bibitem{Maric:08}
    I. Maric, R. Dabora, and A. Goldsmith.
    \newblock{``On the capacity of the interference channel with a relay"}.
    \newblock{\em Proceedings of the IEEE International Symposium on Information Theory (ISIT)}, Jul. 2008, Toronto, Canada, pp. 554--558.

    \bibitem{Maric:09}
    I. Maric, R. Dabora, and A. Goldsmith.
    \newblock{``An outer bound for the Gaussian interference channel with a relay"}.
    \newblock{\em Proceedings of the IEEE Information Theory Workshop (ITW)}, Oct. 2009, Taormina, Italy, pp. 569--573.

    \bibitem{Lasaulce:09_no1}
    B. Djeumou, E. V. Belmega, and S. Lasaulce,
    \newblock{``Interference relay channels --– part I: Transmission rates"}.
    \newblock Submitted to {\em IEEE Trans. Commun.}, 2009.
    \newblock{Available at http://arxiv.org/abs/0904.2585.}

    \bibitem{Dabora:101}
    R. Dabora.
    \newblock{``The capacity region of the interference channel with a relay in the strong interference regime subject to phase fading"}.
    \newblock{\em Proceedings of the IEEE Information Theory Workshop (ITW)}, Aug. 2010, Dublin, Ireland.

    \bibitem{Dabora:10}
    R. Dabora.
    \newblock{``The Capacity region of the fading interference channel with a relay in the strong interference regime"}.
    \newblock Submitted to {\em IEEE Trans. Inform. Theory}, 2010. Revised January 2012.
    \newblock{Available at http://www.bgu.ac.il/$\sim$daborona/.}

    \bibitem{Shannon:56}
    C.~E.~Shannon.
    \newblock{``The zero error capacity of a noisy channel".}
    \newblock{\em IRE Trans. Inform. Theory}, vol. 2, no. 3, pp. 8--19, Sep. 1956.

    \bibitem{Wolf:75}
    N. Gaarder and J. Wolf.
    \newblock{``The capacity region of a multiple-access discrete memoryless channel can increase with feedback".}
    \newblock{\em IRE Trans. Inform. Theory}, vol. 21, no. 1, pp. 100--102, Jan. 1975.

    \bibitem{Bross:09}
    S. Bross and M. Wigger.
    \newblock{``On the relay channel with receiver-–transmitter feedback".}
    \newblock{\em IEEE Trans. Inform. Theory}, vol. 55, no. 1, pp. 275--291, Jan. 2009.

    \bibitem{Kramer:09}
    J. Hou, R. Koetter, and G. Kramer.
    \newblock{``Rate regions for the multiple-access relay channel with relay-source feedback".}
    \newblock{\em Proceedings of the IEEE Information Theory Workshop (ITW)}, Oct. 2009, Taormina, Italy, pp. 288--292.

    \bibitem{Ho:08}
    C. K. Ho, K. T. Gowda, and S. Sun.
    \newblock{``Achievable rates for multiple access relay channel with generalized feedback".}
    \newblock{\em Proc. Int. Symposium on Inform. Theory and its Applications}, Dec. 2008, Auckland, New Zealand.

    \bibitem{Kramer:02}
    G. Kramer.
    \newblock{``Feedback strategies for white Gaussian interference networks".}
    \newblock{\em IEEE Trans. Inform. Theory}, vol. 48, no. 6, pp. 1423--1438, Jun. 2002.

    \bibitem{Sahai:09}
    A. Sahai, V. Aggarwal, M. Yuksel, and A. Sabharwal.
    \newblock{``On channel output feedback in deterministic interference channels".}
    \newblock{\em Proc. IEEE Information Theory Workshop (ITW)}, Oct. 2009, Taormina, Italy, pp. 298--302.

    \bibitem{Tse:09}
    C. Suh and D. N. C. Tse.
    \newblock{``Feedback capacity of the Gaussian interference channel to within 2 bits"}.
    \newblock{\em IEEE Trans. Inform. Theory}, vol. 57, no. 5, pp. 2667--2685, May 2011.

    \bibitem{Tuni:07}
    S. Yang and D. Tuninetti.
    \newblock{``Interference channel with generalized feedback (a.k.a. with source cooperation): part I: achievable region"}.
    \newblock{\em IEEE Trans. Inform. Theory}, vol. 57, no. 5, pp. 2686--2710, May 2011.

    \bibitem{BanNess:04}
    S.~Wu and Y. Bar-Ness.
    \newblock{``OFDM systems in the presence of phase noise: Consequences and solutions"}.
    \newblock{\em IEEE Trans. Commun.} vol. 52, no. 11, pp. 1988--1996, Nov. 2004.

    \bibitem{Sklar:97}
    B. Sklar.
    \newblock{``Rayleigh fading channels in mobile digital communication systems part I: Characterization"}.
    \newblock{\em IEEE Communications Magazine}, vol. 35, no. 7,  pp. 90--100, Jul. 1997.

    \bibitem{Xie:04}
     L.-L. Xie and P. R. Kumar.
    \newblock{``A network information theory for wireless communication: Scaling laws and optimal operation".}
    \newblock{\em IEEE Trans. Inform. Theory}, vol. 50, no. 5, pp. 748--767, May. 2004.

    \bibitem{Xie:05}
     L.-L. Xie and P. R. Kumar.
    \newblock{``An achievable rate for the multiple-level relay channel".}
    \newblock{\em IEEE Trans. Inform. Theory}, vol. 51, no. 4, pp. 1348--1358, Apr. 2005.

    \bibitem{Goldsmith:03}
    N. Jindal,  U. Mitra, A. Goldsmith.
    \newblock{``Capacity of Ad-Hoc networks with node cooperation"}.
    \newblock{\em Proceedings of the IEEE International Symposium on Information Theory (ISIT)}, Jun. 2004, Chicago, IL ,pg. 267.

    \bibitem{Goldsmith:07}
    Chris T. Ng, N. Jindal, A. Goldsmith, and U. Mitra.
    \newblock{``Capacity gain from two-transmitter and two-receiver cooperation".}
    \newblock{\em IEEE Trans. Inform. Theory}, vol. 53, no. 10, pp. 3822--3827, Oct. 2007.

    \bibitem{Massey:93}
    F. D. Neeser and J. L. Massey.
    \newblock{``Proper complex random processes with applications to information theory".}
    \newblock{\em IEEE Trans. Inform. Theory}, vol. 39, no. 4, pp. 1293--1302, Jul. 1993.

    \bibitem{Valerie:81}
    D. Ouellette.
    \newblock{``Schur complements and statistics".}
    \newblock{\em Linear Algebra and its Applications}, vol. 36, pp. 187--295, Mar. 1981.

    \bibitem{Abramowitz}
    M. Abramowitz and I. A. Stegun (Eds.).
    \newblock{\em Exponential integral and related functions}.
    \newblock{Handbook of Mathematical Functions with Formulas, Graphs, and Mathematical Tables}, 9th printing. New York: Dover, 1972.

    \bibitem{Matrix:99}
    R. A. Horn and G. R. Johnson.
    \newblock{\em {Matrix Analysis}}.
    \newblock Cambridge, U.K., Cambridge Univ. Press, 1999.

    \bibitem{MatrixMathematics}
    D. S. Bernstein.
    \newblock{\em{Matrix Mathematics}}.
    \newblock Prinston Univ. Press, 2009.

    \bibitem{cover-thomas:it-book}
    T.~M. Cover and J.~A.~Thomas.
    \newblock {\em {Elements of Information Theory}}.
    \newblock John Wiley and Sons Inc., 1991.

    \bibitem{liang:07}
    Y. Liang and V. V. Veeravalli.
    \newblock{``Cooperative relay broadcast channels"}.
    \newblock{\em IEEE Trans. Inform. Theory} vol. 53, no. 3, pp. 900--928, Mar. 2007.

    \bibitem{Sankar:04}
    L.~Sankaranarayanan, G. Kramer, and N.~B.~Mandayam.
    \newblock{``Capacity theorems for the multiple-access relay channel"}.
    \newblock{\em Proceedings of the Allerton Conference on Communications, Control and Computing}, Sep. 2004, Monticello, IL, pp. 1782--1791.

    \bibitem{Farhadi:08}
    G. Farhadi and N. C. Beaulieu.
    \newblock{``On the ergodic capacity of wireless relaying systems over Rayleigh fading channels"}.
    \newblock{\em IEEE Trans. Wireless Commun.}, vol. 7, no. 11, pp. 4462--4467, Nov. 2008.

    \bibitem{Farhadi:09}
    N. C. Beaulieu and G. Farhadi.
    \newblock{``On the ergodic capacity of multi-hop wireless relaying systems"}.
    \newblock{\em IEEE Trans. Wireless Commun.}, vol. 8, no. 5, pp. 2286--2291, May 2009.

    \bibitem{Kaya:08}
    C. Edemen and O. Kaya.
    \newblock{``Achievable rates for the three user cooperative multiple access channel"}.
    \newblock{\em Proceedings of the IEEE Wireless Communications and Networking Conference}, Las Vegas, NV, Mar. 2008, pp. 1507--1512 .

    \bibitem{Costa:87}
    M. H. M. Costa and A. A. El Gamal.
    \newblock{``The capacity region of the discrete memoryless interferenc channel with strong interference"}.
    \newblock{\em IEEE Trans. Inform. Theory}, vol. 33, no. 5, pp. 710--711, Sep. 1987.

    \bibitem{Wang:05}
    B. Wang, J. Zhang, and A. Host-Madsen.
    \newblock{``On the capacity of MIMO relay channels"}.
    \newblock{\em IEEE Trans. Inform. Theory}, vol. 51, no. 1, pp. 29--43, Jan. 2005.

    \bibitem{Carleial:78}
    A. B. Carleial.
    \newblock{``Interference channels"}.
    \newblock{\em IEEE Trans. Inform. Theory}, vol. 24, no. 1, pp. 60--70, Jan. 1978.

    \bibitem{Cover:75}
    T. Cover.
    \newblock{``An achievable rate region for the broadcast channel"}.
    \newblock{\em IEEE Trans. Inform. Theory}, vol. 21, no.4, pp. 399--404, Jul. 1975.

    \bibitem{Thomas:87}
    J. A. Thomas.
    \newblock{``Feedback can at most double Gaussian multiple access channel capacity".}
    \newblock{\em IEEE Trans. Inform. Theory}, vol. 33, no. 5, pp. 711--716, Sep. 1987.

    \bibitem{Jafar:06}
    S. A. Jafar and A. J. Goldsmith.
    \newblock{``On the capacity of the vector MAC with feedback``},
    \newblock{\em IEEE Trans. Inform. Theory}, vol. 52, no. 7, pp. 3259--3264, Jul. 2006.

    \bibitem{Cadambe:08}
    V. R. Cadambe and S. A. Jafar. ``Feedback improves the generalized degrees of freedom
    of the strong interference channel", CPCC Technical report, 2008. Available at
    {\tt http://escholarship.org/uc/item/02w92010}

    \bibitem{Stridharan:08}
    S. Sridharan, S. Vishwanath, S. A. Jafar, and S. Shamai.
    \newblock{``On the capacity of cognitive relay assisted Gaussian interference channel"}.
    \newblock{\em Proceedings of the IEEE International Symposium on Information Theory (ISIT)}, Jul. 2008, Toronto, Canada, pp. 549--553.

    \bibitem{Cadambe:09}
    V. R. Cadambe and S. A. Jafar.
    \newblock{``Degrees of freedom of wireless networks with relays, feedback, cooperation and full duplex operation"}.
    \newblock{\em IEEE Trans. Inform. Theory}, vol. 55, no. 5, pp. 2334--2344, May 2009.

    \bibitem{Aveshimer:11}
    A. Vahid, C. Suh, and A. S. Avestimehr.
    \newblock{``Interference channels with rate-limited feedback"}.
    Accpeted to the \newblock{\em IEEE Trans. Inform. Theory}, 2011.

    \bibitem{Gastpar:06}
    M. Gastpar and G. Kramer.
    \newblock{``On noisy feedback for interference channels"}.
     {\em Proc. Asilomar Conf. Signals, Syst., Comput.}, Oct. 2006.

    \bibitem{Yang:09}
    S. Yang and D. Tuninetti.
    \newblock{``A new sum-rate outer bound for Gaussian interference channels with generalized feedback"}.
    {\em Proc. IEEE International Symposium on Information Theory (ISIT)},  Jun. 2009, Seoul, Korea, pp. 2356-2360.

    \bibitem{Tandon:11}
    R. Tandon  and S.   Ulukus.
    \newblock{``Dependence balance based outer bounds for Gaussian networks with cooperation and feedback"}.
    {\em IEEE Trans. Inform. Theory}, vol. 57, no. 7,   pp. 4063--4086, Jul. 2011.

    \bibitem{Zheng:03}
    L. Zheng and D. N. C. Tse.
    \newblock{``Diversity and multiplexing: A fundamental tradeoff in multiple antenna channels".}
    \newblock{\em IEEE Trans. Inform. Theory}, vol. 49, no. 5, pp. 1073--1096, May 2003.

    \bibitem{TianYener:2010}
    Y. Tian and A. Yener.
    \newblock{``The ergodic fading interference channel with an on-and-off relay"}.
    \newblock{\em Proc. International Symposium on Information Theory (ISIT)}, Austin, TX, Jun. 2010, pp. 400--404.

    \bibitem{Liang:07}
    Y. Liang, V. V. Veeravalli, and H. V. Poor.
    \newblock{``Resource allocation for wireless fading relay channels: max-min solution"}.
    \newblock{\em IEEE Trans. Inform. Theory}, vol. 53, no. 10, pp. 3432--3453, Oct. 2007.

\end{thebibliography}
\end{document}